\documentclass[12pt,letterpaper]{article}
\pdfoutput=1 
\usepackage[top=3cm, bottom=3cm, left=2.5cm, right=2.5cm]{geometry}
\usepackage{amsmath, amssymb, amsthm, mathabx, mathrsfs, bbm}
\usepackage{graphicx, color}
\usepackage{array, multirow}
\usepackage[font=small,labelfont=bf]{caption} 
\usepackage{lipsum}
    \usepackage[T1]{fontenc}
    \usepackage[utf8]{inputenc} 
    \usepackage{lmodern}
\usepackage{amsfonts}
\usepackage{tabularx}
\usepackage{booktabs}
\usepackage{graphicx}
\usepackage{epstopdf}
\usepackage{algorithmic}
\usepackage{mathtools}
\usepackage{tikz}
\usepackage{todonotes}
\usepackage[normalem]{ulem}
\definecolor{darkgreen}{rgb}{0, .5, 0}
\usepackage[normalem]{ulem}
\ifpdf
\DeclareGraphicsExtensions{.eps,.pdf,.png,.jpg}
\else
\DeclareGraphicsExtensions{.eps}
\fi

\AtBeginDocument{%
\setlength{\oddsidemargin}{\dimexpr(\paperwidth-\textwidth)/2-1in}%
\setlength{\evensidemargin}{\oddsidemargin}%
\setlength{\topmargin}{%
\dimexpr(\paperheight-\textheight)/2-\headheight-\headsep-1in}%
}

\usepackage[shortlabels]{enumitem}

\usepackage{longtable} 

\usepackage[sort]{natbib}

\setcitestyle{numbers,square,comma}
 \usepackage{hyperref}
    \hypersetup{colorlinks=true,
    urlcolor=black,
    linkcolor=black,
    citecolor=black,
    bookmarksdepth=paragraph}
\usepackage[nameinlink]{cleveref}
\crefname{equation}{}{}
\crefname{section}{section}{sections}
\crefname{figure}{figure}{figures}
\crefname{table}{table}{tables}
\crefname{example}{example}{examples}
\crefname{proposition}{proposition}{propositions}
\Crefname{section}{Section}{Sections}
\Crefname{figure}{Figure}{Figures}
\Crefname{table}{Table}{Tables}
\Crefname{definition}{Definition}{Definitions}
\Crefname{theorem}{Theorem}{Theorems}
\Crefname{remark}{Remark}{Remarks}
\Crefname{example}{Example}{Examples}
\Crefname{proposition}{Proposition}{Propositions}
\numberwithin{equation}{section}

\newtheorem{theorem}{Theorem}[section]
\newtheorem{lemma}{Lemma}[section]
\newtheorem{corollary}{Corollary}[section]
\newtheorem{proposition}{Proposition}[section]
\theoremstyle{definition}
\newtheorem{definition}{Definition}[section]
\newtheorem{remark}{Remark}[section]
\newtheorem{assumption}{Assumption}[section]

    \newtheoremstyle{theoreminformal}
      {\topsep}   
      {\topsep}   
      {\itshape}  
      {}          
      {\bfseries} 
      {.}         
      {.5em}      
      {}          
    
    \theoremstyle{theoreminformal}

\newcommand{\ud}{\mathrm{d}}

\renewcommand{\ge}{	\geq }
\renewcommand{\le}{	\leq }
\renewcommand{\geq}{\geqslant}
\renewcommand{\leq}{\leqslant}

\renewcommand{\hat}{\widehat}

\renewcommand{\bar}{\overline}
\renewcommand{\epsilon}{\varepsilon}

\definecolor{amethyst}{rgb}{0.6, 0.4, 0.8}
\definecolor{huntergreen}{rgb}{0.21, 0.37, 0.23}
\definecolor{lavenderindigo}{rgb}{0.58, 0.34, 0.92}
\definecolor{lustred}{rgb}{0.9, 0.13, 0.13}
\definecolor{mediumpersianblack}{rgb}{0.0, 0.4, 0.65}
\definecolor{mediumseagreen}{rgb}{0.24, 0.7, 0.44}
\definecolor{mountainmeadow}{rgb}{0.19, 0.73, 0.56}
\definecolor{myrtle}{rgb}{0.13, 0.26, 0.12}
\definecolor{msugreen}{rgb}{0.09, 0.27, 0.23}

\usepackage{floatrow}
\newfloatcommand{capbtabbox}{table}[][\FBwidth]

\usepackage{blindtext}

\usepackage{multirow}
\usepackage[final]{changes}
\definechangesauthor[name={Anees Kazi}, color=black]{AK}

\usepackage{subcaption}
\usepackage{graphicx}
\title{When large trades are not (automatically) news: liquidity tail risk and price discovery}

\author{\fontsize{12}{12}\selectfont
  Umut \c{C}etin \thanks{The London School of Economics and Political Science.
  \texttt{u.cetin@lse.ac.uk}} \,\,
    Mingwei Lin\thanks{The London School of Economics and Political Science. \texttt{m.lin20@lse.ac.uk}} \,\,
  Giulia Livieri\thanks{The London School of Economics and Political Science. \texttt{g.livieri@lse.ac.uk}}
  \normalsize
}

\newcounter{termcounter}
\renewcommand{\thetermcounter}{\Roman{termcounter}}
\crefname{term}{term}{terms}
\creflabelformat{term}{#2\textup{(#1)}#3}

\makeatletter
\def\term{\@ifnextchar[\term@optarg\term@noarg}
\def\term@optarg[#1]#2{%
  \textup{#1}%
  \def\@currentlabel{#1}%
  \def\cref@currentlabel{[][2147483647][]#1}%
  \cref@label[term]{#2}}
\def\term@noarg#1{%
  \refstepcounter{termcounter}%
  \textup{(\thetermcounter)}%
  \cref@label[term]{#1}}
\makeatother

\RenewCommandCopy{\overbrace}{\LaTeXoverbrace}
\RenewCommandCopy{\underbrace}{\LaTeXunderbrace}

\usepackage{xparse}
\NewDocumentCommand{\Lout}{o}{%
  \mathtt{B}_{\text{out}}\IfValueT{#1}{^{#1}}%
}

\usepackage[most]{tcolorbox}
\usepackage{etoolbox} 
\usepackage{xcolor}
\definecolor{faintgray}{RGB}{245,245,245}     
\definecolor{faintborder}{RGB}{230,230,230}   
\definecolor{lightblack}{gray}{0.4}           
\newcounter{question}
\newtcolorbox[auto counter, use counter=question]{question}[1][]{
  enhanced,
  colback=faintgray,
  colframe=faintborder,
  boxrule=0.2pt,
  arc=2mm,
  title=\textcolor{lightblack}{\textbf{Question~\thequestion}},
  fonttitle=\bfseries,
  before upper={\centering\itshape},
  after title={\vspace{0.5ex}},
  boxsep=4pt,
  left=6pt,
  right=6pt,
  top=4pt,
  bottom=4pt,
  #1
}

\begin{document}

\footnotetext[1]{\textit{All authors contributed equally and are listed in alphabetical order.}}
\maketitle 
\begin{abstract}
We examine how heavy-tailed liquidity demand changes price discovery in a sequential limit order book with asymmetric information. In our setting, liquidity suppliers observe aggregate order flow, not its decomposition into informed demand and uninformed liquidity shocks. With heavy-tailed uninformed aggregated order flow, large trades remain plausibly uninformed over a wider range of depths, flattening price impact and slowing learning; sufficiently extreme trades can nevertheless become informative.  We characterize equilibrium through a non-linear fixed point equation for the marginal-cost schedule; heavy-tailed uninformed aggregated order flow invalidates the monotonicity and compactness arguments available under Gaussianity. Therefore, we establish fixed-point existence within a tail-controlled class, prove posterior consistency for liquidity suppliers in the presence of endogenous dependent order flow, and derive tail asymptotics for marginal costs, informed demand, and aggregate order flow. Additionally, we obtain eventual informed-demand dominance and eventual monotonicity of the book in the far tails. Empirically, using 10-level AAPL data, we document farther-out crossover diagnostics and persistent bid-ask spreads following large heavy-tailed trades.
\end{abstract}




\section{Introduction}
\label{section::Introduction}
Large trades move prices, but a large trade is not automatically news. A large buy order may reflect private information about the fundamental value of the asset, but it may also reflect a rare liquidity shock (for instance, linked to fund flows, margin pressure, or dealer inventory adjustment). In particular -- and crucially -- liquidity suppliers observe only the aggregate order flow that reaches the book, not the traders' motives. Our main interest is in understanding how a competitive limit order book (LOB) prices large trades when uninformed liquidity demand is heavy-tailed. The central mechanism is inference under liquidity tail risk: the heavier the tail of the uninformed liquidity demand, the wider the range of order sizes over which large trades remain plausibly liquidity-driven rather than information-driven.

We develop an (equilibrium) model of trade-market impact in a sequential trade framework, in which privately informed investors (i.e., informed traders or, simply, insiders) execute trades through a trading desk. Consistent with our interest, we will show that (what we call) the \emph{liquidity tail risk} of uninformed (i.e., noisy) liquidity demand greatly influences the pricing problem's mathematics and economics. In essence, we prove that when liquidity demand is thin-tailed, a very large order is strong evidence of private information. When liquidity demand is heavy-tailed, the same order may instead reflect a rare liquidity shock. Hence, large trades are less immediately informative, price impact at depth is more concave, and learning the fundamental value of the asset from the order flow is slower; i.e., liquidity tail risk is a state variable for price discovery in a LOB, and not merely a volatility parameter.

In our model specification, we allow for heavy-tailed uninformed order flow. We work with a Student-$t$ distribution with $\nu>2$ degrees of freedom (Remark \ref{rem::noise_dist}), which provides a parsimonious way to capture rare liquidity regimes through one interpretable parameter (that is, $\nu$). Interestingly, because of the statistical properties of the Student-$t$, the uninformed order flow can be viewed as a Gaussian order flow with a random latent liquidity variance; the parameter $\nu$ controls how often (rare) states of large liquidity demand occur. Importantly, the assumed uninformed order flow represents a significant departure from the equilibrium model proposed in \cite{ccetin2023power}, where noise trading is Gaussian and therefore very large uninformed orders are exponentially unlikely. Instead, we share with \cite{collin2016insider} the view that liquidity conditions are central to information revelation; but, importantly, we shift our attention from the stochastic volatility of uninformed trading volume to the tail risk of liquidity demand.

Our model is discrete-time. A single risky asset has a fundamental value $V$ fixed by nature at time zero and revealed at the terminal trading date. In each trading round, liquidity suppliers face adverse selection from myopic (Remark \ref{rem::remark_myopic}), risk-neutral insiders who submit market orders based on their private information; noise traders submit uninformed market orders. Each trading period should be interpreted as a short execution window and not as an individual message-level order. Consistent with \cite{ccetin2023power}, insiders execute after noise traders in order to reduce the price impact of their trades (Remark \ref{remark::priority}). In contrast to the auction-style pricing in \cite{kyle1985continuous}, we adopt (as said) a LOB structure as in
\cite{glosten1994electronic}, where orders are then executed against the standing book. We assume that limit orders are competitively priced to yield zero expected profit conditional on execution (again, as in \cite{glosten1994electronic}); e.g., the ask-side marginal price at a positive depth is the conditional expectation of $V$ given that aggregate order flow exceeds that depth. Moreover, in our model, informed demand is determined by a one-period first-order condition equating the informed trader's valuation to the marginal cost of execution. Interestingly, combining the insider optimality condition with competitive execution-based pricing yields a fixed-point equation for the marginal cost schedule.

Below, we mathematically characterize the (conditional) equilibrium for our model yielding -- to the best of the authors' knowledge -- new  economic channels with respect to the existing Gaussian-noise LOB models.

First, we find that the informativeness of large trades is endogenous to the tail thickness of liquidity demand. When uninformed order flow is heavy-tailed, large trades are not automatically news; they may be rare liquidity shocks. Therefore, our model predicts a larger ``crossover-depth" (the order size at which a trade becomes (more likely) information-driven than liquidity-driven), in the sense that heavier-tailed liquidity demand pushes the ``large trades are news" threshold farther out.

Second, we find (under endpoint regularity conditions for the distribution of $V$) that the marginal-cost schedule exhibits power-law-type scaling at depth. Interestingly, the exponent is explicitly determined by the liquidity tail risk, the number of informed traders, and liquidity suppliers' beliefs about $V$. In particular, we find that, because of heavy-tailed liquidity shocks, price impact is flatter and more concave at depth. Moreover, competition among informed traders attenuates impact and compresses spreads, but it does not eliminate the heavy-tail signature in large-order regions; in other words, the observed LOB reflects the balance between informed competition and liquidity-tail risk.

Third, we find that heavy-tailed liquidity demand slows price discovery because liquidity suppliers are more cautious in updating their beliefs, treating extreme order-flow observations as potential outliers. This reduces the informational impact of order-flow surprises compared to Gaussian models, causing adverse selection and wider spreads to persist longer. Using Bayesian updating methods for dependent data (e.g., \cite{shalizi2009dynamics}), we establish that beliefs ultimately converge under plausible conditions. However, this learning is, as in, e.g., \cite{ozsoylev2010price}, only asymptotic; this clarifies what kind of information aggregation our model does and does not imply.

Last, but not least, we show that the Student-$t$ assumption is not merely a routine distributional extension, yielding mathematical subtleties reflecting the economic mechanism of the paper. The polynomial decay of the Student-$t$'s tails, which implies that rare liquidity states remain relevant for execution-conditional expectations, breaks the monotonicity and compactness arguments, as well as the continuity of the  fixed-point operator used in the Gaussian benchmark (cf. \cite{ccetin2023power}). We therefore prove fixed-point existence on a tail-controlled compact class of candidate marginal-cost schedules, thus showing that heavy-tailed liquidity shocks do not make the marginal-cost problem ill posed. Moreover, we formulate the equilibrium results conditionally on the monotonicity of the selected pricing schedule, and we study -- analytically and numerically -- when this monotonicity is recovered.

We note that our theoretical fixed-point analysis assumes that $V$ has bounded support, and we solve the model numerically for unbounded specifications that are consistent with the key qualitative predictions of the model.

This article primarily contributes to the literature on LOB models

First, our model builds on the adverse-selection framework and Bayesian pricing logic of \cite{glosten1985bid}, as well as execution-based pricing principles for open LOBs (e.g., \cite{glosten1994electronic}). The most direct theoretical antecedent is \cite{ccetin2023power}, who derives power-law results in a static, one-period competitive LOB setting. We extend this framework to a dynamic, multi-period environment in which liquidity suppliers update their beliefs from observed aggregate order flow. See, also, \cite{parlour1998price, foucault1999orderflow, foucault2005limitbook, goettler2005equilibrium, rosu2009dynamic, cont2010stochastic} for (multi-period) equilibrium models of LOBs. Second, our framework also complements the foundational Kyle model \citep{kyle1985continuous} and the extensive literature on multi-period information revelation and strategic competition among informed traders, such as \cite{HoldenSubrahmanyam1992,BackCaoWillard2000}. Our findings provide implications for market participants and regulators concerned with information diffusion, trading costs, and liquidity provision. Third, a substantial empirical literature documents the prevalence of heavy tails in trading volume and order flow, as well as pronounced nonlinearities in market impact (e.g., \cite{farmer2004origin, Hasbrouck1991, bouchaud2018trades, gopikrishnan2000statistical, lillo2005theory, vaglica2008scalinglaws, cont2001empirical, mertens2022liquidity}. Our model provides an equilibrium mechanism that connects these empirical facts; the resulting comparative statics generate testable predictions linking the tail index of order flow, the extent of informed competition, and the dynamic evolution of
spreads and impact. The paper is also related to the finance literature on liquidity risk, liquidity fragility, and liquidity crises. For instance, \cite{pastor2003liquidity} establish liquidity as an economically important risk dimension, while \cite{cespa2014illiquidity,
brunnermeier2009market,morris2004coordination} study threshold-like liquidity
deterioration through coordination, funding constraints, costly participation, or cross-asset learning. Our mechanism is different. Liquidity fragility arises from inference inside the LOB: the tail of liquidity demand governs the crossover from ambiguous to information-dominated order flow.  Finally, our belief-updating framework uses tools from Bayesian learning with dependent data (e.g., \cite{shalizi2009dynamics}), which are particularly relevant in financial markets where order flow is endogenous and temporally dependent.

We proceed as follows. Section~\ref{section::Model_Setup} lays out the model.
Section~\ref{section::Equilibrium} derives the fixed-point characterization, explains why the Gaussian proof strategy does not extend directly to Student-\(t\) noise, and proves Student-\(t\) fixed-point existence on a tail-controlled compact class. Section \ref{section::The Liquidity_Suppliers_Belief_posterior_consistency} analyzes belief updating and establishes posterior consistency as the number of trading rounds tends to infinity. Section~\ref{section::Asymptotic_Price_Impact} derives asymptotic market impact and links tail exponents to the model's primitives and the history of learning. Section \ref{section::Numerical_Studies} provides numerical solutions and documents the quantitative implications of the model. Section \ref{sec::empirics} contains empirical work. Section~\ref{section::Conclusion} concludes. The appendix contains the technical proofs and additional figures.

\section{Model setup}\label{section::Model_Setup}
Hereafter, all random quantities are defined on a complete probability space $(\Omega, \mathcal{F}, \mathbb{P})$ for concreteness. 

We examine a discrete-time dynamic model of trading involving a single risky financial asset, with time periods indexed by $t = 0, 1, \ldots, T$. Each period can be interpreted as a short trading window during which a LOB is posted, trades execute, and the liquidity suppliers update their beliefs. We allow the time horizon $T$ to be either finite, in which case the fundamental value $V$ of the traded asset becomes public at the end of period $T$, or infinite, in which case the analysis focuses on the sequence of posterior beliefs generated by the order flow (Section \ref{section::The Liquidity_Suppliers_Belief_posterior_consistency}). We assume that $V$ is an integrable random variable. For the analytical results, we assume that $V$ has bounded support $[m,M]$, with $-\infty<m<M<\infty$, and admits a density. Additional assumptions below specify the regularity and endpoint behavior of this density. The numerical section considers
unbounded fundamentals as robustness checks.

Following \cite{ccetin2023power}, in each time period $t \in \{0, 1, \ldots, T\}$, there are four types of agents: (1) competitive (infinitely many) \emph{liquidity suppliers}, (2) $ N_t > 1$ risk-neutral \emph{informed traders}, (3) \emph{noise traders}, and (4) a \emph{trading desk}. At the beginning of every period, liquidity suppliers submit limit orders at competitive prices and form the LOB. Noise trades and informed traders send their orders to the trading desk, which functions as a broker and does not hold inventory. The informed traders are myopic (Remark \ref{rem::remark_myopic}) in their trading decisions, and we assume their trades arrive after noise traders send their orders to the desk (Remark \ref{remark::priority}).

We shall denote by $Y_t$ the sum of informed and noise trades submitted to the trading desk in period $t$. We shall set $Y_0=0$, and  the LOB in period $t$ is characterised by a non-decreasing (pricing) function (we note that we index the LOB by cumulative depth, not by price)
\begin{equation}\label{eq::LimitOrderBook}
    h(\,\cdot\,, t, Y^{t-1}) : \mathbb{R} \mapsto \mathbb{R},
\end{equation}
where $Y^{t}:=(Y_1,\ldots,Y_{t})$ is the history of $Y$ up to and including the period $t$. Thus, a buy market order of size $x>0$ traded against the LOB incurs a cost of
\begin{equation}\label{eq::CostOfAMarketOrder}
    \int_{0}^{x}h(u, t, Y^{t-1})\,\ud u;
\end{equation}
we use the sign convention that the integral in \eqref{eq::CostOfAMarketOrder} is negative for a sell market order $x<0$. Given this LOB, noise traders arrive first and submit the aggregated market orders $Z_t$ in period $t$. We assume $Z_t$s are independent and identically distributed (i.i.d., henceforth) and that $Z_t \overset{d}{\sim} \texttt{T}_{\nu}(0,\sigma)$, where $\overset{d}{\sim}$ stands for ``distributed as", and $\texttt{T}_{\nu}(0,\sigma)$ denotes a location-scale Student's $t$ distribution with $\nu>2$ degrees of freedom, location equal to zero, and scale $\sigma>0$. In general, a continuous random variable $X$ has a location-scale \textit{Student's t} distribution with location $\mu$, scale $\sigma$, and degrees of freedom $\nu$, written $x \overset{d}{\sim} \mathtt{T}_{\nu}(\mu,\sigma)$, if the density function of $X$ is
\begin{equation*}
    \mathfrak{q}_{\nu}(x; \mu, \sigma) = \frac{\Gamma\left(\frac{\nu+1}{2}\right)}{\Gamma\left(\frac{\nu}{2}\right)\sqrt{\pi \nu} \sigma}\left\{1+\frac{1}{\nu}\left(\frac{x-\mu}{\sigma}\right)^2\right\}^{-\frac{(\nu+1)}{2}},\quad x \in \mathbb{R},
\end{equation*}
where $\Gamma(\cdot)$ denotes the gamma function.   We have $\mathbb{E}[x]=\mu$, if $\nu>1$, and $\mathbb{V}\text{ar}[x]=\sigma^2\frac{\nu}{\nu-2}$, if $\nu>2$. In addition, we assume that $(Z_t)_{t=1}^{T}$ are independent of $V$.

In each period $t$, informed traders, conditional on their private information about $V$, choose the order size to maximize their expected trading gain by taking the pricing function in \eqref{eq::LimitOrderBook} as given; $N_t$ is deterministic and known by all the agents in the market (Remark \ref{rmk::random}). If  $X_t$ is the aggregate informed order in period $t$, the total execution cost for  the informed traders is given by (cf. \eqref{eq::CostOfAMarketOrder})
\begin{equation*}
    \int_{0}^{X_t} h(Z_{t} + u , t, Y^{t-1})\,\ud u;
\end{equation*}
we assume the desk charges each insider a pro-rata share of the total execution cost of the aggregate informed block, as in \cite{ccetin2023power}.

Some remarks on our modelling choices are in order.

\begin{remark}[Noise traders' distribution]\label{rem::noise_dist}
We depart from \cite{ccetin2023power} by assuming aggregate noise trader orders follow a heavy-tailed distribution, not a Gaussian one. This leads to a richer equilibrium structure and broader economic implications (Sections \ref{section::Equilibrium}--\ref{section::Asymptotic_Price_Impact}). Large order imbalances are less indicative of private information since they may result from rare liquidity shocks. The tail index $\nu$, which we call \emph{liquidity tail risk} governs this ambiguity: heavier tails (i.e., smaller $\nu$) make extreme shocks more common. In equilibrium, this produces a flatter, more concave price schedule, persistent heavy-tail effects in large orders, and slower declines in spreads and impact slopes when noise trading is more heavy-tailed. Importantly, the Student-$t$ assumption is structural and consistent with recent theory and empirical work (e.g., \cite{gopikrishnan2000statistical, farmer2004origin, bouchaud2018trades, saddier2024bayesian}). In particular, in \cite[][Appendix D.2]{saddier2024bayesian} the uninformed trading volume has an even distribution where the corresponding probability density function (p.d.f.) has the tail decay given by the density $\mathfrak{p}(z) \sim C |z|^{-(1+\nu)}$ with $\nu>0$ being the tail index. The typical value therein for this index is shown to be $\nu \approx 5/2$. This modeling choice complicates fixed-point and equilibrium existence proofs, even with Gaussian scale-mixture representations, but our framework admits greater flexibility in noise specifications. In this respect, notice that our heavy-tailed noise is equivalent to a model in which, at each round of trading, the market draws an unobserved liquidity state that scales the variance of non-informational order flow. An upshot of this observation is that in our theory we can consider noise distributions admitting other  Gaussian  scale-mixture representations.
\end{remark}
\begin{remark}[Random participation of informed traders]\label{rmk::random}
In our model, the number of informed traders $N_t$ is deterministic and common knowledge. This assumption isolates the informational content of order flow from uncertainty about strategic participation.  We leave for future research the possibility of a random $N_t$ and refer the reader to \cite{mingwei2025} for a potential approach.
\end{remark}
\begin{remark}[Within trading round execution priority]\label{remark::priority}
We assume that uninformed orders reach the book before informed orders within each trading round, but insiders do not observe the contemporaneous realization of $Z_t$ when choosing their orders. The ordering affects execution while preserving the information structure:
insiders condition on $V$ and the public history, not on realized noise demand. Camouflage (as in \cite{kyle1985continuous}) is,
therefore, statistical and ex-post from the perspective of liquidity suppliers: conditional on aggregate flow, the latter continue to assign positive probability to a liquidity shock. The assumed protocol avoids within-period feedback and keeps the intertemporal learning channel transparent, as new information enters only via end-of-period aggregate order flow.
\end{remark}
\begin{remark}[Myopic informed traders and informed traders' distribution]\label{rem::remark_myopic}
    The model assumes that informed traders are myopic and may therefore submit large orders, as in \cite{glosten1985bid, easley1987price, ozsoylev2010price}. This behavior aligns also with Seppi's 1990 study \cite{seppi1990equilibrium}, which demonstrates that even when large orders can be divided into a series of smaller orders, informed traders may still choose to place large orders. We leave for future research the possibility of endogenous re-trading and its implications on the conclusion of our model. Regarding informed trades' distribution, in our model trade sizes can vary and we do not confine our analysis to two trade sizes only since our focus is more on the intertemporal equilibrium dynamics. In this respect, our model extends, similarly to \cite{ozsoylev2010price}, the framework in \cite{glosten1985bid}.
\end{remark}

\section{Equilibrium}\label{section::Equilibrium}
This section solves for a symmetric (conditional) equilibrium of the model and derives the first economic implications. The equilibrium is analyzed by first describing the optimal strategies of informed traders (Subsection \ref{subsection::The_Informed_Traders _Optimal_Market_Order}), then the liquidity suppliers' belief update about the risky asset fundamental value (Subsection \ref{subsection::Liquidity_Suppliers_Belief_Update}), and then liquidity suppliers’ pricing function in \eqref{eq::LimitOrderBook} (Subsection \ref{subsection::The_Limit_Prices_in_Equilibrium}), which leads to the definition of equilibrium in our model (Definition \ref{def::equilibrium}), whose existence is discussed in Subsection \ref{subsection::Equilibrium_existence}.\\
\noindent Before proceeding, we summarize strategies and beliefs. An insider strategy in period $t$ is a measurable function $x_{i,t}:\operatorname{supp}(V)\times\mathbb{R}^{t-1}\to\mathbb{R}$ mapping $(v,Y^{t-1})$ into an order size; $\operatorname{supp}(V)$ denotes the support of $V$. A liquidity-supplier (competitive) strategy is a measurable mapping from histories to pricing schedules $h(\cdot,t,Y^{t-1})$. Finally, a belief system is a sequence of posteriors $(\mathbb{P}_t)_{t\ge1}$ with $\mathbb{P}_t(\cdot)=\mathbb{P}(\,V\in\cdot\mid Y^{t-1}\,)$.
\subsection{The informed traders' optimal market order}\label{subsection::The_Informed_Traders _Optimal_Market_Order}
Informed traders arrive simultaneously in period $t$, and each of them is charged by the trading desk for an amount proportional to their order size. Arriving in period $t$ after noise traders (whose aggregate flow is $Z_t$), an insider, who knows the period-$(t-1)$ history and the realization $v_0 \in [m,M]$ of the fundamental value of the asset chosen by nature at time $0$, chooses the market order size $x_t$ to maximize their expected trading gain  
\begin{equation}\label{eq::expected_gain}
    \mathbb{E}^{v_0}\left[V x_t - \frac{x_t}{U_t+x_t}\int_{0}^{U_t+x_t} h(y+Z_t, t, Y^{t-1})\,\ud y\Big\vert Y^{t-1}\right],
\end{equation}
where $\mathbb{E}^{v_0}[\cdot]$ denotes the expectation of $``\cdot"$
under a regular conditional law given the event $\{V=v_0\}$, and $U_t$ is the aggregated demand of the remaining $(N_t-1)$ informed traders; whenever $U_t+x_t=0$, the cost share is interpreted by its continuous extension. The first-order condition associated with the previous optimization problem is given by 
\begin{equation*}
    v_{0} =     \mathbb{E}^{v_0}\left[\frac{x_t}{U_t+x_t} h(U_t+x_t+Z_t, t, Y^{t-1}) + \frac{U_t}{(U_t+x_t)^2} \int_{0}^{U_t+x_t} h(y+Z_t, t, Y^{t-1})\,\ud y\Big\vert Y^{t-1}\right]
\end{equation*}
Risk-neutrality and symmetry of informed traders simplify this condition since the equilibrium demand $x^{\star}_t$ for each insider must be the same. We have:  
\begin{equation*}
    v_{0} =     \mathbb{E}^{v_0}\left[\frac{ h(N_t x^{\star}_t+Z_t, t, Y^{t-1})}{N_t}+\frac{N_t-1}{N_t^2 x^{\star}_t} \int_{0}^{N_t x^{\star}_t} h(u + Z_t, t, Y^{t-1})\,\ud u\Big\vert Y^{t-1}\right],
\end{equation*}
Now, if we define $F(\cdot, t, Y^{t-1}):\mathbb{R}\mapsto\mathbb{R}$ the following marginal cost function 
\begin{equation}\label{eq::expected_margins}
    F(x, t, Y^{t-1}):=\mathbb{E}\left[\frac{ h(x+Z_t, t, Y^{t-1})}{N_t}+\frac{N_t-1}{N_t x} \int_{0}^{x} h(u + Z_t, t, Y^{t-1})\,\ud u\Big\vert Y^{t-1}\right],
\end{equation}
with $F(0, t, Y^{t-1}):=\mathbb{E}[ h(Z_t, t, Y^{t-1}) \mid Y^{t-1}]$ by continuity, the above expression can be written as $V=F(X^{\star}_t, t, Y^{t-1})$ where $X^{\star}_t$ is the total informed demand. Under the assumption that $h(\cdot,t,Y^{t-1})$ is non-decreasing and not constant, the function $F(\cdot, t, Y^{t-1})$ is strictly increasing. As a consequence, the first-order condition uniquely characterizes the equilibrium total informed demand as $X_t^{\star}=F^{-1}(v_0, t, Y^{t-1})$.
\subsection{The liquidity suppliers' belief update}\label{subsection::Liquidity_Suppliers_Belief_Update}
Liquidity suppliers are Bayesian, and update their belief about the fundamental value of the risky asset in each period $t$ after having observed the realized aggregate trades $Y^{t-1}$. Formally, the posterior measure assigned in period $t$ by liquidity suppliers to the risky payoff $V$ given  the realized history $Y^{t-1}$ is proportional (symbol $\propto$) to
\begin{equation}\label{eq::beliefsupdate}
    \mathbb{P}_{t}(V \in \ud v \vert Y^{t-1})\; \propto \; \mathbb{P}(V \in \ud v) \prod_{s=1}^{t-1}\left\{1+\frac{1}{\nu}\left[\frac{Y_{s}-F^{-1}(v, s, Y^{s-1})}{\sigma}\right]^2\right\}^{-\frac{(\nu+1)}{2}},
\end{equation}
where $F(\cdot, s, Y^{s-1})$ is the function defined in  \eqref{eq::expected_margins}. The fact that we do not need the precise normalization constant on the right-hand side of the previous equation will be clear and promptly remarked in Remark \ref{rem::constant}. In order to understand \eqref{eq::beliefsupdate}, it is sufficient to focus on the (conditional) equilibrium informed and uninformed total demand at time $t-1$ (cf. Subsection \ref{subsection::The_Informed_Traders _Optimal_Market_Order}): $X_{t-1}^{\star}=F^{-1}(v, t-1, Y^{t-2})$ and $Z_{t-1}$. The latter follows a location-scale Student's $t$ distribution $\mathtt{T}_{\nu}(0,\sigma)$, independent of $V$. Hence, $Y_{t-1}$ ($=X_{t-1}^{\star}+Z_{t-1}$), conditional on the pair $(Y^{t-2},V)$ follows a location-scale Student's $t$ distribution  $\mathtt{T}_{\nu}(F^{-1}(V, t-1, Y^{t-2}),\sigma)$, whose density appears as a term in the product \eqref{eq::beliefsupdate}. By iterating the previous reasoning from period $t-1$ back to the period $0$, one easily obtains the formula in  \eqref{eq::beliefsupdate}. Before proceeding, we introduce the following notation. We let $\mathfrak{p}_{V}(\cdot)$ be the probability density function (p.d.f., henceforth) of $V$ (i.e., $\mathbb{P}(V \in \ud v)=\mathfrak{p}_{V}(v)\,\ud v$), and $\mathfrak{p}_{t,V}(\cdot \vert Y^{t-1})$ be the conditional p.d.f of $V$ at time $t$ given the period-$(t-1)$ history $Y^{t-1}$ (i.e., $\mathbb{P}_t(V \in \ud v|Y^{t-1})=\mathfrak{p}_{t,V}(\cdot|Y^{t-1})\,\ud v$).
\subsection{The limit prices in equilibrium and definition of equilibrium}\label{subsection::The_Limit_Prices_in_Equilibrium}
Limit prices are given (as in \cite{glosten1985bid, ccetin2023power}) by ``tail expectations". Formally, let $Y_t=Z_t+X_t$ the total demand in period $t$. Then 
\begin{equation}\label{eq::Tail_Expectation}
    h(y, t, Y^{t-1}) = 
    \begin{cases}
        \mathbb{E}[V| Y^{t-1}, Y_t \geq y],\quad\text{if $y>0$}\\
        \mathbb{E}[V| Y^{t-1}, Y_t < y],\quad\text{if $y<0$}.\\
    \end{cases}
\end{equation}
The first case defines the \emph{limit ask price} for the $y$-th share ($y>0$), whereas the second  defines the \emph{limit bid price} ($y<0$). Consistent with the assumption of perfect competition among liquidity suppliers, the rule \eqref{eq::Tail_Expectation} guarantees zero expected profit for them in period $t$ (see Section 2.2 in \cite{ccetin2023power}). We note that we treat buy side $y>0$ and sell side $y<0$ separately. Therefore, the \emph{best ask} and \emph{best bid} correspond to the following two limits, respectively
\begin{equation*}
    h(0+, t, Y^{t-1}):=\lim_{y \downarrow 0} h(y, t, Y^{t-1}),\quad h(0-, t, Y^{t-1}):=\lim_{y \uparrow 0} h(y, t, Y^{t-1}).
\end{equation*}
The \emph{bid-ask spread} is given by the difference between the previous two quantities.

Now, we give the definition of equilibrium for our economy.
\begin{definition}[Equilibrium]\label{def::equilibrium}
The sequence of pairs $(h^{\star}(\cdot, t, Y^{t-1}), X_t^{\star})_{t \in \mathbb{N}}$ is an equilibrium if, in  every period $t$, $h^{\star}(\cdot, t, Y^{t-1})$ is non-decreasing and non-constant, $X_t^{\star} \in \mathbb{R}$ and
\begin{itemize}
    \item[(i)] $h^{\star}(\cdot, t, Y^{t-1})$ satisfies  \eqref{eq::Tail_Expectation} with $Y_t=X_t^{\star}+Z_t$;
    \item[(ii)] $X_t^{\star}$ is the profit maximizing order size for the insider(s) given $h^{\star}(\cdot, t, Y^{t-1})$; in other words, $V=F(X^{\star}_t, t, Y^{t-1})$, where $F(\cdot, t, Y^{t-1})$ is given by  \eqref{eq::expected_margins}.
\end{itemize}
\end{definition}

\noindent The monotonicity  assumption in the previous definition is due to the fact that deeper buys pay higher marginal prices and deeper sells receive lower marginal prices, which is consistent with a LOB. On the other hand, the non-constant assumption rules out a degenerate, i.e., flat, limit price with no adverse selection and no information aggregation.

 Henceforth, to further simplify the notation, we denote the posterior measure in  \eqref{eq::beliefsupdate} and the corresponding expectation by $\mathbb{P}_t(\cdot)$ and $\mathbb{E}_t[\cdot]$ and suppress the dependency on past history. Next, in order to characterize the equilibrium LOB $h^{\star}(\cdot, t, Y^{t-1})$, we define the following right-continuous functions in every period $t$  
\begin{equation}\label{eq::rightcontinuous_functions}
    \Phi_t^+(y):=\mathbb{E}_t[V\mathtt{1}_{\{V \geq y\}}],\,\,\Pi_t^+(y):=\mathbb{P}_t(V \geq y),\,\,\Phi_t^-(y):=\mathbb{E}_t[V\mathtt{1}_{\{V < y\}}],\,\,\Pi_t^-(y):=\mathbb{P}_t(V < y), 
\end{equation}
along with the corresponding tail expectations, when the denominator is strictly positive, 
\begin{equation}\label{eq::tail_expectations}
    \Psi_t^{\pm}(y):=\frac{\Phi_t^{\pm}(y)}{\Pi_t^{\pm}(y)}=
    \begin{cases}
    \mathbb{E}_t[V \vert V \geq y ]\\
    \mathbb{E}_t[V \vert V < y ].
    \end{cases}
\end{equation}
When $\Pi_t^+(M)=0$, we set $\Psi_t^+(M):=M$. When $\Pi_t^-(m)=0$, we set
$\Psi_t^-(m):=m$. These conventions are used only at endpoint limits. We can write $h^{\star}$ in Definition \ref{def::equilibrium} as a function of the previous quantities, for $y>0$
\begin{equation}\label{eq::equilibrium_as_function_of_Pi}
    \begin{split}
        h^{\star}(y, t, Y^{t-1}) &=\mathbb{E}_t[V \vert X_t^{\star}+Z_t \geq y]=\frac{\mathbb{E}_t[V\mathtt{1}_{\{ X_t^{\star}+Z_t \geq y\}}]}{\mathbb{P}_t(X_t^{\star}+Z_t \geq y)}\overset{\text{(ii)}}{=}\frac{\mathbb{E}_t[V\mathtt{1}_{\{ V \geq F(y-Z_t, t, Y^{t-1})\}}]}{\mathbb{P}_t[ V \geq F(y-Z_t, t, Y^{t-1})]}\\
                                 &=\frac{\int_{-\infty}^{\infty}\Phi_t^{+}(F(y-z, t, Y^{t-1})) \mathfrak{q}_{\nu}(z; 0, \sigma)\, \ud z}{\int_{-\infty}^{\infty}\Pi_t^{+}(F(y-z, t, Y^{t-1})) \mathfrak{q}_{\nu}(z; 0, \sigma)\, \ud z};
    \end{split}
\end{equation}
the case $y<0$ is similar. 
\begin{remark}[Normalizing constant in  \eqref{eq::beliefsupdate}]\label{rem::constant}
    Importantly, the normalizing constant of $\mathbb{P}_t(\cdot)$ appears, via $\Phi_t^{\pm}$ and $\Pi_t^{\pm}$, both at the numerator and denominator in the previous expressions, and therefore it cancels out. Consequently, liquidity suppliers' belief update rule \eqref{eq::beliefsupdate} is sufficient for the characterization of the equilibrium, and no further normalization is required.  
\end{remark}

It is also convenient to define, for every $t$ and every continuous function $g_t:\mathbb{R}\mapsto\mathbb{R}$, the following two mappings 
\begin{equation}\label{eq::g_mappings}
    \phi_{g_t}^{+}(x) =\frac{\int_{-\infty}^{\infty}\Phi_t^{+}(g_t(z)) \mathfrak{q}_{\nu}(x-z; 0, \sigma)\, \ud z}{\int_{-\infty}^{\infty}\Pi_t^{+}(g_t(z)) \mathfrak{q}_{\nu}(x-z; 0, \sigma)\, \ud z},\quad     \phi_{g_t}^{-}(x) =\frac{\int_{-\infty}^{\infty}\Phi_t^{-}(g_t(z)) \mathfrak{q}_{\nu}(x-z; 0, \sigma)\, \ud z}{\int_{-\infty}^{\infty}\Pi_t^{-}(g_t(z)) \mathfrak{q}_{\nu}(x-z; 0, \sigma)\, \ud z}.
\end{equation}
If the denominator in \(\phi_g^+\) is zero, which occurs only at the endpoint case \(g\equiv M\), we set \(\phi_g^+(x)=M\). If the denominator in \(\phi_g^-\) is zero, which occurs only at \(g\equiv m\), we set \(\phi_g^-(x)=m\). We also define the following mapping
\begin{equation}\label{eq::equation_for_g}
    \phi_{g_t}(x):=\phi_{g_t}^{+}(x)\mathtt{1}_{\{x \geq 0\}} + \phi_{g_t}^{-}(x)\mathtt{1}_{\{x < 0\}}.
\end{equation}

We now turn to the writing of a fixed-point equation for the function $F(\cdot, t, Y^{t-1})$, that will in turn provide the  characterization of equilibrium that we will use. Equation \eqref{eq::expected_margins} and the expression for $h$ via \eqref{eq::equilibrium_as_function_of_Pi} allow us to write the following equation 
\begin{equation}\label{eq::fixed_point_equation}
F(x, t, Y^{t-1}) = \int_{-\infty}^{+\infty}\left[\frac{1}{N_t}\mathfrak{q}_{\nu}(x-z; 0, \sigma)+\frac{N_t-1}{N_t}\bar{\mathfrak{q}}_{\nu}(x, z; 0, \sigma)\right]\phi_{F(\cdot, t, Y^{t-1})}(z)\,\ud z,
\end{equation}
provided one can interchange the order of integration (see the derivation in the beginning of Section 3 in \cite{ccetin2023power} for more details), where 
\begin{equation}\label{eq::equ_q}
    \bar{\mathfrak{q}}_{\nu}(x, z; 0, \sigma):=\mathtt{1}_{\{x \neq 0\}}\frac{1}{x}\int_{0}^{x}\mathfrak{q}_{\nu}(u-z; 0, \sigma)\,\ud u + \mathtt{1}_{\{x =0\}} \mathfrak{q}_{\nu}(z; 0, \sigma).
\end{equation}
\noindent In the previous equation, we use $\mathfrak{q}_{\nu}(z; 0, \sigma)$ at $x=0$; for symmetric $\mathfrak{q}_{\nu}(z; 0, \sigma)$ (as the Student-$t$ centered at $0$), this matches the limit $\lim_{x\rightarrow 0}\frac{1}{x}\int_{0}^{x}\mathfrak{q}_{\nu}(u-z; 0, \sigma)\,\ud u=\mathfrak{q}_{\nu}(-z; 0, \sigma)=\mathfrak{q}_{\nu}(z; 0, \sigma)$. Now, we define the \emph{Student-$t$ fixed-point operator} as 
\begin{equation}\label{eq::Student_fixed_point_operator}
    (\mathcal{T}_{t,\nu} g)(x):=  \int_{-\infty}^{+\infty}\left[\frac{1}{N_t}\mathfrak{q}_{\nu}(x-z; 0, \sigma)+\frac{N_t-1}{N_t}\bar{\mathfrak{q}}_{\nu}(x, z; 0, \sigma)\right]\phi_{g_t}(z)\,\ud z.
\end{equation}
It is important to note the following. If $(h^{\star}(\cdot, t, Y^{t-1}), X_t^{\star})_{t \in \mathbb{N}}$ is an equilibrium as in Definition \ref{def::equilibrium}, then the associated marginal-cost schedule $(F(\cdot, t, Y^{t-1}))_{t \in \mathbb{N}}$ defined in \eqref{eq::fixed_point_equation} satisfies, at every period $t$, the following three conditions: (1) $F(\cdot, t, Y^{t-1})=\mathcal{T}_{t,\nu}F(\cdot, t, Y^{t-1})$; (2) $h^{\star}(\cdot, t, Y^{t-1})=\phi_{F(\cdot, t, Y^{t-1})}$; and (3) $X_t^{\star}= F^{-1}(V, t, Y^{t-1})$. Conversely, let, for every trading period $t$, $F(\cdot, t, Y^{t-1})$ be continuous from $\mathbb{R}$ to $[m,M]$ and solve $F(\cdot, t, Y^{t-1})=\mathcal{T}_{t,\nu}F(\cdot, t, Y^{t-1})$. Then, if the function $F(\cdot, t, Y^{t-1})$ is strictly increasing, satisfies 
\begin{equation*}
    \lim_{x \rightarrow - \infty} F(x, t, Y^{t-1})=m,\quad\quad \lim_{x \rightarrow + \infty} F(x, t, Y^{t-1})=M,
\end{equation*}
and the associated pricing schedule $h(\cdot, t, Y^{t-1}):=\phi_{F(\cdot, t, Y^{t-1})}$ is non-decreasing and non-constant, then $(h(\cdot, t, Y^{t-1}), X_t^{\star}:=F^{-1}(V, t, Y^{t-1}))_{t \in \mathbb{N}}$ is an equilibrium as in Definition \ref{def::equilibrium}.

Equation \eqref{eq::Student_fixed_point_operator} is the main fixed-point equation of the paper. In the case of a \textit{Gaussian fixed-point operator} (i.e., $\nu \rightarrow \infty$) and in a one-period competitive LOB framework, \cite{ccetin2023power} prove that the map $g \mapsto \phi_{g}$ does preserve monotonicity, i.e., $g$ non-decreasing implies $\phi_{g}$ non-decreasing (cf. their Lemma 4.1 (iv)), and then apply Schauder’s fixed-point on the following class:
\begin{equation}\label{eq::monotone_class}
    D:=\{g \in \mathcal{X}\,|\,g=g_0,\,\mu_0-\text{a.e.}\,\text{for some}\,g_0 \in D_0\},
\end{equation}
where 
\begin{equation}\label{eq::monotone_class_one}
    D_0:=\{g | g:\mathbb{R}\rightarrow[m,M]\,\text{is such that}\,|g(x)-g(y)|\leq K_0 |x-y|,\,\forall x,y \in \mathbb{R}\},
\end{equation}
and $\mathcal{X}:=L^{2}(\mathbb{R},\mu_0)$ is the space of Borel measurable functions that are square integrable with respect
to $\mu_0$, where $\mu_0(\ud x):=\frac{1}{\sqrt{2 \pi}} e^{-\frac{x^2}{2}}\,\ud x$. The assumption of Gaussianity for the distribution of noise trades in \cite{ccetin2023power} is crucial for the proof of their Lemma 4.1 (iv) because it is the basis of the $h$-transform technique based on the Brownian motion. Therefore, in order to prove that $g_t \mapsto \phi_{g_t}$ does preserve monotonicity, it might be tempting to use the fact that  a location-scale Student's $t$-distribution arises as a Gaussian distribution whose variance is randomly scaled by an independent inverse-gamma random variable (cf., also, Lemma \ref{lemma::monotonicity}). However, we provide in Appendix \ref{app::monotonicity_failure}\footnote{In order to keep the writing of the main text concise, we have decided to move the subsequent example in the appendix; however, together with the example in Appendix \ref{app::continuity_failure}, it represents an important contribution of the current paper.} an example showing that, in our framework, the monotonicity preservation fails even on the class of bounded, non-decreasing, and Lipschitz functions. In addition, to prove the existence of fixed-points for the Student-$t$ operator in \eqref{eq::Student_fixed_point_operator}, it might be tempting to repeat the same compactness argument -- in a $\mathcal{X}$ topology with $\mu_0(\ud x)$ as a reference measure -- in \cite{ccetin2023power} by using, instead, the Student-$t$ distribution as a reference measure. However, we provide in Appendix \ref{app::continuity_failure} an example that shows that the Student-$t$ fixed-point operator in \eqref{eq::Student_fixed_point_operator} is not continuous on the closure of the previous class. This is due to the fact that the compact class used in the Gaussian case is too large under Student-$t$ noise. The main difficulty is the Student-$t$ kernel inside the pricing operator. Indeed, because its tails decay only polynomially, remote tail regions remain visible to the conditional-expectation ratio defining $\phi_g$. At endpoint schedules, the numerator and denominator of this ratio may both converge to zero while the ratio converges to a non-endpoint value. Hence the operator is not continuous on the broad class. Importantly, this mathematical obstruction is the same feature that makes the model economically interesting: rare liquidity states remain relevant for pricing at polynomial order.

We now make the following important remark. 
\begin{remark}\label{rmk::monotonicity_failure}
In the example in Appendix \ref{app::monotonicity_failure},  we do not use as $g$ a function solving \eqref{eq::fixed_point_equation}. In this remark, we investigate the monotonicity preservation problem if we consider \eqref{eq::fixed_point_equation} with $\mathcal{T}_{t,\nu}F(\cdot, t, Y^{t-1})=F(\cdot, t, Y^{t-1})$ and $\phi_{F(\cdot, t, Y^{t-1})}=h(\cdot, t, Y^{t-1})$ the associated pricing schedule, then it can be rewritten as
\begin{equation}\label{eq::rewriting_F_Section_3}
    \begin{split}
        F(x,t,Y^{t-1}) &= \frac{1}{N_t}\int_{-\infty}^{+\infty} \mathfrak{q}_{\nu}(x-z;0,\sigma)\phi_{F(\cdot,t,Y^{t-1})}(z)\,\ud z\\
                       &+ \frac{N_t-1}{N_t}\int_{-\infty}^{+\infty}\bar{\mathfrak{q}}_{\nu}(x,z;0,\sigma)\phi_{F(\cdot,t,Y^{t-1})}(z)\,\ud z\\
                       &:=\frac{1}{N_t} G(x,t,Y^{t-1}) + \frac{N_t-1}{N_t} H(x,t,Y^{t-1})\\
        \text{where\quad}H(x,t,Y^{t-1}) &=\frac{1}{x}\int_{0}^{x} G(u,t,Y^{t-1})\,\ud u\quad (x \neq 0),\,\,\text{and}\,\,G(x, t, Y^{t-1})=(\mathfrak{q}_{\nu}*h)(x),
    \end{split}
\end{equation}
where the symbol ``*" denotes, hereafter, the convolution. Now, we solve the identity in \eqref{eq::rewriting_F_Section_3} for $G$ and $H$ directly in terms of $F$; for the sake of notation, we drop the dependence on time and the order-flow history. We have $x H^{'}(x) + H(x) = G(x)$. Substituting the previous equality into $N F = G + (N-1) H$, we obtain $x H^{'}(x) + N H(x) = N F(x)$, and set $\Phi_{x}(s) := s^{N}H(s x)$; in particular $\Phi_{x}^{'}(s)=N s^{N-1} F(s x)$. By integrating from $0$ to $1$, we obtain $H_F(x)=N \int_{0}^{1} s^{N-1}F(s x)\ud s$, where $x \in \mathbb{R}$. Therefore:
\begin{equation*}
    G_F(x)=N F(x)-(N-1)H_F(x)=N F(x)-N(N-1)\int_{0}^{1} s^{N-1} F(s x)\,\ud s.
\end{equation*}
So, every fixed-point satisfies the following convolution identity: $\mathfrak{q}_{\nu} * h=G_F$ and $h = \phi_{F}$. From the previous equation we obtain:
\begin{equation*}
    G_F^{'}(x)=N F^{'}(x)-(N-1)H_F^{'}(x)=N F^{'}(x)-N(N-1)\int_{0}^{1} s^{N} F^{'}(s x)\,\ud s.
\end{equation*}
In Section \ref{section::Asymptotic_Price_Impact}, we will prove that, under some regularity assumptions on the p.d.f of $V$, $M-F(x)$ is regularly varying at $+\infty$ with index $\rho^{+} \in (-1,0)$; see Appendix \ref{app::Supporting_Results_Market_Impact}, Subsection \ref{subsec::regular_variation} for the definition of regularly varying functions and the corresponding results. In this case, $F^{'}(x) \sim - c \rho^{+} x^{\rho^{+}-1}$ is regularly varying with index $\rho^{+}-1$ if, in addition, $\frac{x F'(x)}{M-F(x)}\to -\rho^+$, and we can write that
\begin{equation*}
    G_F^{'}(x)\sim N \left(1-\frac{N-1}{N+\rho^{+}}\right)F^{'}(x)=N\frac{1+\rho^{+}}{N+\rho^{+}}F^{'}(x),
\end{equation*}
which is strictly positive for sufficiently large $x$. Therefore, we can conclude that the convolution $\mathfrak{q}_{\nu}*h$ is asymptotically increasing in the ask tail. Similar computations yield the exact analogue on the bid tail. However, this does not prove the monotonicity of $h$ in either tail. Nevertheless, the previous calculation is useful. It shows that along a regularly varying fixed-point branch, the smoothed price \(\mathfrak{q}_{\nu}*h\) is asymptotically increasing in the ask tail. This is the
first indication that monotonicity can be recovered in the tails even though the global monotonicity-preservation property fails, as we will confirm in Section \ref{section::Asymptotic_Price_Impact}. Nonetheless, we strongly believe, that also the global monotonicity-preservation property holds true in this case.
\end{remark}

In the next subsection, we thus replace the broad Gaussian class with a tail-controlled compact class adapted to the Student-$t$ kernel.

\subsection{Existence of Student-\texorpdfstring{$t$}{t} fixed-points}\label{subsection::fixed_point_existence}
We shall denote by $m$ (resp. $M$) the lower (resp. the upper) endpoint of the support of the random variable $V$. Note that, since $V$ is assumed to be non-degenerate, we have  $-\infty < m < M < \infty$. We impose the following condition on the function $F$ to ensure that the integral
equation \eqref{eq::fixed_point_equation} is well-defined and changing the order of integration  is justified.
\begin{assumption}[Integrability]\label{ass::integrability}
In every period $t$, the function  $F(\cdot,t,Y^{t-1})$ satisfies  the integrability condition
\begin{equation*}
    \int_{-\infty}^{0}\phi^-_{F(\cdot,t,Y^{t-1})}(z)\mathfrak{q}_{\nu}(z; 0, \sigma)\,\ud z >- \infty.
\end{equation*}
\end{assumption}
Observe that the above is automatically satisfied if $V$ is bounded from below. Effectively, the above  assumption is a mild ``finite expected price'' requirement. It ensures that, when liquidity suppliers average tail expectations over the distribution of noise demand, equilibrium marginal prices remain finite at all depths. Without such an integrability condition the fixed-point operator in \eqref{eq::fixed_point_equation} may be ill-defined.

We now state and prove the following lemma, which is analogous to Lemma 4.1  in \cite{ccetin2023power}. However, in our case, it is only useful for interpretation and for the origin-spread diagnostic. In particular, it does not contain an analogous statement to Lemma 4.1-(iv): the example in Appendix \ref{app::monotonicity_failure} shows that the map $g \mapsto \phi_{g}$ does not preserve monotonicity under Student-$t$ noise trades.\\ 
\begin{lemma}\label{lemma::monotonicity}
Let  $t\geq 0$ and $x\in \mathbb{R}$ be given  and consider a filtered probability space \hfil\\ $(\Omega, \mathcal{F}, (\mathcal{F}_s)_{s \geq t}, \mathbb{P}_t)$, which supports a standard Brownian motion $B=(B_{s})_{s \geq t}$  with $B_t=x$, and an independent   $T \overset{d}{\sim}\texttt{Inv-Gamma}\left(\frac{\nu}{2}, \frac{\sigma^2 \nu}{2}\right)$ with $\nu>2$ and $\sigma>0$; in general, $\texttt{Inv-Gamma}(\alpha,\beta)$ denotes an inverse Gamma distribution with shape (real) $\alpha>0$ and scale (real) $\beta>0$. Let $g:\mathbb{R}\to \mathbb{R}$ be a continuous function and define the functions 
\begin{equation*}
    u^{+}(s, y) = \mathbb{E}_t[\Pi^{+}_t(g(B_T))\mathtt{1}_{\{T>s\}}\vert B_s=y],\quad u^{-}(s, y) = \mathbb{E}_t[\Pi^{-}_t(g(B_T))\mathtt{1}_{\{T>s\}}\vert B_s=y],
\end{equation*}
where $\Pi^\pm$ is defined according to the formulae in \eqref{eq::rightcontinuous_functions}. Then, the following statements hold.
\begin{itemize}
    \item[(i)] Define the probability measures $(\mathbb{Q}^{+}_t, \mathbb{Q}^{-}_t)$ on $\mathcal{F}$ by
    \begin{equation*}
        \frac{\ud \mathbb{Q}_{t}^{\pm}}{\ud \mathbb{P}} = \frac{\Pi_t^{\pm}(g_t(B_T))}{\mathbb{E}_t[\Pi_t^{\pm}(g_t(B_T))]}.
    \end{equation*}
    Then, on  $(\Omega, \mathcal{F}, (\mathcal{F}_s)_{s \geq t}, \mathbb{Q}_t)$, $B$ is a solution of the following stochastic differential equation (SDE)
    \begin{equation}\label{eq::SDE}
        \ud B_s = \ud W_s+ \frac{u_y(s,B_s)}{u(s,B_s)}\ud s,\,\,s \in [t, T),\,\,B_t=x,
    \end{equation}
    where $W^=(W_s)_{s \geq t}$ is a Brownian motion under $\mathbb{Q}_t^{\pm}$ with $W_t^{\mathbb{Q}_t^{\pm}}=0$, and $u$ is either $u^+$ or $u^-$.
    \item[(ii)] Suppose that $\mathbb{P}_t$ corresponds to the posterior measure in \eqref{eq::beliefsupdate}. Then, we have that $\phi_{g}^{+}(x)=\mathbb{E}_t^{\mathbb{Q}_t^{+}}[\Psi_t^{+}(g(B_T))]$ and $\phi_{g}^{-}(x)=\mathbb{E}_t^{\mathbb{Q}_t^{-}}[\Psi_t^{-}(g(B_T))]$, where $\Psi^{\pm}$ is as in \eqref{eq::tail_expectations} and $(B,\mathbb{Q}_t^{+})$ (resp. $(B,\mathbb{Q}_t^-)$) corresponds to the solution of \eqref{eq::SDE} if $u=u^+$ (resp $u=u^-$) and $\mathbb{E}^{\mathbb{Q}}$ stands for the expectation under $\mathbb{Q}$. 
    \item[(iii)] $\phi_{g}^{+}(0)>\phi_{g}^{-}(0)$.
\end{itemize}
\end{lemma}
\begin{proof}
See Appendix \ref{app::proof_existence_fixed_point}, Subsection \ref{app::proof_Lemma_monotonicity}.
\end{proof}
Notice that the inequality $\phi_g^+(0)>\phi_g^-(0)$ is the local adverse-selection spread. At the origin, an infinitesimal buy execution is more likely in high-value states than in low-value states, while an infinitesimal sell execution is more likely in low-value states. The bid-ask
spread is therefore generated endogenously by the execution event itself.

\begin{lemma}\label{lem::keyproperties_F}
    Suppose Assumption \ref{ass::integrability} holds.  Fix $t$ and let $F(\cdot,t,Y^{t-1})$ be a continuous non-decreasing solution of \eqref{eq::fixed_point_equation}. Then  $\lim_{x\rightarrow+\infty}F(x,t,Y^{t-1})=M$ and $\lim_{x\rightarrow-\infty}F(x,t,Y^{t-1})=m$. Moreover, if $h(\cdot, t, Y^{t-1})=\phi_{F(\cdot, t, Y^{t-1})}$ is non-decreasing and non-constant, then $F(\cdot, t, Y^{t-1})$ is strictly increasing.
\end{lemma}
\begin{proof}
See Appendix \ref{app::proof_existence_fixed_point}, Subsection \ref{app::proof_keyproperties_F}.
\end{proof}
\begin{remark}
The previous lemma links the tail of the marginal-cost branch to economic informativeness. Along a monotone branch, sufficiently large buy orders reveal the upper endpoint of the fundamental support, and sufficiently large sell orders reveal the lower endpoint. If the associated pricing
schedule is monotone, the marginal-cost branch is strictly increasing, so larger informed orders correspond to higher fundamentals and the informed demand is well defined.  
\end{remark}

We now prove fixed-point existence of the Student-$t$ operator in \eqref{eq::fixed_point_equation}. We prove fixed-point existence under Student-$t$ noise by constructing a compact invariant class directly from endpoint primitives of the posterior and by applying Schauder's theorem on that class. First, we impose the following condition on the posterior measure $\mathbb{P}_t(\cdot)$ in \eqref{eq::beliefsupdate}, which guarantees the regularity needed for the tail functionals $\Phi_t^{\pm}$, $\Pi_t^{\pm}$, and $\Psi_t^{\pm}$ in \eqref{eq::rightcontinuous_functions} and \eqref{eq::tail_expectations} to be continuous. From an economic point of view, it says that once the public history $Y^{t-1}$ is fixed, liquidity suppliers have a smooth posterior density to the fundamental on the whole support interval, so there are no atoms or holes in beliefs.
\begin{assumption}\label{ass:posterior-regularity}
For each trading period $t$ and a period-$(t-1)$ history $Y^{t-1}$, the posterior measure
$\mathbb P_t(\cdot)$ is atomless, supported on $[m,M]$, and admits a continuous p.d.f.
$\mathfrak p_{t,V}(\cdot)$ on $[m,M]$ which is strictly positive on $(m,M)$. Endpoint behavior is specified separately in Assumption~\ref{ass:endpoint-regularity} below.
\end{assumption}
In order to construct a concrete invariant class, we need one-sided endpoint asymptotics for the posterior tail operators in \eqref{eq::rightcontinuous_functions} and \eqref{eq::tail_expectations}. 
\begin{assumption}\label{ass:endpoint-regularity}
Fix a trading period $t$ and a period-$(t-1)$ history $Y^{t-1}$. Assume that there exist constants $\kappa_t^+,\kappa_t^- >0$, $C_{t,+},C_{t,-}>0$, and $\beta_t^+,\beta_t^-\in(0,1)$, with $\beta_t^\pm=\frac{\kappa_t^\pm}{\kappa_t^\pm+1}$, such that, as $s\downarrow 0$,
\begin{equation*}
\Pi_t^+(M-s)=C_{t,+}\,s^{\kappa_t^+}\bigl(1+o(1)\bigr),\qquad
M-\Psi_t^+(M-s)=\beta_t^+\,s\bigl(1+o(1)\bigr),  
\end{equation*}
and 
\begin{equation*}
\Pi_t^-(m+s)=C_{t,-}\,s^{\kappa_t^-}\bigl(1+o(1)\bigr),\qquad
\Psi_t^-(m+s)-m=\beta_t^-\,s\bigl(1+o(1)\bigr).
\end{equation*}
\end{assumption}

\begin{remark}
\label{rem:endpoint-density}
Assumption~\ref{ass:endpoint-regularity} is satisfied, for example, if the posterior density
has endpoint power behavior
\begin{equation*}
\mathfrak p_{t,V}(M-s\mid Y^{t-1})\sim c_{t,+} s^{\kappa_t^+-1},
\qquad
\mathfrak p_{t,V}(m+s\mid Y^{t-1})\sim c_{t,-} s^{\kappa_t^- -1},
\qquad s\downarrow 0.
\end{equation*}
In that case $\beta_t^\pm=\frac{\kappa_t^\pm}{\kappa_t^\pm+1}$. Therefore, the exponents introduced below are determined directly from endpoint primitives of the posterior. It is easy to show that the previous conditions are satisfied by a broad class of bounded-support distributions. For instance, the uniform distribution, the Beta distribution on $[m,M]$, smooth densities positive at both endpoints (e.g., the truncated normal on $[m,M]$, the truncated logistic on $[m,M]$, and the truncated Student-$t$ on $[m,M]$), the triangular distribution on $[m,M]$, and power densities near the endpoints. On the other hand, distributions that do not satisfy the previous conditions are unbounded distributions, densities with endpoint atoms, and densities with holes near endpoints. 
\end{remark}
\noindent We now state the following lemma, whose proof is immediate.
\begin{lemma}\label{lem:candidate-exponents}
Suppose Assumption~\ref{ass:endpoint-regularity} holds and $N_t > 1$. Define 
\begin{equation*}
\rho_t^+ := \frac{\beta_t^+-1}{1-\beta_t^+/N_t}\in(-1,0),\qquad
\rho_t^- := \frac{\beta_t^--1}{1-\beta_t^-/N_t}\in(-1,0).
\end{equation*}
Then, the following identities hold
\begin{enumerate}
    \item[(i)] $\beta_t^\pm \left(\frac1{N_t}+\frac{N_t-1}{N_t(1+\rho_t^\pm)}\right)=1$.
    \item[(ii)] $\kappa_t^\pm(-\rho_t^\pm)=\frac{\beta_t^\pm}{1-\beta_t^\pm/N_t}<\frac{N_t}{N_t-1}\le 2<\nu$.
    \item[(iii)] $(\kappa_t^\pm+1)(-\rho_t^\pm)=\frac{1}{1-\beta_t^\pm/N_t}
<\frac{N_t}{N_t-1}\le 2<\nu.$
\end{enumerate}
\end{lemma}
The inequalities in the previous lemma imply that the powers generated by
the endpoint asymptotics are heavier-tailed than the Student-$t$ kernel itself. This is the key analytic fact behind the invariance of the asymptotic class we will define below: convolution with the Student-$t$ density preserves the relevant branch tail order.

We are now ready to define our tail-controlled compact class adapted to the Student-$t$ kernel. Notice that in the following definition, the quantities $\rho_t^{+}$ and $\rho_t^{-}$ are defined from the endpoint behaviour of the posterior tail-expectation operators, as in Lemma \ref{lem:candidate-exponents} above.
\begin{definition}\label{def:template-class}
Fix constants $R>1, L>0$ and $0<\underline c_t^\pm<\overline c_t^\pm<\infty$, and let $\varepsilon_t:[R,\infty)\to(0,1)$ be continuous, decreasing, and such that $\varepsilon_t(x)\downarrow 0$ as $x\to\infty$. We define $\mathcal K_{t,\nu}=\mathcal K_{t,\nu}(R,L,\varepsilon_t,\underline c_t^\pm,\overline c_t^\pm)$
as the set of all continuous functions $g:\mathbb R\to\mathbb R$ such that:
\begin{enumerate}
\item[(i)] \(m\le g(x)\le M\) for all \(x\in\mathbb R\);
\item[(ii)] \(g\) is globally Lipschitz with \(\operatorname{Lip}(g)\le L\);
\item[(iii)] there exist constants $c_g^+ \in[\underline c_t^+,\overline c_t^+]$, $c_g^-\in[\underline c_t^-,\overline c_t^-]$, such that
\begin{equation*}
\left|
\frac{M-g(x)}{c_g^+(1+x)^{\rho_t^+}}-1
\right|
\le \varepsilon_t(x),\,\forall x \geq R,\quad\text{and}\quad \left|
\frac{g(x)-m}{c_g^-(1+|x|)^{\rho_t^-}}-1
\right|
\le \varepsilon_t(|x|),\,\forall x \leq -R 
\end{equation*}
\end{enumerate}
\end{definition}

We now state the main theorem of this section which proves the existence of a fixed-point for the Student-$t$ operator in \eqref{eq::Student_fixed_point_operator}.
\begin{theorem}\label{th::existence_of_a_fixed_point}
Suppose Assumptions \ref{ass::integrability}, \ref{ass:posterior-regularity}, and \ref{ass:endpoint-regularity} hold, and $N_t>1$. Then, for every $\nu > 2$, each period $t$, and every $(t-1)$-history, the Student-$t$ operator in \eqref{eq::Student_fixed_point_operator} has at least one fixed-point $F(\cdot, t, Y^{t-1}) \in \mathcal{K}_{t,\nu}$, where $\mathcal{K}_{t,\nu}$ is defined in Definition \ref{def:template-class}. Moreover, every fixed-point is continuously differentiable and satisfies 
\begin{equation*}
\|F^{'}(\cdot, t, Y^{t-1})\|_\infty \le \overline M\,\frac{N_t+1}{2N_t}I_\nu,\quad
\overline M:=\max\{|m|,|M|\},\quad I_\nu:=\int_{-\infty}^{+\infty}
\left|\frac{\partial}{\partial u}\mathfrak q_\nu(u;0,\sigma)\right|\,\ud u<\infty
\end{equation*}
\end{theorem}
\begin{proof}
    See Appendix \ref{app::proof_existence_fixed_point}, Subsection \ref{subsec::proof_existence_fixed_point}.
\end{proof}

The previous Theorem~\ref{th::existence_of_a_fixed_point} is the constructive replacement for the Gaussian fixed-point argument. It shows that heavy-tailed liquidity shocks do not make the marginal-cost problem ill posed. What changes is the admissible class of candidate books: because Student-$t$ noise keeps remote liquidity states relevant at polynomial order, candidate marginal-cost schedules must control how they approach their endpoint values.

\subsection{From fixed points to monotone LOB equilibria}\label{subsection::Equilibrium_existence}
Theorem \ref{th::existence_of_a_fixed_point} is a fixed-point theorem, not an equilibrium one. It guarantees that the Student-$t$ marginal-cost equation is well behaved on the class in Definition \ref{def:template-class}. In order to obtain an equilibrium as in Definition \ref{def::equilibrium}, we need the global monotonicity of the associated pricing schedule. As discussed earlier, when the distribution of noise trades is Gaussian the existence of a fixed-point and the monotonicity problem collapse because of the global monotonicity preservation of $g \mapsto \phi_{g}$. Under Student-$t$ noise trade, instead, these two questions must be separated. Global
monotonicity is not automatic under Student-$t$ noise. Therefore, we state equilibrium results conditional on the monotone branch and support this selection by tail monotonicity and numerical verification. In particular, in Section \ref{section::Asymptotic_Price_Impact}, we are able to show that monotonicity can be recovered in the tails. Moreover, in our numerical section, we show that we have global monotonicity for a broad range of (economically meaningful) choices of the parameter $\nu$ and of the specification of the fundamental value's prior. Motivated by this analytical result and numerical evidence the remainder of the paper conditions on the following assumption, i.e., all subsequent references to equilibrium mean equilibrium on this selected branch.

\begin{assumption}[Selected monotone branch]
In each period $t$ and after each public history $Y^{t-1}$, the fixed point $F(\cdot,t,Y^{t-1})$ selected from $K_{t,\nu}$ (Definition \ref{def:template-class}) is strictly increasing, and the
associated pricing schedule $h(\cdot,t,Y^{t-1})=\phi_{F(\cdot,t,Y^{t-1})}$ is non-decreasing and non-constant. In particular,  $(h(\cdot,t,Y^{t-1}),X_t^{\star})_{t \in \mathbb{N}}$ with $X_t^{\star}=F^{-1}(V,t,Y^{t-1})$ is a conditional LOB equilibrium.
\end{assumption}

We now analyse the liquidity suppliers' belief posterior consistency.
\section{Bayesian learning from order flow}\label{section::The Liquidity_Suppliers_Belief_posterior_consistency}
The dynamics of informed strategies and the liquidity suppliers' learning process about the fundamental asset value $V$ are closely related (cf. \eqref{eq::beliefsupdate}). Because informed traders' trading horizon is limited to one period (they are myopic), the only link between consecutive periods is the liquidity suppliers' belief on the asset fundamental value. In this section, we investigate whether the liquidity suppliers' beliefs on the risky asset payoff converge to the realization chosen by nature at time zero as the number of trading periods tends to infinity. In particular, this is in line with the convergence result in \cite{glosten1985bid}, although in their model the result holds true if and only if the trade size is the unit trade size, and in \cite{ozsoylev2010price}, where instead multiple trade sizes are allowed. We interpret the time-zero realization of $V$, say $v_0$, as the parameter that governs the history of aggregated orders $Y_1, Y_2, \ldots$. Given this interpretation, to prove the just-mentioned convergence, we use the theory of ``dynamics of Bayesian updating with dependent data and misspecified model'' proposed in \cite{shalizi2009dynamics}; we direct the reader to the extensive coverage on this topic in the introduction of the previous paper. In our case, we have dependence between consecutive aggregated orders (i.e., we deal with non-i.i.d. data) but the model is specified, in the sense that the sequence $Y_1, Y_2, \ldots$ is governed/parameterized by the true value $v_0$. We therefore have to simplify the framework in \cite{shalizi2009dynamics}. In an effort to keep the current paper as self-contained, we report, in Appendix \ref{app::consistency_beliefs} any essential background material required in the derivations of our main convergence result. 
\subsection{Proof of the Bayesian learning from order}
In our model, in equilibrium, informed traders submit orders $X_t^{\star}$ in every round of trading $t$ according to the following rule: $X_t^{\star}=F^{-1}(v_0, t, Y^{t-1})$, where $v_0 \in \operatorname{supp}(V)=[m,M]$ is the realization of the fundamental value of the asset $V$ chosen by nature at time zero. In particular, $Y_t=Z_t+X_t^{\star}=Z_t+F^{-1}(v_0, t, Y^{t-1})$, where $Z_t$ is the aggregated market orders of the noise traders in the trading round $t$. Under the distributional assumptions on $Z_t$, we have that $\mathfrak{p}_{Y_t | Y^{t-1}, v_0}(Y_t) \overset{d}{\sim} \mathtt{T}_{\nu}(F^{-1}(v_0, t, Y^{t-1}),\sigma)$. Because $\lim_{|x|\rightarrow\infty} F^{'}(x, t, Y^{t-1})=0$, we have $\inf_{x \in \mathbb{R}} F^{'}(x, t, Y^{t-1})=0$, and therefore the slope of $F^{-1}(\cdot, t, Y^{t-1})$ is arbitrarily large near the boundaries $m$ and $M$; from this, $F^{-1}(\cdot, t, Y^{t-1})$ cannot be globally Lipschitz, a property that would be required in the proof of the liquidity suppliers' belief posterior consistency. We thus impose the following local regularity of the equilibrium inverse. We fix $\varepsilon\in(0,(M-m)/2)$ and let $\Theta_\varepsilon:=[m+\varepsilon,M-\varepsilon]$ be the compact interior parameter set. We assume $v_0 \in \operatorname{int}(\Theta_{\varepsilon})$, where $\mathrm{int}(\cdot)$ denotes the interior of a set, and that the prior $\Pi_0$ is supported on $\Theta_\varepsilon$, and has $v_0$ in its KL-support.

For all trading rounds $t$ and every public order flow history $Y^{t-1}$, we have that the function $x\mapsto F(x,t,Y^{t-1})$ is $C^{1}$ and strictly increasing on $\mathbb R$; hence, $F'(x,t,Y^{t-1})>0$ for all $x$.  Therefore, the inverse map $v\mapsto F^{-1}(v,t,Y^{t-1})$ is well-defined and Lipschitz on
$\Theta_\varepsilon$, with (random) Lipschitz constant
\begin{equation}\label{eq:Leps_t_def}
L_{\varepsilon,t}(Y^{t-1})
:=\sup_{\substack{v,v'\in\Theta_\varepsilon\\v\neq v'}}
\frac{\left|F^{-1}(v,t,Y^{t-1})-F^{-1}(v',t,Y^{t-1})\right|}{|v-v'|}
=\left(\inf_{x\in I_{\varepsilon,t}(Y^{t-1})}F'(x,t,Y^{t-1})\right)^{-1},
\end{equation}
where
\begin{equation*}
I_{\varepsilon,t}(Y^{t-1})
:=\Big[F^{-1}(m+\varepsilon,t,Y^{t-1}),\,F^{-1}(M-\varepsilon,t,Y^{t-1})\Big].
\end{equation*}
The quantity in \eqref{eq:Leps_t_def} leads to an economically-grounded interpretation. It measures the time-varying elasticity of informed demand with respect to fundamentals. As liquidity suppliers learn, adverse-selection premia typically decline and informed demand becomes more sensitive to the fundamental value. Therefore, it is natural that $L_{\varepsilon,t}$ is not uniformly bounded in $t$. At this point, we make the following
\begin{assumption}\label{ass:moderate_steepening}
Fix $\varepsilon\in(0,(M-m)/2)$ and let, as above, $\Theta_\varepsilon=[m+\varepsilon,M-\varepsilon]$. There exist constants $C_L>0$ and $\kappa>0$ and a deterministic sequence $(L_t)_{t\ge1}$
such that for all $t\ge1$,
\begin{equation*}
L_{\varepsilon,t}(Y^{t-1}) \le L_t \le C_L(1+t)^\kappa
\qquad\text{a.s.}
\end{equation*}
\end{assumption}
\noindent Hence, Assumption~\ref{ass:moderate_steepening} allows for
time-variation and growth in the above-mentioned elasticity but rules out explosive steepening by requiring that the local Lipschitz constant of $v\mapsto F^{-1}(v,t,Y^{t-1})$ grows at most polynomially.
Technically, this polynomial envelope is used in Appendix \ref{app::proof_Lemma_uniform_convergence_loglikelihood_ratio}, \emph{Step 3}, to control a martingale term uniformly over $v\in\Theta_\varepsilon$ via a covering argument and the Borel-Cantelli lemma. We now discuss how to link Assumption ~\ref{ass:moderate_steepening} to the equilibrium primitives of our model. From the proof of Lemma \ref{lemma::monotonicity}, we easily have the following bound\footnote{The bound based on \(\mathfrak q_\nu\) is conservative. The tail analysis in Section~\ref{section::Asymptotic_Price_Impact} typically yields
a slower polynomial decay for the branch derivative, but the \(q_\nu\)-bound is sufficient for the local inverse-regularity argument.}
\begin{equation}\label{ass::bound_assumption_1}
    F^{'}(x,t,Y^{t-1}) \geq \frac{N+1}{2 N} \Delta_t \mathfrak{q}_{\nu}(x;0,\sigma),\quad \forall x \in \mathbb{R},
\end{equation}
where $\Delta_t:=h(0+,t,Y^{t-1})-h(0-,t,Y^{t-1})$ is the equilibrium spread at the origin. Notice that in this section, we assume that $N_t\equiv N$ for all $t$ in order to isolate the learning mechanism and keep the likelihood stationary across periods. Extending the consistency argument to time-varying $(N_t)$ is feasible but requires additional bookkeeping and is left for future work. Now, we define $R_{\varepsilon,t}(Y^{t-1}):=\max\left(|F^{-1}(m+\varepsilon,t,Y^{t-1})|,|F^{-1}(M-\varepsilon,t,Y^{t-1})|\right)$. Then, $\forall x \in I_{\varepsilon,t}$ we have that $|x|\leq R_{\varepsilon,t}$, and $\mathfrak{q}_{\nu}(x;0,\sigma) \geq \mathfrak{q}_{\nu}(R_{\varepsilon,t}(Y^{t-1});0,\sigma)$. By inserting the lower bound of \eqref{ass::bound_assumption_1} into the definition of $L_{\varepsilon,t}$ in \eqref{eq:Leps_t_def}, we obtain the following implication:
\begin{align*}
    &\inf_{x \in I_{\varepsilon,t}} F^{'}(x,t,Y^{t-1}) \geq \frac{N+1}{2 N}\Delta_t \mathfrak{q}_{\nu}(R_{\varepsilon,t}(Y^{t-1});0,\sigma)\\
    &\Rightarrow L_{\varepsilon,t} \leq \frac{2 N}{(N+1)\Delta_t \mathfrak{q}_{\nu}(R_{\varepsilon,t}(Y^{t-1});0,\sigma)}\\
    &\leq C(\nu, \sigma, N) \frac{1}{\Delta_t}\left(1+\frac{R_{\varepsilon,t}(Y^{t-1})^2}{\nu \sigma^2}\right)^{\frac{(\nu+1)}{2}},
\end{align*}
where $C(\nu, \sigma, N)$ is a constant depending only on $(\nu, \sigma, N)$. In particular, all quantities on the right-hand side of the last expression in the previous equation are either a primitive $(\nu, \sigma, N)$ or an equilibrium observable. By imposing that the spread does not collapse too fast, e.g., $\Delta_t \geq c_{\Delta} (1+t)^{-\beta}$ a.s. and that the relevant inverse region does not drift too fast, e.g., $R_{\varepsilon,t}(Y^{t-1}) \leq c_R (1+t)^{\alpha}$ a.s., we obtain $L_{\varepsilon,t}=O((1+t)^{\beta+\alpha(\nu+1)})$, which is exactly Assumption \ref{ass:moderate_steepening} with $\kappa=\beta+\alpha(\nu+1)$.

We also assume the following
\begin{assumption}\label{ass:KL_stabilization}
Fix $\varepsilon$ and $\Theta_\varepsilon$ as in Assumption~\ref{ass:moderate_steepening}. For each $v\in\Theta_\varepsilon$, let
$K_t(v,v_0):=-\mathbb E^{v_0}_{t-1}[\ell_t(v)]=-\mathbb E^{v_0}_{t-1}\left[\log \frac{\mathfrak{f}_{Y_t|Y^{t-1},v}(Y_t)}{\mathfrak{f}_{Y_t|Y^{t-1},v_0}(Y_t)}\right]$.
Then, there exists a deterministic function $K(\cdot,v_0)$ such that
\begin{equation*}
\sup_{v\in\Theta_\varepsilon}\left|
\frac{1}{t}\sum\limits_{s=1}^{t} K_s(v,v_0) - K(v,v_0)
\right|\to 0,
\qquad \mathbb P_{v_0}^\infty\text{-a.s.}
\end{equation*}
\end{assumption}
\noindent The limit order book and, therefore, informed demand depend on liquidity suppliers' beliefs; in particular, the conditional
distribution of the order flow $Y_t\mid Y^{t-1}$ is not i.i.d. over time, even though $(\nu,\sigma,N)$ are fixed. The quantity $K_t(v,v_0)$ measures how informative period-$t$ order flow is for distinguishing the true fundamental $v_0$ from a candidate $v$, given the prevailing belief state. In particular, Assumption~\ref{ass:KL_stabilization} requires that the \emph{time-average} of this conditional informativeness stabilizes and
converges to a deterministic long-run information rate $K(v,v_0)$, uniformly over $v\in\Theta_\varepsilon$. Equivalently, it rules out persistent drift in the informational content of order flow in a stable market environment (no structural breaks in $(\nu,\sigma,N)$), even though the book evolves as learning progresses. From a technical point of view, it plays a role in Appendix \ref{app::consistency_beliefs}, where we show that 
\begin{equation*}
    \frac{1}{t}\log R_t(v)
= -\frac{1}{t}\sum\limits_{s=1}^t K_s(v,v_0)
+ \frac{1}{t}\sum\limits_{s=1}^t D_s(v)
\end{equation*}
with $(D_s(v))_{s\ge1}$ a martingale difference sequence. The martingale term is controlled by
concentration inequalities and vanishes uniformly over $v\in\Theta_\varepsilon$. Assumption~\ref{ass:KL_stabilization} is therefore the stability condition that upgrades this decomposition into the uniform almost sure convergence of $(1/t)\log R_t(v)$ to $-K(v,v_0)$ required by \cite{shalizi2009dynamics}'s framework.

Before proceeding, we state and prove the following lemma which guarantees that $K(v,v_0)>0$ for $v \neq v_0$.
\begin{lemma}\label{lem:KL_separation}
Fix $\varepsilon\in(0,(M-m)/2)$ and let $\Theta_\varepsilon=[m+\varepsilon,M-\varepsilon]$, as above. For $t\ge 1$, we  define $m_t(v):=F^{-1}(v,t,Y^{t-1})$ and $\Delta_t(v):=m_t(v)-m_t(v_0)$. Then, for every $v\in\Theta_\varepsilon$ with $v\neq v_0$, we have:
\begin{enumerate}[(i)]
    \item $\Delta_t(v)\neq 0$ for all $t$ almost surely under $\mathbb P_{v_0}^\infty$.
    \item Conditional on $Y^{t-1}$, the one-step KL increment satisfies
    \begin{equation*}
    K_t(v,v_0)
    =\mathrm{KL}\!\left(\mathtt{T}_\nu(0,\sigma)\,\big\|\,\mathtt{T}_\nu(\Delta_t(v),\sigma)\right)
    :=k(\Delta_t(v)),
    \end{equation*}
    where $k(\Delta)>0$ for $\Delta\neq 0$, $k(0)=0$, and $k$ is continuous.
    \item Let $K_0<\infty$ denote the uniform bound on $F'(x,t,Y^{t-1})$ established in Theorem \ref{th::existence_of_a_fixed_point}.
    Then the inverse schedule is \emph{co-Lipschitz}:
    \begin{equation*}
    |\Delta_t(v)| = |m_t(v)-m_t(v_0)| \ge \frac{|v-v_0|}{K_0},
    \end{equation*}
    and therefore for every $\eta>0$,
    \begin{equation*}
    \inf_{\substack{v\in\Theta_\varepsilon\\ |v-v_0|\ge \eta}} K_t(v,v_0)\ \ge\ \kappa(\eta)\ >\ 0,
    \qquad 
    \kappa(\eta):=\inf_{|\Delta|\ge \eta/K_0}k(\Delta).
    \end{equation*}
\end{enumerate}
\end{lemma}
\begin{proof}
See Appendix \ref{app::consistency_beliefs}, Subsection \ref{app::proof_KL_separation}
\end{proof}
The previous lemma formalizes the fact that distinct candidate fundamentals imply distinct predicted informed demand, hence distinct conditional order-flow distributions. Even with heavy-tailed noise, this creates a strictly positive information distance each period, uniformly away from the truth on the interior set.

The posterior probability measure after observing $Y^{t}$ is given by (cf. \eqref{eq::update_priors})
\begin{equation*}
    \Pi_{t}(\ud v)=\frac{\mathfrak{f}_{v}(Y^t)\Pi_0(\ud v)}{\int_{\Theta_{\varepsilon}}\mathfrak{f}_{v}(Y^t)\Pi_0(\ud v)}.
\end{equation*}
 First, we observe that the likelihood ratio is given by $R_t(v):=\frac{\mathfrak{f}_{v}(Y^{t})}{\mathfrak{p}_{v_0}(Y^{t})}$, which is $\sigma(Y^{t}) \times \mathcal{T}$ measurable, with $\mathcal{T}:=\mathcal{B}(\Theta_{\varepsilon})$ (see Assumption \ref{ass::assumption_consistency_zero}). Then, we state and prove the following important proposition.
\begin{proposition}\label{lemm::Lemma_uniform_convergence_loglikelihood_ratio} Suppose Assumption \ref{ass:KL_stabilization} and \ref{ass:moderate_steepening} hold. Let $\nu>2$, $v \in \Theta_{\varepsilon}$ such that $0<K(v,v_0)<+\infty$, where $K(v,v_0)>0$, with $v \neq v_0$  is defined in \eqref{eq::KL_divergence}.  Then, 
    \begin{equation*}
        \lim_{t \rightarrow \infty} \frac{1}{t}\log\left(R_t(v)\right)=-K(v,v_0),\quad\mathbb{P}_{v_0}^{\infty}-a.s.
    \end{equation*}
    Moreover, the convergence is uniform on $\Theta_{\varepsilon}$
    \begin{equation*}
       \lim_{t \rightarrow \infty}\sup_{v \in \Theta_{\varepsilon}} \Big\vert  \frac{1}{t}\log\left(R_t(v)\right) + K(v,v_0)\Big\vert = 0,\quad\mathbb{P}_{v_0}^{\infty}-a.s.
    \end{equation*}    
\end{proposition}
\begin{proof}
    See Appendix \ref{app::consistency_beliefs}, Subsection \ref{app::proof_Lemma_uniform_convergence_loglikelihood_ratio}.
\end{proof}
Finally, we need to check the verification of Assumption \ref{ass::assumption_consistency_four} and posterior consistency. Because the prior $\Pi_0$ is supported on the compact set $\Theta_\varepsilon$, we may take
$G_t=\Theta_\varepsilon$ for all $t$, so that $\Pi_0(G_t)=1$ and the first bullet of
Assumption~\ref{ass::assumption_consistency_four} holds trivially for any fixed $\alpha>0$ and $\beta>0$. Moreover, Proposition~\ref{lemm::Lemma_uniform_convergence_loglikelihood_ratio} establishes the uniform convergence in Assumption~\ref{ass::assumption_consistency_four} on $G_t=\Theta_\varepsilon$. Finally, since $\Theta_\varepsilon$ is compact and $K(\cdot,v_0)$ is lower semicontinuous with $K(v,v_0)=0$ iff $v=v_0$ (Assumption~\ref{ass::assumption_consistency_one}), for every $\eta>0$ we have the separation $\delta_\eta := \inf\{K(v,v_0): v\in\Theta_\varepsilon,\ |v-v_0|\ge \eta\} > 0$. Applying Theorem~\ref{th:theorem41} with $A_\eta := \{v\in\Theta_\varepsilon:\ |v-v_0|\ge \eta\}$ yields $\Pi_t(A_\eta)\to 0$ almost surely. Hence liquidity suppliers' posteriors concentrate at the true fundamental value $v_0$.

\section{Asymptotics of price impact}\label{section::Asymptotic_Price_Impact}
In this section we use the fixed-point equation \eqref{eq::fixed_point_equation} to determine the shape of price impact for large orders. To this end,  we shall establish the tail behaviour of the marginal cost function $F(\cdot, t, Y^{t-1})$ in every period $t$ by using tools from the theory of regular variation (see, again, Appendix \ref{app::Supporting_Results_Market_Impact}, Subsection \ref{subsec::regular_variation}). We recall the definitions of $\Psi_t^{+}(y)=\mathbb{E}_t[V|V\geq y]$ and $\Psi_t^{-}(y)=\mathbb{E}_t[V|V < y]$. Next, we denote
\begin{equation}\label{eq::derivative}
    \partial_x \Psi_t^{+}(M) = \lim_{x\rightarrow M}\frac{M-\Psi_t^{+}(x)}{M-x}\,\,\,\text{and}\,\,\,\partial_x \Psi_t^{-}(m)=\lim_{x\rightarrow m}\frac{\Psi_t^{-}(x)-m}{x-m}
\end{equation}
In the case where $\operatorname{supp}(V)=[m,M]$, a direct application of integration by parts and L'H\^opital rule implies that  
\begin{equation}\label{eq::derivative_rewriting}
 \partial_x \Psi_t^{+}(M)=\lim_{x\rightarrow M}\frac{1}{-1+(M-x)\frac{\partial_x\Pi_t^{+}(x)}{\Pi_t^{+}(x)}}+1,\,\,\text{and}\,\,\partial_x \Psi_t^{-}(m)=\lim_{x\rightarrow m}-\frac{1}{1+(x-m)\frac{\partial_x\Pi_t^{-}(x)}{\Pi_t^{-}(x)}}+1  
\end{equation}
Now, we remind the following notation. We let $\mathfrak{p}_{V}(\cdot)$ be the p.d.f. of $V$ (i.e., $\mathbb{P}(V \in \ud v)=\mathfrak{p}_{V}(v)\,\ud v$), and $\mathfrak{p}_{t,V}(\cdot \vert Y^{t-1})$ be the conditional p.d.f of $V$ at time $t$ given the period-$(t-1)$ history $Y^{t-1}$ (i.e., $\mathbb{P}_t(V \in \ud v|Y^{t-1})=\mathfrak{p}_{t,V}(\cdot|Y^{t-1})\,\ud v$). In our model, under mild conditions on the fundamental value of the asset $V$, the price impact obeys a power law in any period $t$, as evidenced by the theorem below. Moreover, as expected (cf. Section \ref{section::The Liquidity_Suppliers_Belief_posterior_consistency}), the marginal informativeness of each trade declines as $t$ increases. Before proceeding, we make the following observation.

The existence proof in Section \ref{section::Equilibrium} assumes that $\operatorname{supp}(V)=[m,M]$, which implies that very large equilibrium orders correspond to fundamentals arbitrarily close to the boundary; see Lemma \ref{lem::keyproperties_F}. While this is mathematically consistent, it can make large-order asymptotics sensitive to endpoint behavior of beliefs. We therefore provide an economically natural robustness check to mitigate the previous boundary interpretation. Precisely, in Section \ref{section::Numerical_Studies}, we numerically solve the equilibrium for unbounded fundamentals and confirm that the qualitative predictions of our model remain valid.

\begin{theorem}\label{th::power_law_price_impact}
In any period $t$, let \(F(\cdot,t,Y^{t-1})\) be the fixed-point obtained in Theorem~\ref{th::existence_of_a_fixed_point}, so that it belongs
to the tail-controlled class \(\mathcal K_{t,\nu}\), and assume that $\partial_{x}\Psi_t^{+}(M)$ and $\partial_{x}\Psi_t^{-}(m)$ exist, $N_t>1$, and that $\operatorname{supp}(V)=[m,M]$ with $-\infty<m<M<+\infty$. Then,
\begin{itemize}
    \item If there exists a constant $L \in \mathbb{R}$ such that 
    \begin{equation}\label{eq::tail_behaviour}
        \lim_{x \rightarrow M^{-}}(M-x)\frac{\mathfrak{p}_{V}^{'}(x)}{\mathfrak{p}_{V}(x)}=L,
    \end{equation}
    then, $M-F(\cdot, t, Y^{t-1})$ is regularly varying at $+\infty$ with index $\rho_t^{+}=\frac{\partial_x\Psi_t^{+}(M)-1}{1-\frac{\partial_x\Psi_t^{+}(M)}{N_t}}$, where 
    \begin{equation}\label{eq::Psi_t_plus}
       \partial_x\Psi_t^{+}(M) = \frac{1-L-(\nu+1)\sum\limits_{s=1}^{t-1}\frac{1}{\rho_s^{+}}}{2-L-(\nu+1)\sum\limits_{s=1}^{t-1}\frac{1}{\rho_s^{+}}} \in [0,1).
    \end{equation}
    \item If there exists a constant $L^{'} \in \mathbb{R}$ such that 
    \begin{equation}\label{eq::tail_behaviour2}
        \lim_{x \rightarrow m^{+}}(x-m)\frac{\mathfrak{p}_{V}^{'}(x)}{\mathfrak{p}_{V}(x)}=L^{'},
    \end{equation}
    then, $F(\cdot, t, Y^{t-1})-m$ is regularly varying at $-\infty$ with index $\rho_t^{-}=\frac{\partial_x\Psi_t^{-}(m)-1}{1-\frac{\partial_x\Psi_t^{-}(m)}{N_t}}$, where
    \begin{equation*}
       \partial_x\Psi_t^{-}(m) = \frac{1+L^{'}-(\nu+1)\sum\limits_{s=1}^{t-1}\frac{1}{\rho_s^{-}}}{2+L^{'}-(\nu+1)\sum\limits_{s=1}^{t-1}\frac{1}{\rho_s^{-}}} \in [0,1).
    \end{equation*}
    \item Moreover, 
    \begin{equation*}
        \lim_{t \rightarrow \infty} \partial_x \Psi_t^{+}(M) = 1,\quad  \lim_{t \rightarrow \infty}\partial_x \Psi_t^{-}(m) = 1, \quad,\quad \lim_{t \rightarrow \infty}|\rho_t^{+}| = 0, \quad  \lim_{t \rightarrow \infty}  |\rho_t^{-}| = 0.
    \end{equation*}
\end{itemize}
\end{theorem}
\begin{proof}
See Appendix \ref{app::Supporting_Results_Market_Impact}, Subsection \ref{app::price_impact_proof}.
\end{proof}
Theorem~\ref{th::power_law_price_impact} is stated recursively because the posterior endpoint behavior at date $t$ depends on earlier equilibrium branches. Theorem~\ref{th::power_law_price_impact} identifies the tail exponents selected by the fixed-point equation and shows that they coincide with the exponents used in the tail-controlled existence class. Before discussing the economic implications of Theorem \ref{th::power_law_price_impact}, we notice that the condition in \eqref{eq::tail_behaviour} -- and analogously the left-boundary condition in \eqref{eq::tail_behaviour2} -- is a smoothness assumption on extreme fundamentals and it rules out irregular densities near the valuation cap. Economically, it says that the upper tail of beliefs about $V$ behaves regularly, so extreme order-flow events do not get dominated by idiosyncratic boundary artifacts.

\indent We now turn to the economic implication of Theorem \ref{th::power_law_price_impact}. First, observe that the asymptotic shape of marginal costs (i.e., marginal prices) is independent of the scale parameter $\sigma$ of noise traders, which controls the distribution's width, but it depends on the shape parameter $\nu$, which determines the ``heaviness" of the tails. This finding is interpreted as being symptomatic of the tail expectation condition that defines the LOB and the fact that noise demand $Z_t \overset{d}{\sim} \mathtt{T}_{\nu}(0,\sigma)$ with $\nu>2$. In particular, as $\nu$ increases, the distribution is concentrated mostly around two standard deviations in a neighborhood of zero, and so the informed-demand tail dominates the noise tail under the stated exponent inequality. Instead, for small values of $\nu$ large trades can come also from noise traders, in addition to informed traders. In other words, the scale parameter $\sigma$ merely rescales the magnitude of noise shocks, while $\nu$ determines whether moments (e.g., variance) are finite and thus whether extreme noise realizations affect equilibrium prices asymptotically. Extreme realizations of noise traders obscure the signal from insiders, reducing the sensitivity of the expected fundamental value to aggregate order flow. As a consequence, the regular variation indices are larger, reflecting fatter tails in the price impact function.  Informed traders can optimally hide their trades within these large noise realizations, extracting higher expected profits without immediately revealing information about $V$. Second, by plugging \eqref{eq::Psi_t_plus} into $\rho_t^{+}=\frac{\partial_x\Psi_t^{+}(M)-1}{1-\frac{\partial_x\Psi_t^{+}(M)}{N_t}}$, we obtain that
\begin{equation*}
    \rho_t^{+} = -\frac{1}{1+\frac{N_t-1}{N_t}\alpha_t}\,\,\text{with}\,\,\alpha_t:=1-L-(\nu+1)\sum\limits_{s=1}^{t-1}\frac{1}{\rho_{s}^{+}}\,\,\,(\alpha_0:=1-L),
\end{equation*}
a writing that evidences the dependence of the market impact exponent on the upper tail behaviour of the fundamental asset value $V$ (via $L$), the distribution of the aggregated market orders of noise traders (via $\nu$), and the ``realized" market impact exponents up until $t$; notice that $\alpha_t=\alpha_{t-1}-\frac{(\nu+1)}{\rho_{t-1}^{+}}$, so that $\alpha_t>\alpha_{t-1}$ since $\rho_{t-1}^{+} \in (-1,0)$. Furthermore, for a fixed $t$, the price impact decreases with an increase in the number of insiders $N_t$ due to increased competition among them.

\indent We comment now the fact that the price impact of a large buy market order depends on the distribution of the fundamental asset value $V$ via its upper tail behaviour, through the constant $L$, and not on the shape of the distribution elsewhere. The condition in \eqref{eq::tail_behaviour} requires that the density of $V$ decays smoothly near the right endpoint of the support by ensuring that extreme realizations do not dominate the asymptotic behaviour of the marginal price. From an economic point of view, condition \eqref{eq::tail_behaviour} guarantees that while large values of $V$ are possible, their probability decays at a controlled rate. As a result, the equilibrium pricing mechanism is not unduly affected by rare, extreme payoffs, allowing for a well-defined asymptotic marginal price. For example, both power-law distribution (i.e., Pareto-type tails) and sub-Gaussian distribution (e.g., truncated normal) satisfy the required regularity conditions. In the first case, near the upper bound $M$, $\mathfrak{p}_{V}(x) \propto (M-x)^{\alpha-1}$ for some $\alpha>0$ and $L=-(\alpha-1)$. In the second case,  $\mathfrak{p}_{V}(x) \propto \exp[-c (M-x)^{\beta}]$ for some $\beta>0$ and $L=0$. This shows that our modeling approach is quite flexible and our equilibrium construction applies to a broad class of fundamental value distributions commonly used in asset pricing and market microstructure models. 
\begin{remark}
    If $\lim_{x\rightarrow M}(M-x)\frac{\partial_x\Pi_t^+(x)}{\Pi_t^+(x)}=\pm\infty$, according to \eqref{eq::derivative_rewriting}, $\partial_x\Psi_t^+(M)=1$. As Theorem B.2 in \cite{ccetin2023power}, in any period $t$, if there exists an integer $n\geq1$ and a real constant $k\in(0,\infty)$ such that $\lim_{x\rightarrow M}\frac{\Psi_t^+(x)-x}{(M-x)^{n+1}}=\frac{1}{k}$, then the following logarithmic asymptotics price impact holds:
    \begin{equation*}
        M-F_t(x,t,Y^{t-1})\sim \left(\frac{N_t}{N_t-1}\frac{n}{k}\right)^{-\frac{1}{n}}(\log x)^{-\frac{1}{n}},\quad\text{as $x\rightarrow\infty$.}
    \end{equation*}
    From \eqref{eq::first_equation_updated}, a sufficient condition is
    $\lim_{x \rightarrow M }(M-x)\frac{\mathfrak{p}_{V}^{'}(x)}{\mathfrak{p}_{V}(x)}=\pm\infty$. This condition is satisfied, for example, when $\mathfrak{p}_{V}(x) \propto \exp[-c/(M-x)^{\beta}]$ for some $\beta>0$.
    
\end{remark}
Under the assumptions of Theorem \ref{th::power_law_price_impact}, marginal costs $F(\cdot, t, Y^{t-1})$ and LOB $h(\cdot, t, Y^{t-1})$ behave similarly for large values, in the sense that $M-h(\cdot, t, Y^{t-1})$ is regularly varying at $+\infty$ with an index equal to the one of $M-F(\cdot, t, Y^{t-1})$. This follows from the same argument as in the proof of Theorem \ref{th::power_law_price_impact}, \eqref{eq::trick_proof}. Using the same definitions, we have
\begin{equation*}
    \phi_{F(\cdot, t, Y^{t-1})}(x)=\int_{-\infty}^{+\infty} \Psi_t^{+}(F(x y,t,Y^{t-1}))\Lambda_t(x,1,\ud y),
\end{equation*}
where $\Lambda_t(x,1,\ud y)$ converges to the point mass at $1$ as $x \rightarrow \infty$. Then, the mean value theorem, the continuity of $\partial_x \Psi_t^{+}(\cdot)$, and the fact that $F(\cdot, t, Y^{t-1})$ is regularly varying at $+\infty$ with index $\rho_t^{+}$ allows us to conclude that 
\begin{equation*}
    \lim_{x \rightarrow \infty}\frac{M-h(x, t, Y^{t-1})}{M-F(x, t, Y^{t-1})}=\lim_{x \rightarrow \infty}\frac{M-\phi_{F(\cdot, t, Y^{t-1})}(x)}{M-F(x, t, Y^{t-1})}=\partial_x \Psi_t^{+}(M).
\end{equation*}
The previous identity shows that, in the far tails, the price is asymptotically a fixed multiple of the marginal cost. Hence, the asymptotic geometry of the book is already encoded at the level of the marginal-cost fixed-point. Economically,
this means that large-order prices inherit their tail shape directly from the strategic informed demand schedule, rather than generating an independent asymptotic regime of their own.

Theorem \ref{th::power_law_price_impact} provides more precise information about the distribution of the traded volume in equilibrium, via the following corollary.

\begin{corollary}\label{coroll::volume_varying}
In any period $t$, Assume $\partial_x \Psi_t^{+}(M)$ and $\partial_x \Psi_t^{-}(m)$ exists, $N_t>1$, and that $\operatorname{supp}(V)=[m,M]$ with $-\infty<m<M<+\infty$. Then
\begin{itemize}
    \item If $M-F\left(x,t,Y^{t-1}\right)$ is regularly varying of index $\rho^{+}_t$ at $\infty$, then $\Pi^{+}_t\left(F\left(x,t,Y^{t-1}\right)\right)$ is regularly varying of index $\frac{\partial_x\Psi_{t}^{+}\left(M\right)}{1-\partial_x\Psi_{t}^{+}\left(M\right)}\rho^{+}_{t}$ at $\infty$.
    \item If $F\left(x,t,Y^{t-1}\right)-m$ is regularly varying of index $\rho^{-}_t$ at $-\infty$, then $\Pi_t^{-}\left(F\left(x,t,Y^{t-1}\right)\right)$ is regularly varying of index $\frac{\partial_x\Psi_{t}^{-}\left(m\right)}{1-\partial_x\Psi_{t}^{-}\left(m\right)}\rho^{-}_{t}$ at $-\infty$.
\end{itemize}
\end{corollary}
\begin{proof}
    See Appendix \ref{app::price_impact_proof}, Subsection \ref{app::volume_varying}.
\end{proof}

\noindent In addition, notice that we can write (a similar argument holds true also for the bid-side)
\begin{equation*}
    \mathbb{P}_t(X_{t}^{*}>x) = P_t(F^{-1}(V,t,Y^{t-1})>x) = P_t\left(V>F(x,t,Y^{t-1})\right)
= \Pi^{+}_t\left(F(x,t,Y^{t-1})\right),
\end{equation*}
which is (Corollary \ref{coroll::volume_varying}) regularly varying at $+\infty$ of index 
\begin{equation*}
    \frac{\partial_x\Psi_{t}^{+}\left(M\right)}{1-\partial_x\Psi_{t}^{+}\left(M\right)}\rho^{+}_{t}.
\end{equation*}
\noindent  Moreover, the aggregate order flow $Y^{*}_{t} = X^{*}_{t} + Z_{t}$ is also regularly varying at $+\infty$ of index $\frac{\partial_x\Psi_{t}^{+}\left(M\right)}{1-\partial_x\Psi_{t}^{+}\left(M\right)}\rho^{+}_{t}$. In fact, we have 
\begin{equation*}
\begin{split}
    \lim_{\alpha \rightarrow\infty} \frac{\mathbb{P}_t(Y^{*}_t > \alpha y)}{\mathbb{P}_t(Y^{*}_t > \alpha)}
&= \lim_{\alpha \rightarrow\infty} \frac{\int_{-\infty}^{\infty} \Pi^{+}_t(F(z,t,Y^{t-1}))\, \mathfrak{q}_{\nu}\left(\alpha y-z;0, \sigma\right)\,\ud z}{
\int_{-\infty}^{\infty} \Pi^{+}_t(F(z,t,Y^{t-1}))\,\mathfrak{q}_{\nu}\left(\alpha-z;0, \sigma\right)\ud z} \\
&= \lim_{\alpha \rightarrow\infty}
\frac{
\int_{-\infty}^{\infty} \Pi^{+}_t(F(\alpha z,t,Y^{t-1}))\, \mathfrak{q}_{\nu}\left(y-z;0, \frac{\sigma}{\alpha},\right)\, \ud z
}{\int_{-\infty}^{\infty} \Pi^{+}_t(F(\alpha z,t,Y^{t-1}))\, \mathfrak{q}_{\nu}\left(1-z; 0, \frac{\sigma}{\alpha}\right)\, \ud z} \\
&= \lim_{\alpha \rightarrow\infty}\frac{\Pi_t^+(F(\alpha y, t, Y^{t-1}))}{\Pi_t^+(F(\alpha, t, Y^{t-1}))},
\end{split} 
\end{equation*}
where we used the scaling property of the Student-$t$ density in the second equality. Hence, $Y_t^{\star}$ inherits the regular variation from the insider demand $X_t^{\star}$ with the same exponent $\frac{\partial_x\Psi_{t}^{+}\left(M\right)}{1-\partial_x\Psi_{t}^{+}\left(M\right)}\rho^{+}_{t} < \nu$. In what follows, we set $\alpha_t^{+}:=\frac{\partial_x\Psi_{t}^{+}(M)}{1 - \frac{\partial_x\Psi_{t}^{+}(M)}{N_t}}$ and write that $Y_t^{*}$ and $X_t^{\star}$ are regularly varying at $+\infty$ with index $-\alpha_t^{+}$. Notice that since $0\le \partial_x\Psi_t^+(M)<1$ and $N_t>1$, one has $\alpha_t^+<\frac{N_t}{N_t-1}\le 2<\nu$. 

\begin{remark}
\label{rem:rare-shocks}
Student-$t$ noise makes large order imbalances substantially more ambiguous than
in the Gaussian benchmark over the economically relevant range of depths. What the previous discussion shows is that, in the bounded-support branch regime, this ambiguity is not asymptotically permanent: because the induced informed-demand tail is heavier than the Student-$t$ noise tail, sufficiently deep ask or bid executions are eventually dominated by informed trading. Therefore, rare liquidity shocks determine the so-called pre-asymptotic range and the crossover scale to the asymptotic informed regime, and it is precisely this extended ambiguity region that generates flatter impact and slower learning under heavy-tailed noise.
\end{remark}

The following proposition formalizes the previous remark; again, a similar argument holds true also for the bid-side.
\begin{proposition}
\label{prop:tilted-posterior}
Suppose the hypotheses of Theorem \ref{th::power_law_price_impact} hold. For $y>0$, we define the ask-side tilted posterior of branch-induced informed demand by
\begin{equation*}
    \pi_{y,t}(A):=
\mathbb P_t\!\big(X_t^\ast\in A \,\big|\, X_t^\ast+Z_t\ge y\big), \qquad
A\in\mathcal B(\mathbb R).
\end{equation*}
Then $\pi_{y,t}([y,\infty))\rightarrow 1$ as $y\to+\infty$. Equivalently, sufficiently deep ask executions are asymptotically generated by large informed
demand rather than by extreme noise.
\end{proposition}
\begin{proof}
    See Appendix \ref{app::price_impact_proof}, Subsection \ref{app::tilted-posterior}.
\end{proof}

We now state the following corollary, whose proof follows directly from Proposition \ref{prop:tilted-posterior} and the fact that we are working with a strictly increasing $F(\cdot, t, Y^{t-1})$.

\begin{corollary}
\label{cor:tail-price-reduction}
Under the assumptions of Proposition~\ref{prop:tilted-posterior}, we have
\begin{equation*}
h(y,t,Y^{t-1})
=
\mathbb E_t[V\mid X_t^\ast+Z_t\ge y]
=
\mathbb E_t[V\mid X_t^\ast\ge y]+o(1),
\quad y\to+\infty
\end{equation*}
Since $X_t^\ast\ge y \Longleftrightarrow  V\ge F(y,t,Y^{t-1})$,it follows that
\begin{equation*}
h(y,t,Y^{t-1})=\Psi_t^+ \big(F(y,t,Y^{t-1})\big)+o(1),
\quad y\to+\infty.
\end{equation*}
\end{corollary}

In the remaining part of the present section, we show that branch monotonicity can be recovered in the tails even though global monotonicity preservation fails. We start with the following lemma.

\begin{lemma}\label{lem:ask-derivative}
For $y>0$, we define $r_\nu(u):=\frac{\mathfrak{q}_\nu(u;0,\sigma)}{\mathbb P(Z_t\ge u)}$. Then
\begin{equation*}
 \partial_y h(y,t,Y^{t-1})=-\operatorname{Cov}_{\pi_{y,t}}\!\big(V,r_\nu(y-X_t^\ast)\big)   
\end{equation*}
\end{lemma}
\begin{proof}
    See Appendix \ref{app::price_impact_proof}, Subsection \ref{app::ask-derivative}.
\end{proof}

We then state and prove the following proposition.

\begin{proposition}
\label{prop:tail-gap}
Suppose the hypothesis of Theorem~\ref{th::power_law_price_impact} hold. Assume in addition that $X_t^\ast$ admits an eventually monotone density $\mathfrak p_{t,X^\ast}(\cdot\mid Y^{t-1})$, and that $M-F(\cdot,t,Y^{t-1})$ is regularly varying at $+\infty$ with index $\rho_t^+\in(-1,0)$. Then, as $y\to+\infty$,
\begin{equation*}
-\operatorname{Cov}_{\pi_{y,t}}\!\big(V,r_\nu(y-X_t^\ast)\big)
\sim
\frac{\alpha_t^+(-\rho_t^+)}{\alpha_t^+-\rho_t^+}
\frac{M-F(y,t,Y^{t-1})}{y}
\end{equation*}
In particular, $\partial_y h(y,t,Y^{t-1})>0$ for all sufficiently large $y$. The analogous left-tail statement also holds, with $m$, $\rho_t^-$, $\alpha_t^-$, and the lower-tail hazard in place of the right-tail quantities.
\end{proposition}
\begin{proof}
See Appendix~\ref{app::Supporting_Results_Market_Impact},
Subsection~\ref{subsec:tail-gap-proof}. 
\end{proof}
The previous proposition identifies why branch monotonicity can be
recovered in the tails, even though global monotonicity preservation fails. Indeed, in the far ask tail, the execution event $\{X_t^\ast+Z_t\ge y\}$ is asymptotically generated by large informed demand, so the tilted posterior (cf. Proposition \ref{prop:tilted-posterior}) concentrates on states in which $X_t^\ast$ already lies above the threshold. In that regime, the ask-side hazard of the Student-$t$ noise is sampled on the monotone part
of its support, and the covariance formula of Lemma~\ref{lem:ask-derivative} acquires a strictly negative sign. Economically, the heavy-tailed noise continues to matter -- it pushes the
crossover farther out and enlarges the ambiguity region -- but it no longer dominates the ultimate branch tail. We conclude this section with the following corollary whose proof is immediate.
\begin{corollary}
\label{cor:eventual-tail-monotonicity}
Suppose the hypothesis of Proposition~\ref{prop:tail-gap} hold. Then, there exists $Y_t^\ast>0$ such that $\partial_y h(y,t,Y^{t-1})>0$ for every $y\ge Y_t^\ast$. Similarly, there exists $\underline Y_t^\ast>0$ such that $\partial_y h(y,t,Y^{t-1})>0$ for every $y\le -\underline Y_t^\ast$.
\end{corollary}

\section{Numerical Studies}\label{section::Numerical_Studies}
This section is devoted to the quantitative implications of the model obtained in the previous sections. More precisely, we investigate the market equilibrium (cf. Definition~\ref{def::equilibrium}) implied by the fixed-point characterisation in ~\eqref{eq::fixed_point_equation} and illustrate the liquidity suppliers' learning in Section~\ref{section::The Liquidity_Suppliers_Belief_posterior_consistency}, and the market dynamics and market impact asymptotics in Section~\ref{section::Asymptotic_Price_Impact}. 

The proofs of the statements in the main text are based on a boundedness assumption on the fundamental asset value $V$. Hence, in this section we shall also present results for unbounded signals, after a truncation on a sufficiently wide grid. Therefore, we consider both bounded- and unbounded-support distributions. Bounded-support experiments include the following distributions: (i) the uniform distribution on $[m,M]$ with $0<m<M<+\infty$, $\mathtt{Unif}([m,M])$; (ii) the Beta distribution with real parameters $a$ and $b$ on $[0,1]$, $\mathtt{Beta}(a,b)[0,1]$; (iii) the truncated Gaussian distribution on $[m,M]$ with mean $\mu \in \mathbb{R}$ and standard deviation $\sigma > 0$, $\mathtt{Trunc-Gaussian}(\mu, \sigma^2, [m,M])$. Instead, unbounded-support experiments include: (a) the Gaussian distribution with mean $\mu \in \mathbb{R}$ and standard deviation $\sigma > 0$, $\mathtt{Gaussian}(\mu, \sigma^2)$; (b) The Student-$t$ distribution with mean $\mu \in \mathbb{R}$, shape $\sigma \in \mathbb{R}$, and degrees of freedom $\nu>0$, $\mathtt{Student}-t(\nu, \mu, \sigma)$ (or $\mathtt{Student}-t(\nu, \sigma)$ if $\mu=0$); (c) the Pareto distribution with shape parameter $\alpha>0$, $\mathtt{Pareto}(\alpha)$.

Our numerical experiments cover the following four aspects. First, we verify the uniqueness and monotonicity of the solution of the Student-$t$ fixed-point in \eqref{eq::fixed_point_equation} across the range of the above prior distribution, different initializations for the numerical fixed-point algorithm, and for different choices of the number of informed traders $N_t$. Second, we illustrate finite-horizon posterior concentration. Third, we investigate the equilibrium dynamics of the marginal cost function $F(\cdot, t, Y^{t-1})$, the corresponding marginal price function $h^{\star}(\cdot, t, Y^{t-1})$, and of the bid-ask spread; for the sake of notation, in what follows, we use simply $F^{\star}_t$ and $h^{\star}_t$. Finally, we verify the regular-variation tail prediction of the price impact predicted by Theorem~\ref{th::power_law_price_impact}; also, we study the sensitivity of these results to the parameter $\nu$, the number of insiders $N_t$, and the number of trading periods. Before presenting the four analyses, we describe our numerical setup.\\
\noindent\textbf{Numerical setup.} We fix a grid for the order size, $x_i=-X_{\max} + (i-1)\Delta x$, $i=1,\dots,n_x$, with $n_x$ odd so that $x_{(n_x+1)/2}=0$.  We represent the marginal cost function $F(\cdot, t, Y^{t-1})$ and the marginal price function $h(\cdot, t, Y^{t-1})$ on this grid. For beliefs, we use a grid for fundamentals $v_j = v_{\min} + (j-1)\Delta v$, $j=1,\dots,n_v$, where $[v_{\min},v_{\max}]$ is chosen so that the prior mass outside this interval is negligible and posterior mass remains negligible throughout the simulation. All integrals in $x$ are approximated by the trapezoidal rule on the $x$-grid. At the beginning of period $t$, liquidity suppliers hold a posterior density $\mathfrak{p}_{t-1,V}(\cdot)$. On the $v$-grid we represent it by weights $\pi_{t-1,j}\approx \mathfrak{p}_{t-1,V}(v_j)\Delta v$ with $\sum\limits_{\substack{j}} \pi_{t-1,j}=1$. We then pre-compute the tail probability and tail first-moment functional needed in \eqref{eq::g_mappings}; for values $y$ not on the $v$-grid we use linear interpolation of these arrays in $y$. 
Let $\mathfrak{q}_\nu(\cdot;\,0,\sigma)$ be the Student-$t$ density. Integrals of the form $x \mapsto \int_{-\infty}^{\infty} \mathfrak{q}_\nu(x-z;\,0,\sigma)\,g(z)\,\mathrm{d}z$ are approximated by discrete convolution on the $z$-grid. Given beliefs at the start of period $t$, we solve \eqref{eq::fixed_point_equation} by fixed-point iteration on the $x$-grid, evaluating the map $\phi_F$ in \eqref{eq::g_mappings} and \eqref{eq::equation_for_g} by numerical quadrature at each step.  The kernel $\bar{\mathfrak{q}}(x,z; 0, \sigma) = \frac{1}{x}\!\int_0^x \mathfrak{q}_\nu(u-z;\,0,\sigma)\,\mathrm{d}u$ is precomputed once via the closed form $\bar{\mathfrak{q}}(x,z; 0, \sigma) = \bigl[F_\nu(x-z;\,0,\sigma) - F_\nu(-z;\,0,\sigma)\bigr]/x$ for $x\neq 0$, where $F_\nu(\cdot;\,0,\sigma)$ denotes the Student-$t$ cumulative distribution function.  We iterate until $\|F^{(k+1)}_t-F^{(k)}_t\|_\infty < \texttt{tol}$, and set $h(\cdot,t,Y^{t-1}) \equiv \phi_{F_t}(\cdot)$ upon convergence. Finally, given the fixed-point solution $F(\cdot,t,Y^{t-1})$, we compute $F^{-1}(\cdot,t,Y^{t-1})$ by linear interpolation of $\{(x_i,F(x_i))\}$. 
Under value $v$, the model implies $Y_t \mid (Y^{t-1},V=v)\sim\mathtt{T}_{\nu}(F^{-1}(v,t,Y^{t-1}),\sigma)$, so the Bayes update is $\pi_{t,j} \propto \pi_{t-1,j}\cdot \mathfrak{q}_\nu(Y_t - F^{-1}(v_j,t,Y^{t-1});\,0,\sigma)$, $j=1,\dots,n_v$, normalised so that $\sum_j \pi_{t,j}=1$. In summary, at each period $t$, we
(a) solve for $F(\cdot,t,Y^{t-1})$ given current beliefs,
(b) compute $X_t^\star=F^{-1}(v_0,t,Y^{t-1})$,
(c) draw $Z_t\sim \mathtt{T}_\nu(0,\sigma)$ and set $Y_t=X_t^\star+Z_t$, and
(d) update beliefs.

{\color{black} In the following implementation, we set $X_{\text{max}}=10,\,\Delta x=0.1,$ and we use a $z$-grid on $[-30,30]$ with step size $\Delta z=0.1$ for convolution. The fundamental value grid uses $\Delta v=0.001$, with bounds chosen for each prior (e.g. $[0,1]$ for uniform and beta priors, $[-3,3]$ for the truncated Gaussian, $[-8,8]$ for the Gaussian, and $[-40,40]$ for the Student-$t$ prior). The noise scale is fixed at $\sigma=0.1$, and the marginal cost $F_t^\star$ is computed as the fixed point of~\eqref{eq::fixed_point_equation} via iterative updates $F_t^{(k+1)}=\mathcal{T}(F_t^{(k)})$, terminated when the sup-norm criterion $\|F_t^{(k+1)}-F_t^{(k)}\|_\infty \le 10^{-10}$ is satisfied. Unless otherwise stated, these choices are used throughout all numerical experiments.}

\vspace{0.35cm}
We start with the uniqueness and monotonicity of the solution of the Student-$t$ fixed-point in \eqref{eq::fixed_point_equation}, and the corresponding $\phi_F(\cdot)$ in \eqref{eq::equation_for_g}. It is sufficient to examine the problems at a fixed trading period $t$ because the fixed-point operator has the same structural form at each trading period, with the posterior belief replacing the prior. Starting with an $F^{(0)}$, Fig. 1, \emph{Left panel}, displays the iterates $F_1^{(k+1)}=\mathcal{T}_{1,\nu}F_1^{(k)}$ for $k\in\{1,\ldots,8\}$ when $V \overset{d}{\sim} \mathtt{Unif}([0,1])$, $N_1=2$, and $Z_1 \overset{d}{\sim} \mathtt{T}_{3}(0,0.1)$; 
we also display the final $F_1$, denoted by $F_1^{*}$. Despite the fact the initial iterates do not produce a monotone $F_1$, the final $F_1^{*}$ is strictly increasing. Fig. 1, \emph{Central panel}, confirms the same pattern for $h_1^{(k)}=\phi_{F_1^{(k)}}$; the \emph{Right panel} displays the quantity $\|F^{(k+1)}-F^{(k)}\|_\infty$ as a function of the iterations. The figures are clear in indicating that $(h_1^*, X_1^*)$ can be considered an equilibrium as in Definition \ref{def::equilibrium}.

\begin{figure}[htbp]
    \centering
    \includegraphics[width=\textwidth]{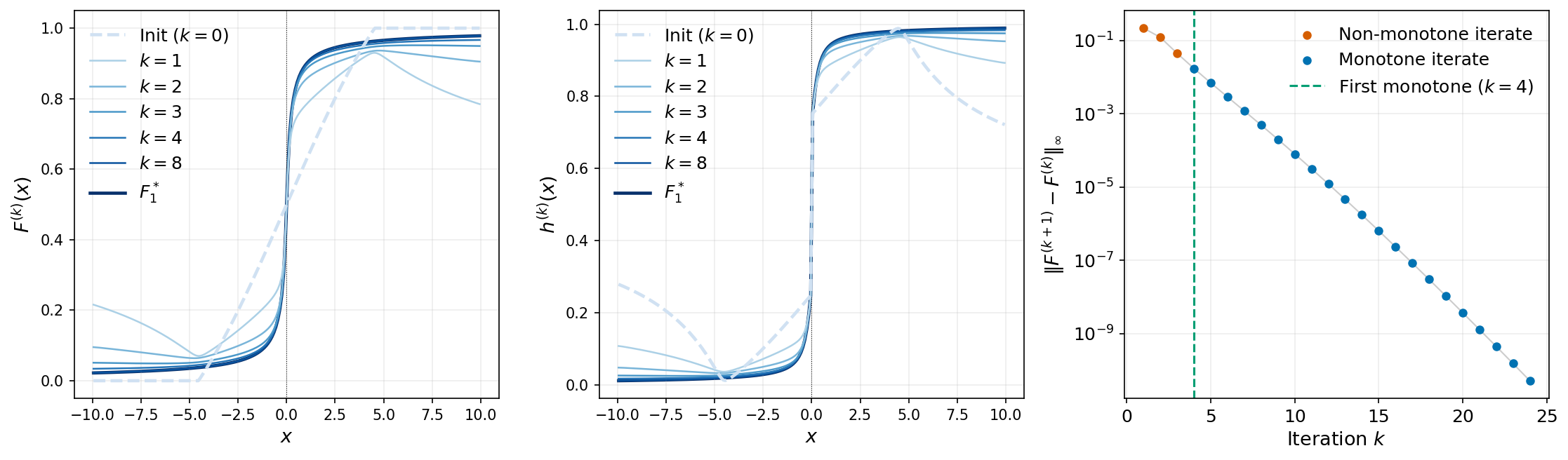}
    \caption{\emph{Left panel}: Fixed-point iterates $F_1^{(k+1)}=\mathcal{T}_{1,\nu}F_1^{(k)}$ for $k\in\{0,\ldots,8\}$ together with $F_1^{*}$ when $V \overset{d}{\sim} \mathtt{Unif}([0,1])$, $N_1=2$, and $Z_1 \overset{d}{\sim} \mathtt{T}_{3}(0,0.1)$. \emph{Central panel}: Corresponding iterates of $h_1^{(k)}=\phi_{F_1^{(k)}}$; \emph{Right panel}: $\|F^{(k+1)}-F^{(k)}\|_\infty$ as a function of the iterations $k$.}
    \label{fig::nonmono_to_mono}
\end{figure}
Fig. \ref{fig::nonmono_to_mono} is, however, conservative. Therefore, we provide numerical evidence of a common
equilibrium branch as in Definition \ref{def::equilibrium} by varying the initialization $F^{(0)}$ of the fixed-point numerical algorithm; we use eight different initializations, four monotone and four non-monotone. The monotone ones are hyperbolic tangent functions at three scales {\color{black}, \(F^{(0)}(x;\ell)
=m + (M - m)\,\frac{1 + \tanh(x/\ell)}{2},
\,
\ell \in \{\sigma,\, 2\sigma,\, 4\sigma\}\), where $\sigma=0.1$ is the noise trade scale in $Z_t\sim \mathtt{T}_\nu(0,\sigma)$,} and a linear function scaled on the fundamental value's support $[m,M]${\color{black}, \(F^{(0)}(x)=\min\bigl\{\,\max\{\tfrac{m+M}{2} + 0.1\,x,\, m\},\, M\bigr\}.\)} Instead, the non-monotone ones are oscillatory perturbations of hyperbolic tangent functions {\color{black} around the monotone initialisation $F^{(0)}(x;
2\sigma)$, that is, \(F^{(0)}(x;\sigma,\varepsilon,\omega)
=
\min\bigl\{\,\max\{F^{(0)}(x;2\sigma) + \varepsilon \sin(\omega x),\, m\},\, M\bigr\},\) with the choice \((\varepsilon,\omega)\in\{(0.02,0.7),\, (0.04,1.4),\, (0.04,2.8),\, (0.08,0.9)\}\)} . In addition, we use seven different prior distributions for the fundamental asset value ($\mathtt{Unif}([0,1])$, $\mathtt{Beta}(2,2)[0,1]$, $\mathtt{Beta}(1/2,1/2)[0,1]$, $\mathtt{Trunc-Gaussian}(0,1,[-3,3])$, $\mathtt{Gaussian}(0,1)$, $\mathtt{Student}-t(5,1)$, $\mathtt{Pareto}(3)$), and four different choices for the number of insider trades ($N\in\{2,3,5,10\}$). Fig.~\ref{fig::uniqueness_summary} shows the maximum pairwise sup-norm distance $\max_{a\neq b}\|F_{1}^{a}-F_{1}^{b}\|_\infty$ across all seven prior distributions and $N\in\{2,3,5,10\}$, where for each specification the maximum is taken over the $\binom{8}{2}=28$ pairs of initialisations; the tolerance threshold is $10^{-6}$. The figure provides numerical evidence for branch uniqueness for the Student-$t$ fixed-point operator.  Importantly, Fig. \ref{fig::diag_monotonicity} clearly indicates that the fixed point solution $F$ and the associated $h=\phi_F$ are also monotone. It shows the minimum finite-difference slopes $\min_i \Delta F_1^\star / \Delta x$ and $\min_i \Delta h_1^\star / \Delta x$ of the final iterates. All values are strictly positive, confirming that both $F_1^\star$ and $h_1^\star$ are strictly increasing, thereby showing that the selected fixed point has the economic interpretation of a monotone LOB equilibrium in these specifications $(h_{1}^\star,X_{1}^\star)$.

\begin{figure}[htbp]
    \centering
    \includegraphics[width=0.7\textwidth]{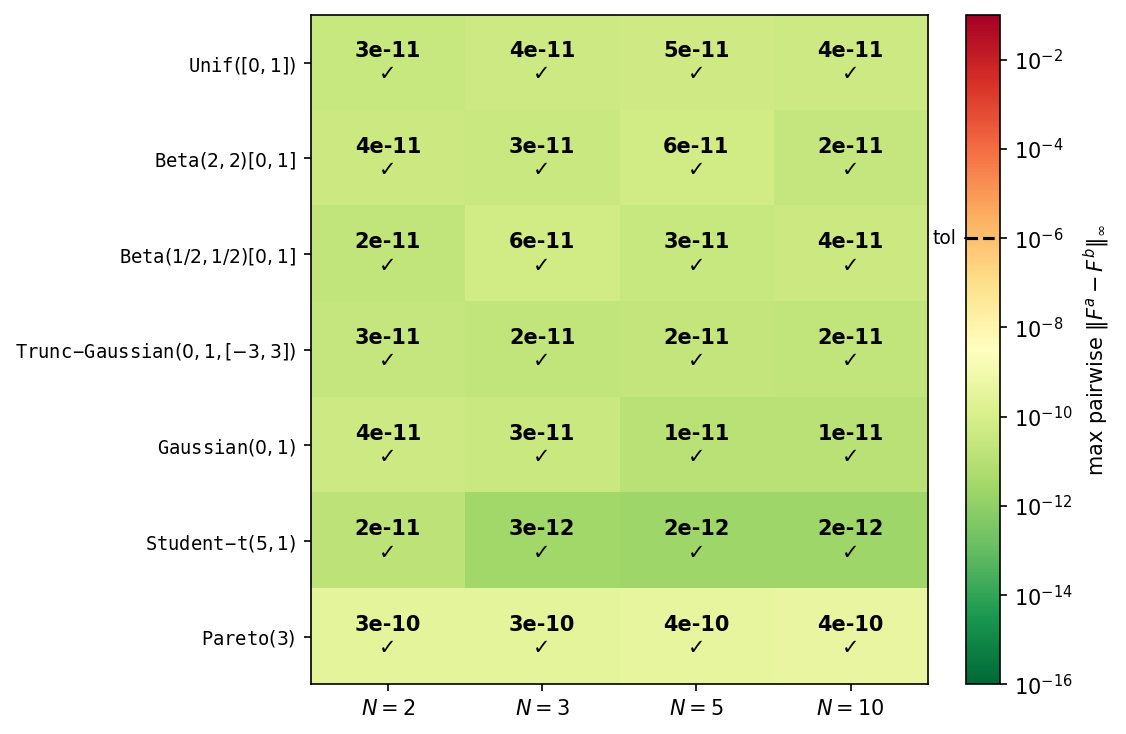}
    \caption{Maximum pairwise sup-norm distance $\max_{a\neq b}\|F_{1}^{a}-F_{1}^{b}\|_\infty$ 
    across all $\binom{8}{2}=28$ pairs of converged solutions, computed separately for each prior distribution 
    and $N\in\{2,3,5,10\}$.}
    \label{fig::uniqueness_summary}
\end{figure}
\begin{figure}[htbp]
    \centering
    \includegraphics[width=1.0\linewidth]{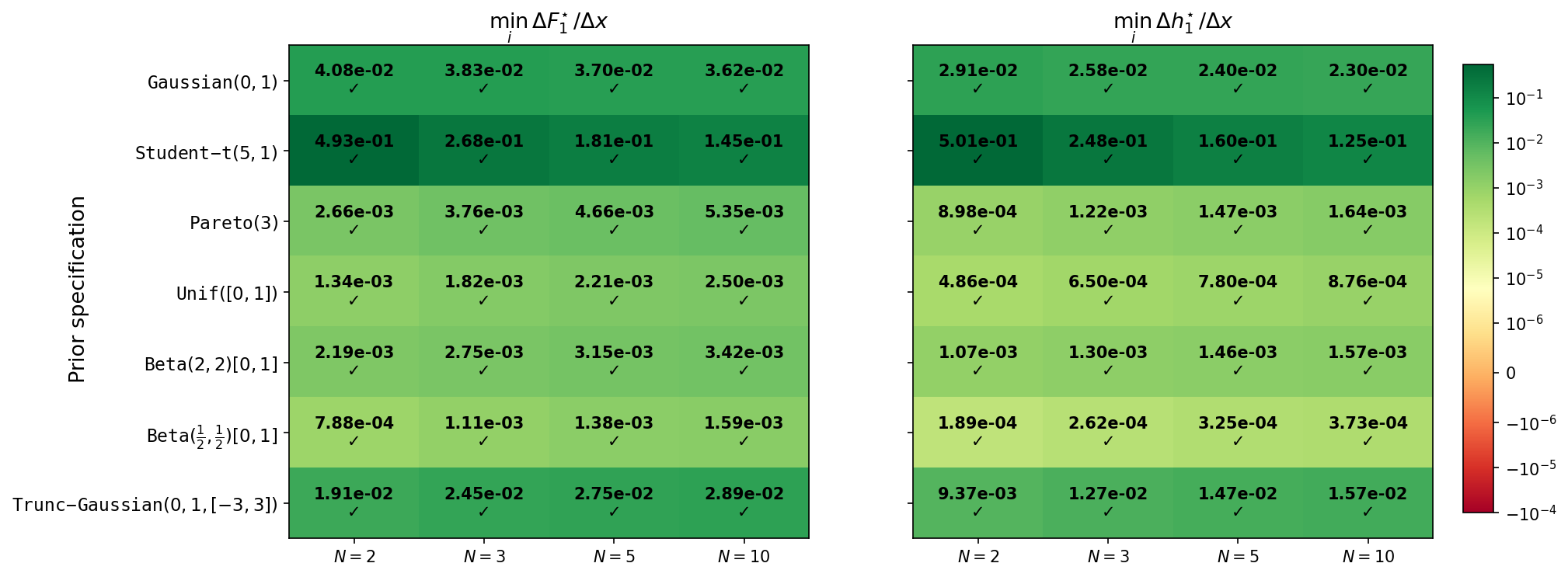}
    \caption{Minimum finite-difference slope of the $F_1^\star$ (\emph{Left panel}) and $h_1^\star$ (\emph{Right panel}), across 
    seven prior distributions and four choices for the number of insider traders.}
    \label{fig::diag_monotonicity}
\end{figure}
 
We now illustrate finite-horizon posterior learning.
{\color{black} As in the previous experiment, we employ the same seven priors for fundamental value $V$ and numbers of informed traders \(N\in\{2,3,10\}\), yielding a total of $7\times3=21$ specifications. For each specification, we simulate ten independent realizations of the uninformed aggregate order-flow distributed as a $\mathtt{Student}-t(3,0.1)$. Within each realization, the posterior distribution is updated via~\eqref{eq::beliefsupdate}, using the equilibrium marginal cost $F_t^\star$ at each trading period. 
For each specification, the posterior mean and standard deviation at each trading period are then computed by averaging across the ten independent realizations.
}
Fig.~\ref{fig::posterior_consistency} shows the posterior mean $\mathbb{E}[V \mid Y^{t-1}]$ and the posterior standard deviation $\sqrt{\mathrm{Var}[V \mid Y^{t-1}]}$ across the trading periods $t \in \{1,\ldots,5\}$. 
The figures are clear in indicating that the posterior mean converges to the realization $v_0$ of the fundamental value chosen by nature at time zero, and that posterior dispersion declines over the five simulated periods, hence consistent with the posterior-consistency theorem. In addition, a larger number of informed traders accelerates learning across all prior specifications, as competition among insiders increases the informativeness of aggregate order flow.
\begin{figure}[htbp]
    \centering
     \includegraphics[width=\textwidth]{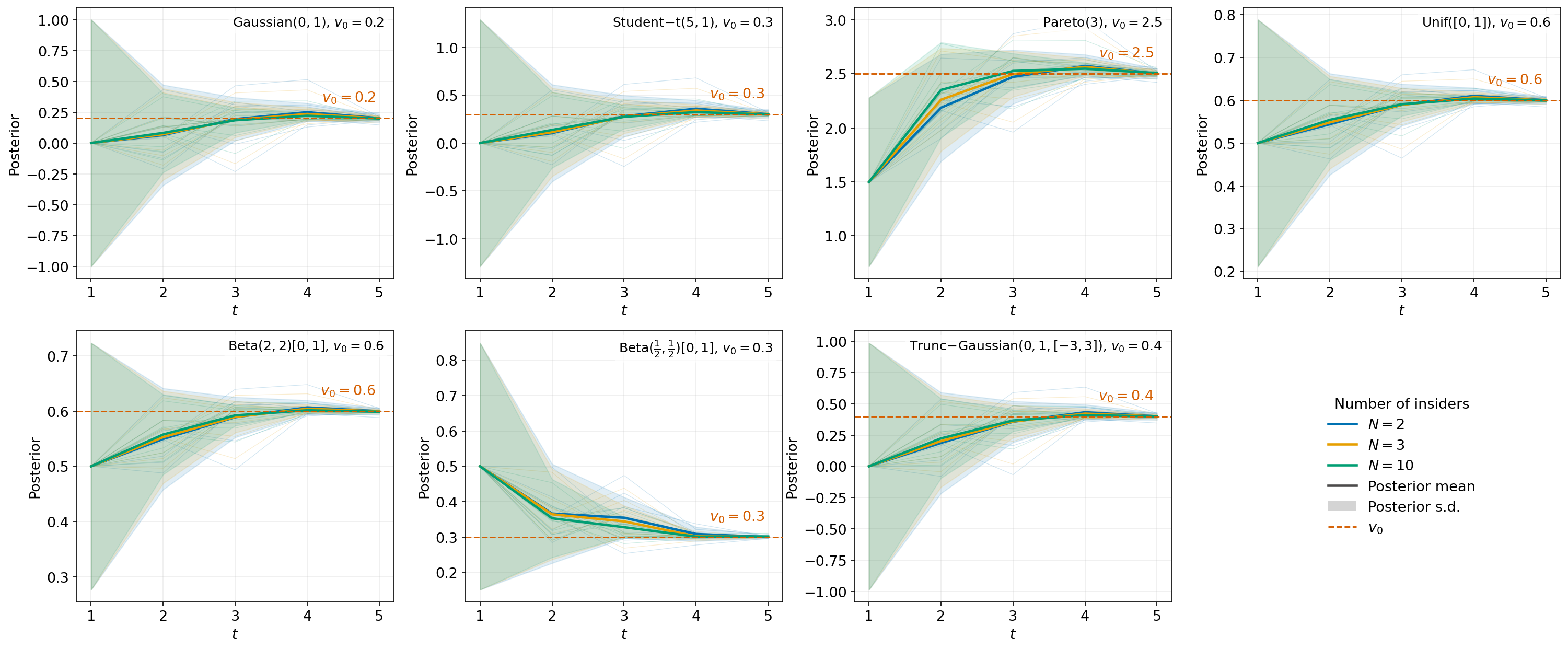}
    \caption{Posterior mean $\mathbb{E}[V \mid Y^{t-1}]$ and standard 
    deviation $\sqrt{\mathrm{Var}[V \mid Y^{t-1}]}$ across the trading periods $t \in \{1,\ldots,5\}$. We employ seven prior distributions and four choices for the number of insider traders, averaged over ten realisations of the noise trader order flow. \emph{Bold lines}: cross-seed average posterior mean; \emph{Bands}: average posterior standard deviation; \emph{Dashed line}: true value $v_0$.}
    \label{fig::posterior_consistency}
\end{figure}

We now investigate the equilibrium dynamics of the marginal cost function $F^{\star}_t$ and the corresponding marginal price function $h^{\star}_t$, and of the bid-ask spread across the trading periods $t \in \{1,\ldots,5\}$ as liquidity suppliers update their beliefs about the fundamental value $V$; we also vary the number of informed traders.  
{\color{black} Numerically, for each specification of the prior and each $N$, we simulate trading periods $t\in\{1,\dots,5\}$. At each period $t$, given the current posterior belief of liquidity suppliers, we compute $F_t^\star$ as the fixed point in~\eqref{eq::fixed_point_equation} via iterative updates, and obtain $h_t^\star=\phi_{F_t^\star}$ from~\eqref{eq::g_mappings} and~\eqref{eq::equation_for_g}. The insider's optimal trade $X_t^\star=F^{-1}(v_0,t,Y^{t-1})$ is computed by linear interpolation of the discretized pairs \(\{x_i,F_t^\star(x_i)\}\) on the $x$-grid, exploiting the monotonicity of $F_t^\star$, where the realization $v_0$ of the fundamental value is chosen at time zero and fixed over trading periods. The observed aggregate order flow is generated as $Y_t=X_t^\star+Z_t$, with $Z_t\sim\mathtt{T}_{\nu}(0,\sigma)$, and the posterior belief is updated accordingly. We track the evolution of $F_t^\star$ and $h_t^\star$ across periods, and compute the bid-ask spread as $h_t^\star(0^+)-h_t^\star(0^-)$, where $h_t^\star(0^-)$ and $h_t^\star(0^+)$ are evaluated from the discretized function on the $x$-grid as the values at the grid points immediately to the left and right of zero, respectively.}
In the main text, we provide such a dynamic for two unbounded distributions of the fundamental asset price: the Gaussian and the Pareto. Results for the other distributions are reported in Appendix \ref{app::numerics_dynamics}. Fig.~\ref{fig::Fh_pareto} and Fig.~\ref{fig::Fh_gaussian} report the equilibrium $(F_t^\star, h_t^\star)$ for the $\mathtt{Pareto}(3)$ and 
$\mathtt{Gaussian}(0,1)$, respectively. Patterns are consistent across both specifications. Specifically, both $F^\star_t$ and $h^\star_t$ are concave at depth, and are monotone in the number of informed traders $N$. Indeed, a larger  $N$ flattens $F_t^\star$ and $h_t^\star$, as increasing competition among insiders induces more aggressive trading, consistent with the decay of $|\rho_t^+|$ with $N$ as in Theorem~\ref{th::power_law_price_impact}. In addition, 
as $t$ increases and liquidity suppliers learn, both $F^\star_t$ and $h^\star_t$ flatten and converge, reflecting the shrinkage of posterior uncertainty and the reduction in adverse selection. 
Both effects are further reflected in the bid-ask spread $h_t^\star(0^+) - h_t^\star(0^-)$, reported in Fig.~\ref{fig::spread_all_priors} across, this time, all seven prior distributions: the spread narrows monotonically with $t$ and with $N$, confirming the joint effects of belief updating and insider competition.

\begin{figure}[htbp]
    \centering
    \includegraphics[width=\linewidth]{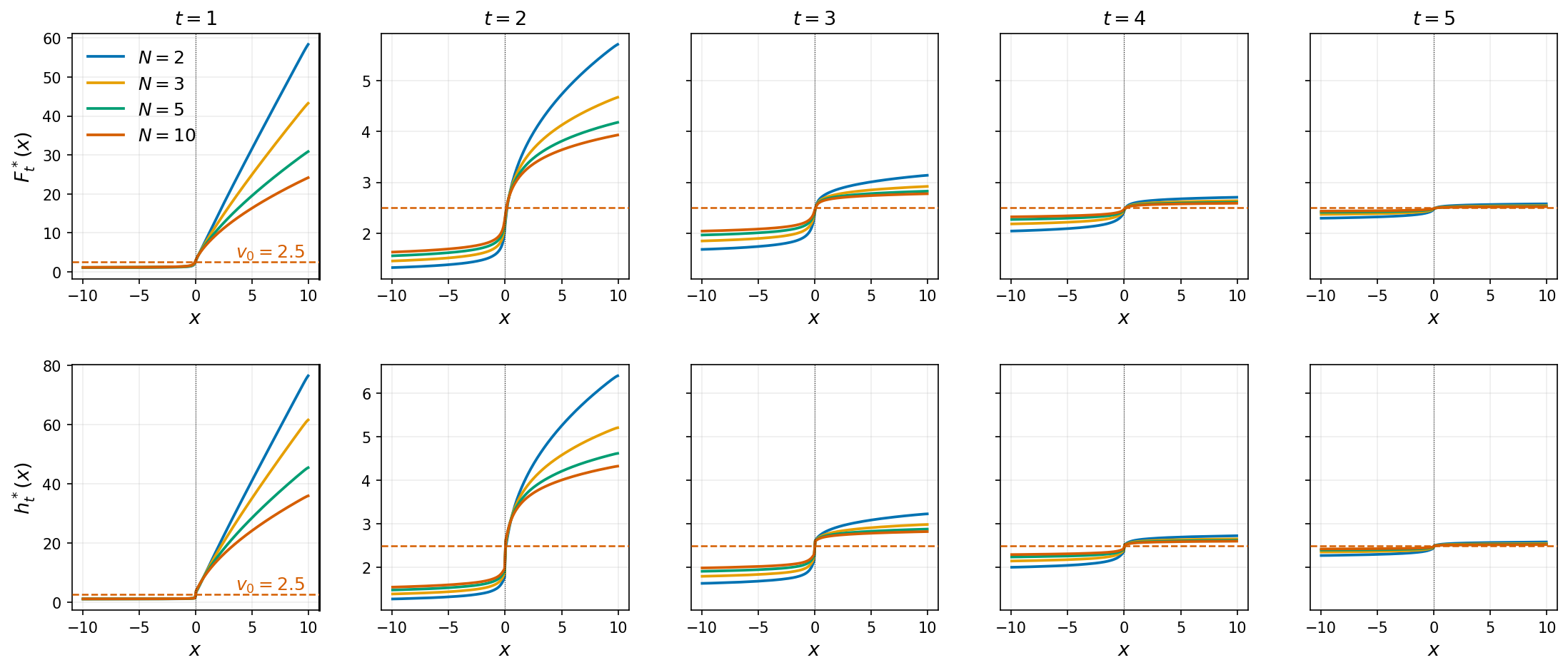}
    \caption{Equilibrium marginal cost$F_t^\star$ (\emph{Top panel}) and limit prices $h_t^\star$ (\emph{Bottom panel}) across the trading periods $t\in \{1,\ldots,5\}$, for the $\mathtt{Pareto}(3)$ asset distribution with $v_0=2.5$ and $N\in\{2,3,5,10\}$.}
    \label{fig::Fh_pareto}
\end{figure}
\begin{figure}[htbp]
    \centering
    \includegraphics[width=\linewidth]{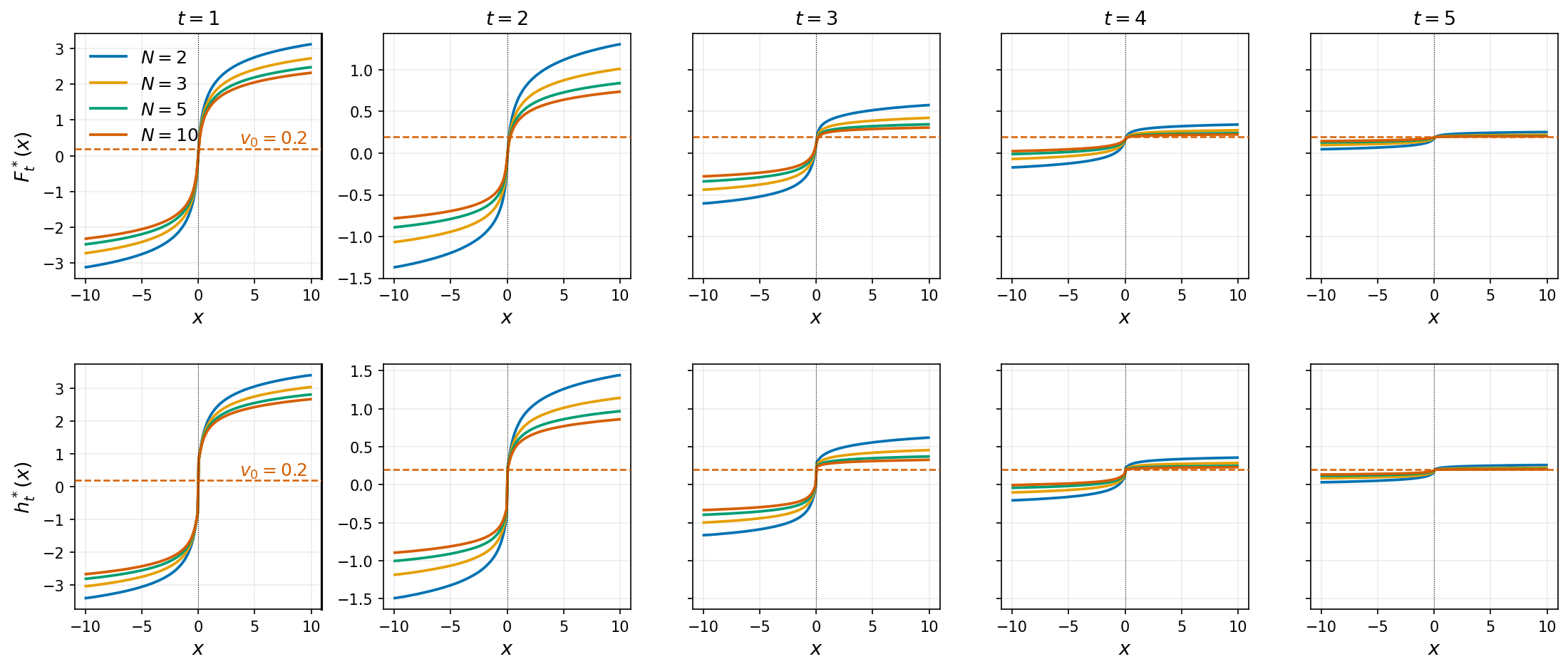}
    \caption{Equilibrium marginal cost$F_t^\star$ (\emph{Top panel}) and limit prices $h_t^\star$ (\emph{Bottom panel}) across the trading periods $t\in \{1,\ldots,5\}$, for the $\mathtt{Gaussian}(0,1)$  asset distribution with $v_0=0.2$ and $N\in\{2,3,5,10\}$.}
    \label{fig::Fh_gaussian}
\end{figure}

\begin{figure}[htbp]
    \centering
    \includegraphics[width=\linewidth]{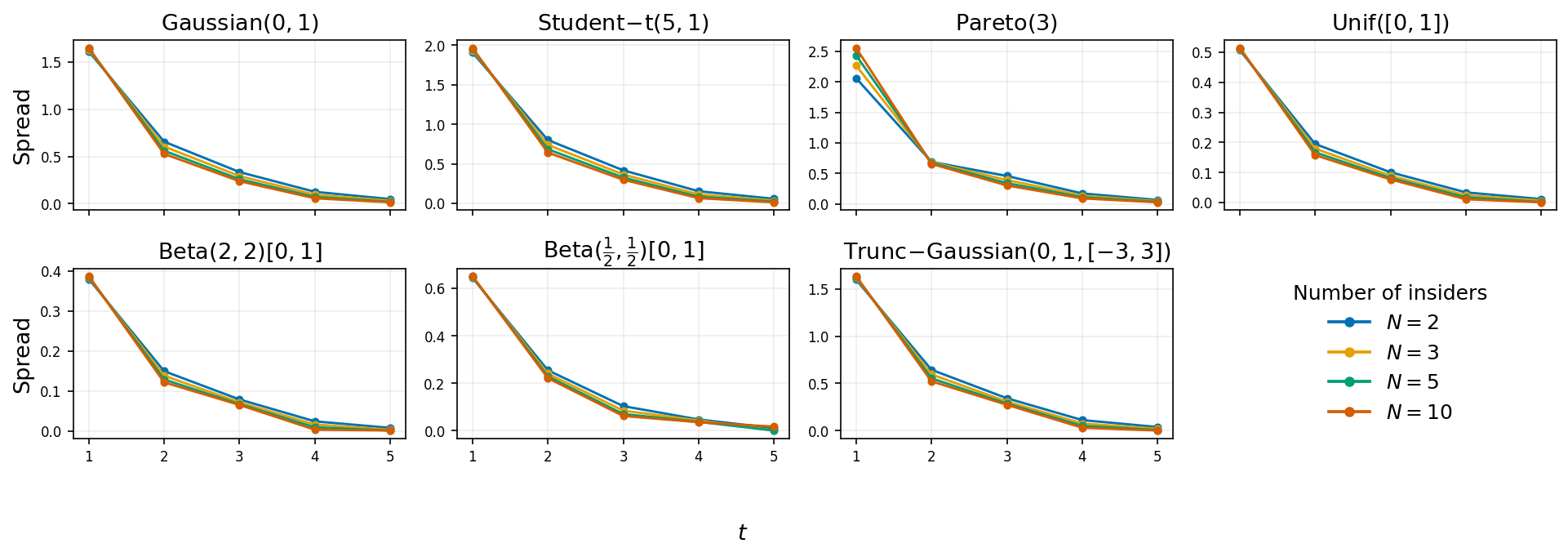}
    \caption{Bid-ask spread $h_t^\star(0^+) - h_t^\star(0^-)$ across trading periods $t=1,\ldots,5$, for the considered seven prior distributions and $N\in\{2,3,5,10\}$.}
    \label{fig::spread_all_priors}
\end{figure}
 
Finally, we assess the regular-variation tail prediction of the price impact predicted by Theorem~\ref{th::power_law_price_impact}; also, we study the sensitivity of these results to the parameter $\nu$, the number of insiders $N$, and the number of trading periods. 
Theorem~\ref{th::power_law_price_impact} shows that, for bounded-support fundamentals, the right tail of the marginal-cost schedule $M - F_t^\star(x)$ is regularly varying at infinity with index $\rho_t^+$,
where \(\rho_t^+<0\). The exponent \(\rho_t^+\) depends on the endpoint behavior of the fundamental distribution through \(L\), on the tail thickness of noise trading through \(\nu\), on the number of informed traders \(N\), and on the previously realized price-impact
exponents \(\rho_1^+,\ldots,\rho_{t-1}^+\). In particular, \(|\rho_t^+|\) decreases with \(N\) and converges to zero as \(t\to\infty\). For \(t\ge2\), \(|\rho_t^+|\) is also decreasing
in \(\nu\), while \(\rho_1^+\) is independent of \(\nu\). Hence, as \(\nu\) increases and noise-trader tails become thinner, the marginal-cost schedule approaches its upper endpoint more slowly at large order sizes. Equivalently, the book is less aggressive at depth and price
impact is lower. Economically, thinner noise tails make order flow more informative across periods, accelerating learning and reducing residual adverse-selection risk in later books. We examine these theoretical predictions numerically. 
{\color{black} For each prior distribution, $N\in\{2,3,10\}$, and $\nu\in\{3,5,10,30\}$, we solve the fixed-point in~\eqref{eq::fixed_point_equation} over trading periods $t\in\{1,\dots,5\}$. 
The resulting limit prices are computed through~\eqref{eq::g_mappings} and~\eqref{eq::equation_for_g}, and liquidity suppliers' belief is updated via~\eqref{eq::beliefsupdate}.}
For bounded-support priors, we estimate the tail exponent by fitting $\log\bigl(M-F_t^\star(x)\bigr)
\simeq \hat\rho_t^+\log x$ over the region \(x>\underline x\). For unbounded priors, we instead estimate $\log F_t^\star(x) \simeq \hat\gamma_t\log x$ over the region \(x>\underline x\). 
{\color{black}The fitted exponents \(\hat\rho_t^+\) and \(\hat\gamma_t\) are obtained by ordinary least squares on the corresponding log-log regression, using observations \(x>\underline{x}\).
}
The cut-off \(\underline x\) is held fixed across periods to make the estimated exponents
comparable over time. We vary $\nu\in\{3,5,10,30\}$ to study how the tail thickness of noise trading affects price impact. The prediction is that larger \(\nu\), corresponding to thinner liquidity-shock tails, lowers the price-impact curve in later periods. In the main text, we present results for the truncated Gaussian among the bounded-support priors, and for the Pareto and Gaussian distributions among the unbounded-support priors; results for the remaining distributions are collected in Appendix~\ref{app::price_figures}. Fig.~\ref{fig::power_law_loglog_trunc_gaussian} confirms the power law behaviour of $M - F_t^\star(x)$ for the Truncated Gaussian distribution: the log-log plot is approximately linear in the fit region, with slope $\hat{\rho}_t^+$, confirming that $M - F_t^\star(x)$ exhibits power law decay $M - F_t^\star(x) \sim x^{\rho_t^+}$ as $x \to \infty$. Fig.~\ref{fig::power_law_loglog_pareto} and Fig.~\ref{fig::power_law_loglog_gaussian} show the analogous power law behaviour $F_t^\star(x) \sim x^{\gamma_t}$ for the Pareto and Gaussian distributions.

\begin{figure}[htbp]
    \centering
    \includegraphics[width=1\linewidth]{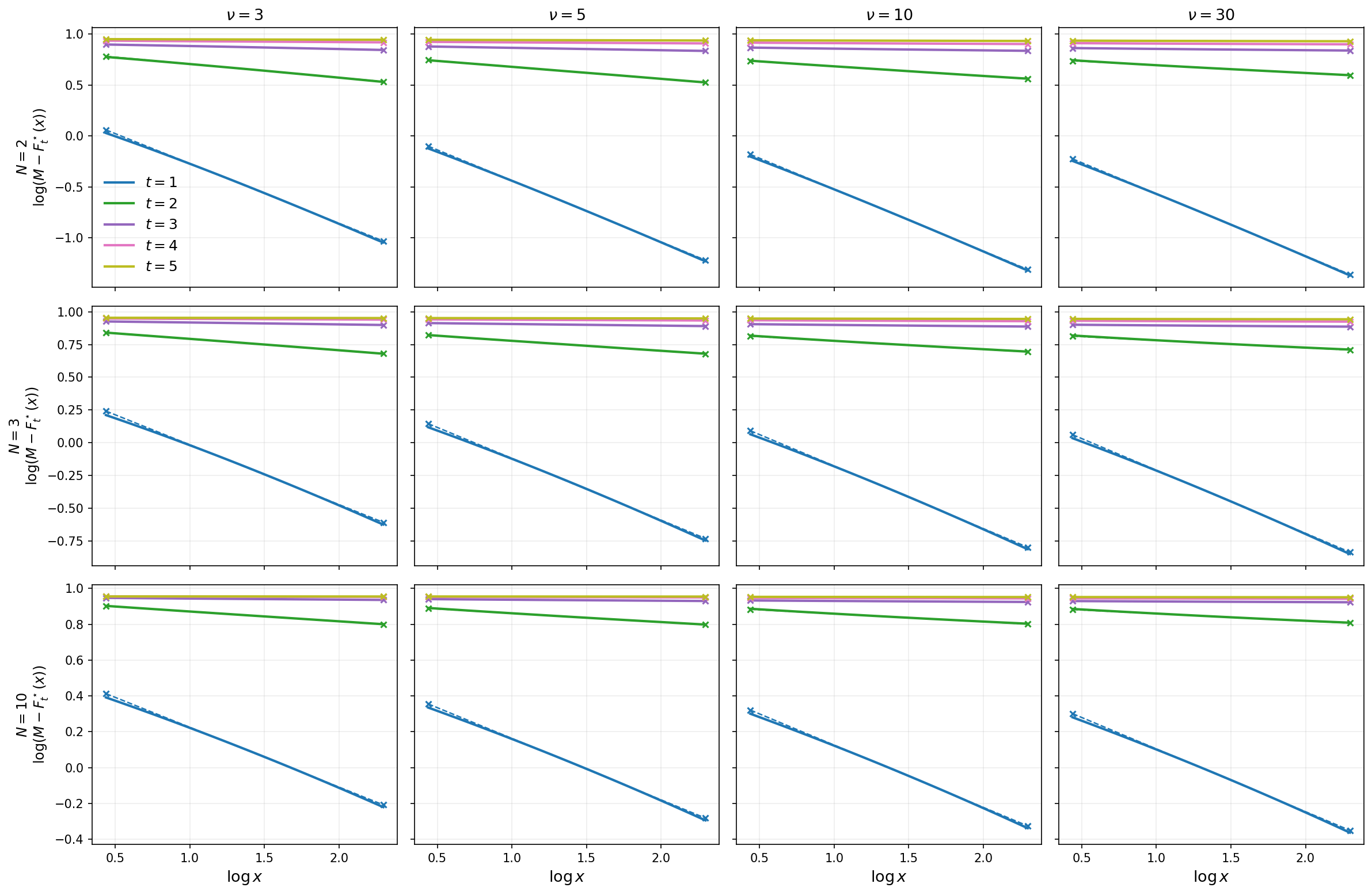}
    \caption{Log-log plot of $\log(M - F_t^\star(x))$ against $\log x$ for the $\mathtt{Trunc-Gaussian}(0,1, [-3,3])$  asset distribution, across trading periods $t=1,\ldots,5$, for $\nu\in\{3,5,10,30\}$ and $N\in\{2,3,10\}$. Solid lines: $\log(M-F_t^\star)$ from fixed point solutions; dashed lines with markers: fitted slope $\hat{\rho}_t^+$. Fit region $x > 1.5$.}
    \label{fig::power_law_loglog_trunc_gaussian}
\end{figure}
\begin{figure}[htbp]
    \centering
    \includegraphics[width=1\linewidth]{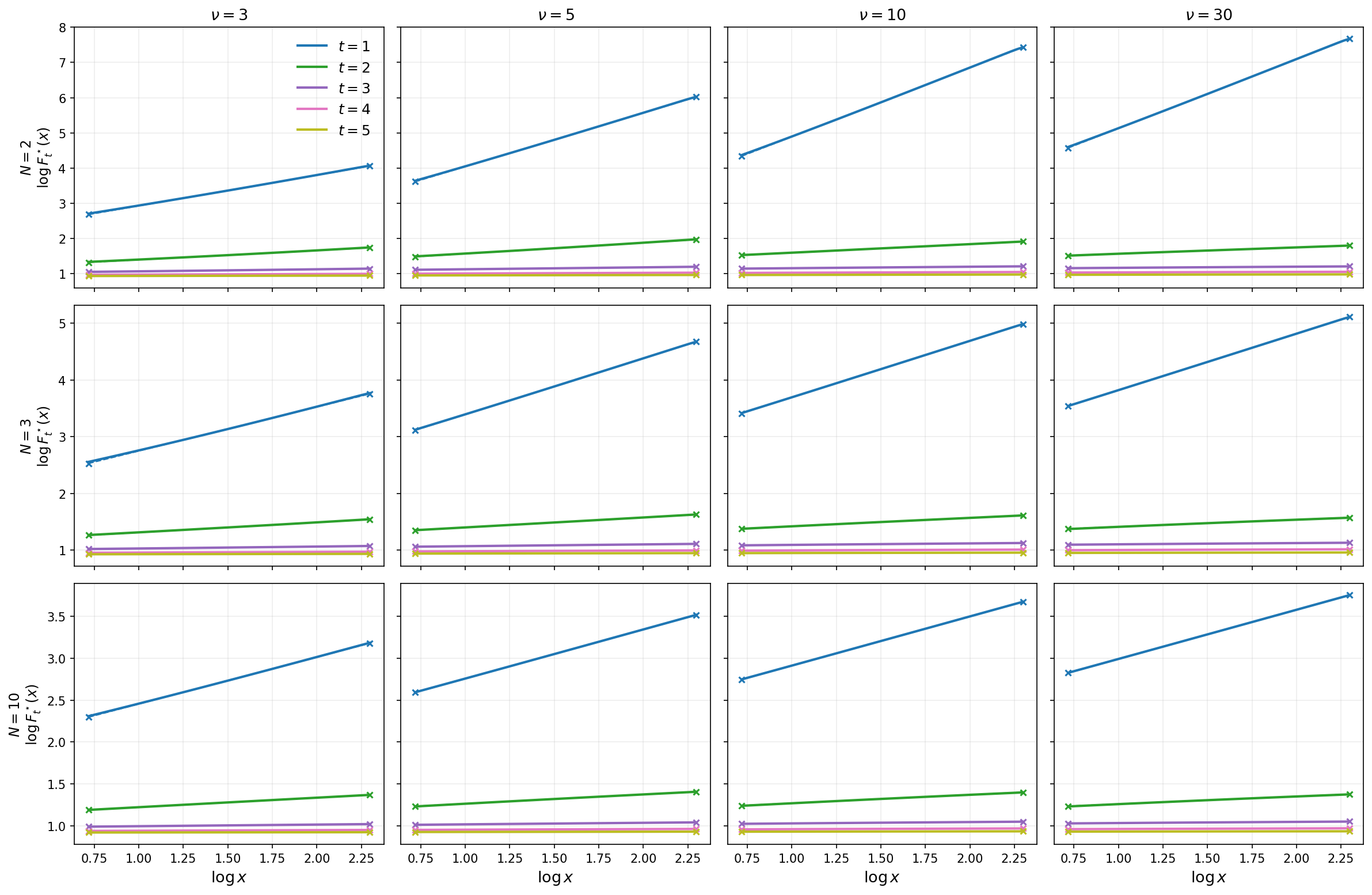}
    \caption{Log-log plot of $\log F_t^\star(x)$ against $\log x$ for the $\mathtt{Pareto}(3)$ asset distribution, across trading periods $t=1,\ldots,5$, for $\nu\in\{3,5,10,30\}$ and $N\in\{2,3,10\}$. Solid lines: $\log F_t^\star$ from fixed point solutions; dashed lines with markers: fitted slope $\hat{\gamma}_t$. Fit region $x > 2.0$.}
    \label{fig::power_law_loglog_pareto}
\end{figure}
\begin{figure}[htbp]
    \centering
    \includegraphics[width=1\linewidth]{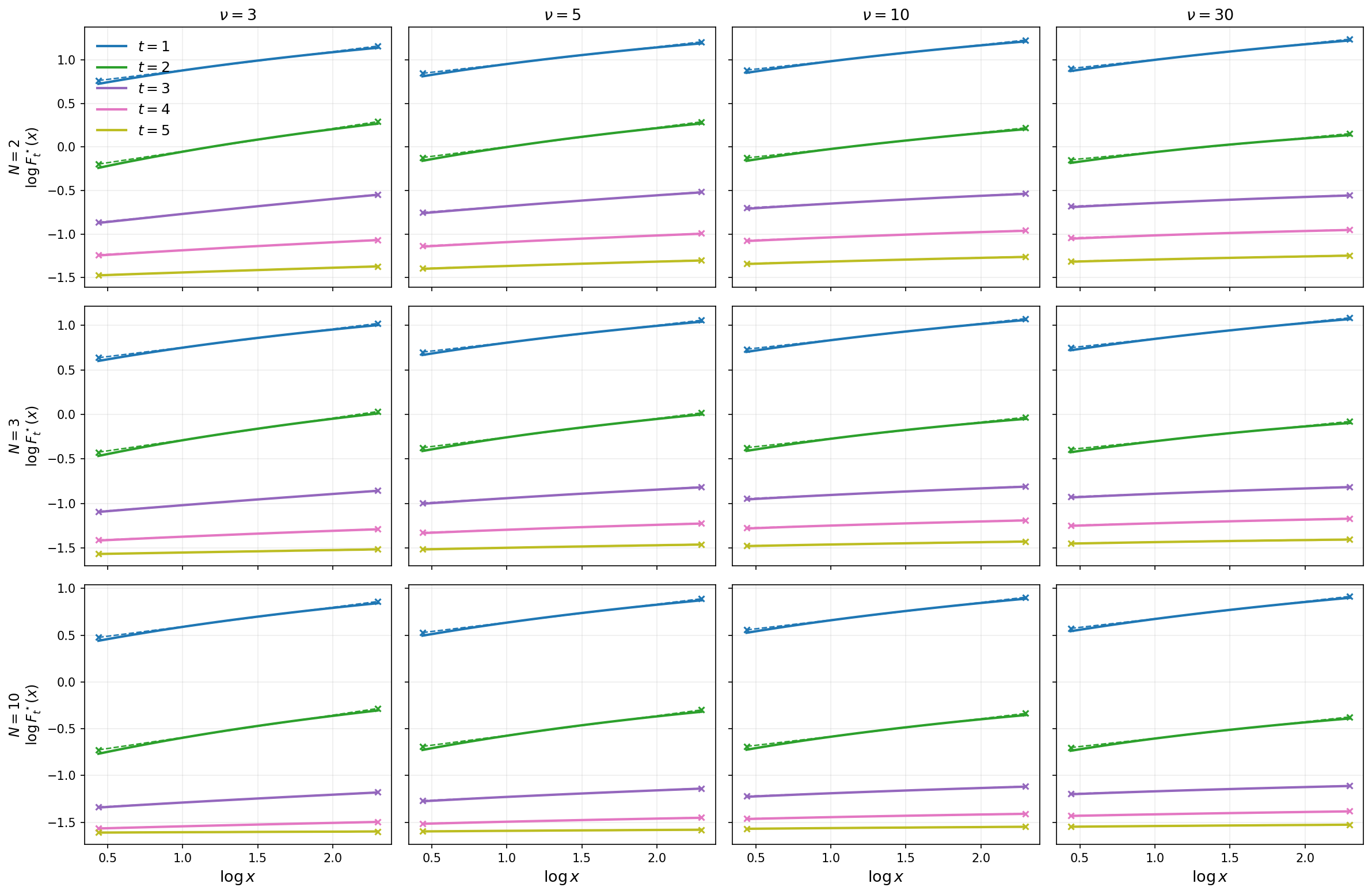}
    \caption{Log-log plot of $\log F_t^\star(x)$ against $\log x$ for the $\mathtt{Gaussian}(0, 1)$ asset distribution, across trading periods $t=1,\ldots,5$, for $\nu\in\{3,5,10,30\}$ and $N\in\{2,3,10\}$. Solid lines: $\log F_t^\star$ from fixed point solutions; dashed lines with markers: fitted 
    slope $\hat{\gamma}_t$. Fit region $x > 1.5$.}
    \label{fig::power_law_loglog_gaussian}
\end{figure}
Fig.~\ref{fig::power_law_nu_N2} reports the fitted power-law exponents across periods for \(N=2\). Both \(|\hat\rho_t^+|\) in the bounded-support case and \(\hat\gamma_t\) in the unbounded-support case decline monotonically over time. This pattern is consistent with the
learning mechanism of the model: as liquidity suppliers observe more order flow, residual adverse-selection risk falls and the price-impact schedule becomes flatter. The comparative statics with respect to \(\nu\) are also in line with Theorem~\ref{th::power_law_price_impact}. Larger values of \(\nu\) correspond to thinner noise-trader tails. In this case, large order imbalances are less easily attributed to rare
liquidity shocks and therefore reveal information more quickly. As a result, liquidity suppliers learn faster, residual adverse selection declines more rapidly, and price impact is lower in later periods. This is reflected in the faster decline of the fitted exponents as \(\nu\) increases. At \(t=1\), the curves for different values of \(\nu\) coincide, because the initial exponent \(\rho_1^+\) is independent of the noise-trading tail parameter. From \(t\ge2\) onward, the curves separate, reflecting the recursive dependence of \(\rho_t^+\) on \(\nu\) and on the previous-period exponents in Theorem~\ref{th::power_law_price_impact}. The bounded- and
unbounded-support specifications display qualitatively similar patterns. Results for \(N=3\) and \(N=10\) are reported in Appendix~\ref{app::price_figures}.

\begin{figure}
    \centering
    \includegraphics[width=\linewidth]{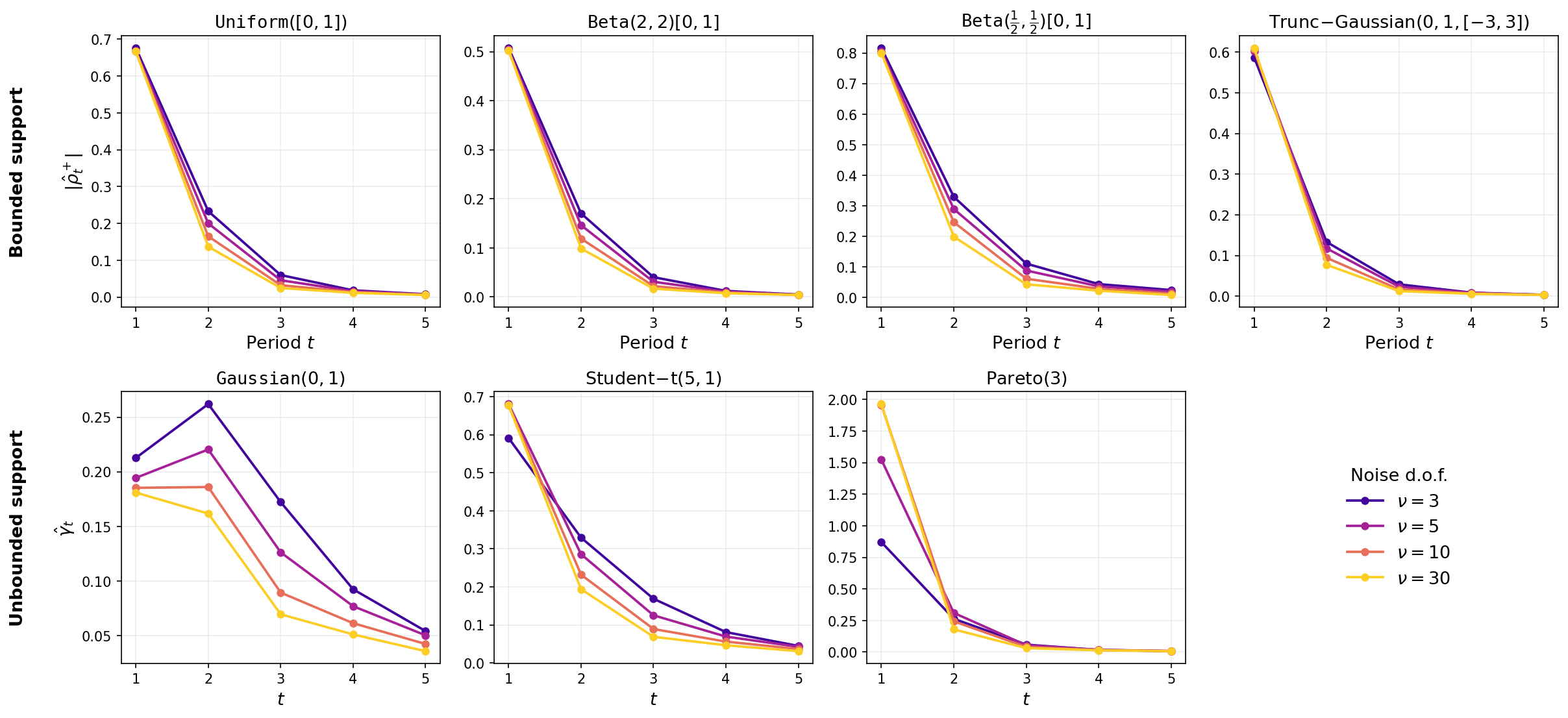}
    \caption{Fitted power law exponents $|\hat{\rho}_t^+|$ (bounded priors, from 
    $M - F_t^\star$) and $\hat{\gamma}_t$ (unbounded priors, from $F_t^\star$) 
    across trading periods $t=1,\ldots,5$, for noise trades $\nu\in\{3,5,10,30\}$ and $N=2$.}
    \label{fig::power_law_nu_N2}
\end{figure}

\section{Liquidity tail risk and price discovery in market returns}\label{sec::empirics}
Our goal in this section is to perform a descriptive, rather than fully structural (indeed, we do not estimate the model primitives) empirical exercise. In agreement with the theory developed in Sections \ref{section::Equilibrium}--\ref{section::Asymptotic_Price_Impact}, our first objective is to document that over a pre-asymptotic region of large but not  overwhelming order flow, large trades are more ambiguous in heavy-tail regimes. Second, and separately, we want to test whether, once order flow is sufficiently extreme, the model's far-tail logic allows trades to become informative again. In order to do so, we use ten levels of Apple Inc. (ticker AAPL) high-resolution LOB data, provided by LOBSTER, from April 1, 2026, to May 1, 2026, for a total of 22 trading days; we will briefly discuss -- mainly for robustness and heterogeneity checks -- results for 
Microsoft Corporation (ticker MSFT) and Intel Corporation (ticker INTC) in Appendix \ref{app::additional_empirical_results}. Each day consists of $3,960$ non-overlapping five-second windows from 10:00am to 3:30pm Eastern time; windows with no messages inherit the most recent valid quote state but have zero signed execution flow (see below). We apply minimal (and standard) filters to the data (e.g., invalid quote states removal or nonpositive or crossed quoted spread.). 

The main analyses that we will perform are the following. First, we test whether (intraday) blocks with lower Hill tail-index estimates (\cite{hill1975simple}) contain more extreme signed order flow realizations. Second, we discuss the price response to signed order flow, across tail-risk regimes and trade-size categories. Finally, we check whether spreads remain wider after large trades in heavy-tail regimes. Toward this aim, we use execution records (i.e., types $e \in \{4,5\}$ in LOBSTER) and use as direction, denoted by $\iota_{e}$, the opposite direction with respect to the one used by the data provider; i.e., buy-initiated executions are positive. Then, for every five-second window $(d,\tau)$, we define -- in accordance with the fact that liquidity suppliers observe aggregate order flow but not the motive behind it -- the absolute signed order flow ($|Q_{d,\tau}|$) and the normalized order flow size ($X_{d, \tau}$)
\begin{equation}\label{eq::definition_reg_1}
|Q_{d,\tau}|=\Big|\sum_{\substack{e \in \{4,5\} \\ e \in (d,\tau)}}\iota_{e}{\rm Size}_{e}\Big|,\quad X_{d, \tau}=\frac{|Q_{d,\tau}|}{\operatorname{median}\{|Q_{d,\tau}|\,:\,|Q_{d,\tau}|>0\}}.
\end{equation}
Also, we define the forward mid-quote return (in basis points) as
\begin{equation*}
R_{d,\tau}^{(h)}=10^{4} (\log(m_{d,\tau+h})-\log(m_{d,\tau})),
\end{equation*}
where $m_{d,\tau}$ is the end-of-window mid-quote, $h \in \{10, 30, 60, 300\}$ seconds, and $\tau+h$ denotes the window $h/5$ five second intervals ahead. Analogously, the change in relative bid-ask spread is denoted by $\Delta S_{d,\tau}^{(h)}$; the relative, with respect to the mid-price, bid-ask spread, is denoted by $S_{d,t}$ and reported in basis points (bps).

\subsection{Liquidity tail risk proxy, tail regime classification, and trade-size split}
We now construct -- our proposed -- liquidity tail-risk proxy. First, we divide each trading day into six equal $55$-minute blocks (i.e., $132$ blocks for the $22$ trading days). Within each day-block, say $g$, we estimate, via the Hill estimator, the right-tail index $\widehat{\alpha}_{g}$, of non-zero absolute signed execution flow. We use the top $10\%$ of observations of $|Q_{d,\tau}|$, subject to the requirement that the block has at least $80$ nonzero observations. Importantly, we note that empirical tail estimates should not be interpreted as direct estimates of the Student-$t$ parameter $\nu$ since the data at our disposal contain (observed) aggregate signed execution flow, which combines informed and uninformed liquidity demand. In particular, Lemma \ref{lem:candidate-exponents} shows that equilibrium aggregate flow can be heavier-tailed than the primitive liquidity shock because the induced informed-demand tail may dominate in the far tail; empirical Hill indices below $2$ do not contradict the (theoretical) assumption that $\nu>2$.\\
\noindent Tail-regime classification of blocks is then given by the following terciles: ${\rm HeavyTail}_{g}=\mathsf{1}_{\{\widehat{\alpha}_{g} \leq q_{1/3}\}}$, ${\rm LightTail}_{g}=\mathsf{1}_{\{\widehat{\alpha}_{g} \geq q_{2/3}\}}$, and the remaining blocks are classified as ${\rm MiddleTail}_g$; $\mathsf{1}_{\{\cdot\}}$ denotes the indicator function of the event $``\{\cdot\}"$, and $q_{1/3}$ and $q_{2/3}$ are the empirical terciles of block-level Hill estimates for the stock. These block-level labels are attached to each five-second window in the corresponding day-block.\\
\noindent To separate the two far-tail regimes mentioned above, we split non-zero absolute signed flow into three size categories. Precisely, let $c_{0.90}$ and $c_{0.975}$ denote the $90^{th}$ and $97.5^{th}$ percentiles of $|Q_{d,\tau}|$ among non-zero flow windows. We define: ${\rm Normal}_{d,\tau}=\mathsf{1}_{\{0<|Q_{d,\tau}|<c_{0.90}\}}$, ${\rm ModerateLarge}_{d,\tau}=\mathsf{1}_{\{c_{0.90} \leq |Q_{d,\tau}| < c_{0.975}\}}$, and ${\rm Extreme}_{d,\tau}=\mathsf{1}_{\{|Q_{d,\tau}| \geq c_{0.975}\}}$; for AAPL -- thresholds are stock specific -- the moderate-large threshold is approximately $1,115$ shares, whereas the extreme
threshold is approximately $2,623$ shares.\\
\noindent Table \ref{tab::aaple} summarizes sample construction, tail regime classification, and size split. Instead, Table \ref{tab:aapl_tail} reports liquidity tail risk classification. As intended, the {\rm Heavy-tail} blocks have substantially lower Hill tail-index estimates than middle and light-tail blocks. Importantly, the median quoted spread is similar across regimes; this indicates that the tail-risk proxy is not simply a spread proxy. Moreover, the $99^{th}$ percentile of normalized absolute signed flow (cf. last column of Table \ref{tab:aapl_tail}) shows that heavy-tail blocks also contain substantially more extreme realizations of the signed flow. Also, in AAPL, heavy-tail blocks contain about $8.54\%$ moderate-large windows and $3.92\%$ extreme windows. The corresponding shares are about $6.52\%$ and $2.02\%$ in middle blocks, and about $6.78\%$ and $1.32\%$ in light-tail blocks. In particular, heavy-tail blocks are (also) blocks where large signed execution-flow realizations occur most frequently.

\begin{table}[!ht]
\centering
\caption{\emph{AAPL sample construction and trade-size classification.} The table reports the construction of the AAPL sample used in the empirical analysis. The sample uses LOBSTER message and order-book files with the first ten displayed book levels. We remove the first and last 30 minutes of the continuous trading session and aggregate the remaining interval, into non-overlapping five-second windows. Each day therefore contributes $3,960$ windows. Tail-risk regimes are estimated over six equal $55$-minute intraday blocks per day. ``Moderate-large" and ``extreme" trades are defined using the $90^{th}$ and $97.5^{th}$ percentiles, respectively, of nonzero absolute signed execution flow in the final AAPL sample. Execution flow is signed from the market-order initiator side, with buy-initiated executions positive and sell-initiated executions negative. ``Rows dropped by core filters" refers to windows removed because of, e.g., invalid quotes}
\label{tab::aaple}
\begin{tabular}{lr}
\toprule
\textsc{Quantity} & \textsc{Value} \\
\midrule
Trading days & 22 \\
Date range & 01-Apr-2026 to 01-May-2026 \\
LOB depth used & 10 levels \\
Window length & 5 seconds \\
Usable minutes per day & 330 \\
Windows per day & 3,960 \\
Total windows & 87,120 \\
Tail blocks per day & 6 \\
Tail-block length & 55 minutes \\
Total tail blocks & 132 \\
Heavy-tail blocks & 44 \\
Middle blocks & 44 \\
Light-tail blocks & 44 \\
Moderate-large threshold & 1,115 shares \\
Extreme threshold & 2,622.6 shares \\
Moderate-large windows & 6,340 \\
Extreme windows & 2,109 \\
Rows dropped by core filters & 0 \\
\bottomrule
\end{tabular}
\end{table}

\begin{table}[!ht]
\centering
\caption{\emph{AAPL tail-regime diagnostics} This table summarizes the empirical liquidity-tail regimes used in the baseline AAPL analysis. For each stock-day-block, we estimate a Hill tail index from the upper $10\%$ of nonzero absolute signed execution-flow observations. Lower Hill estimates correspond to heavier-tailed signed flow. Blocks are classified into stock-specific terciles: the lowest-tercile Hill estimates are labeled HeavyTail (\emph{first row}), the highest-tercile estimates are labeled LightTail (\emph{last row}), and the remaining blocks form the Middle regime (\emph{second row}). Median spread ($5^{th}$ \emph{column})) is the median relative quoted spread in bps within each regime. The final column reports the $99^{th}$ percentile of normalized absolute signed flow, where normalized flow is absolute signed execution flow divided by the sample median nonzero absolute signed flow.}
\label{tab:aapl_tail}
\begin{tabular}{lrrrrr}
\toprule
Regime & Blocks & Mean Hill & Median Hill & Median spread (bps) & $p_{99}$ of $X$ (Eq. \eqref{eq::definition_reg_1}) \\
\midrule
HeavyTail & 44 & 1.46 & 1.45 & 0.767 & 29.10 \\
Middle    & 44 & 1.79 & 1.79 & 0.756 & 15.02 \\
LightTail & 44 & 2.17 & 2.09 & 0.740 & 11.61 \\
\bottomrule
\end{tabular}
\end{table}

\subsection{Empirical strategy}
We adopt the following empirical strategy. In the case of price response to signed flow $q_{d,\tau}:=\operatorname{sign}(Q_{d,\tau})\log(1+X_{d,\tau})$ (see \eqref{eq::definition_reg_1} for the definition of the two quantities), for every horizon $h \in \{10,30,60,300\}$ (seconds), we run a linear regression of forward mid-quote returns ($R_{d,\tau}^{(h)}$) on (a subset of) the estimated quantities described in the previous subsection:
\begin{equation*}
 \begin{split}
R_{d,\tau}^{(h)} ={}& \beta_0^{(h)}+
\beta_1^{(h)}q_{d,\tau}+
\beta_2^{(h)}q_{d,\tau}\mathrm{HeavyTail}_{d,\tau}+
\beta_3^{(h)}q_{d,\tau}\mathrm{LightTail}_{d,\tau} \\
&+\beta_4^{(h)}q_{d,\tau}\mathrm{ModerateLarge}_{d,\tau}
+\beta_5^{(h)}q_{d,\tau}\mathrm{Extreme}_{d,\tau} \\
&+\beta_6^{(h)}q_{d,\tau}\mathrm{ModerateLarge}_{d,\tau}\mathrm{HeavyTail}_{d,\tau}
+\beta_7^{(h)}q_{d,\tau}\mathrm{Extreme}_{d,\tau}\mathrm{HeavyTail}_{d,\tau} \\
&+\beta_8^{(h)}q_{d,\tau}\mathrm{ModerateLarge}_{d,\tau}\mathrm{LightTail}_{d,\tau}
+\beta_9^{(h)}q_{d,\tau}\mathrm{Extreme}_{d,\tau}\mathrm{LightTail}_{d,\tau} \\
&+\Gamma^{(h)\prime}W_{d,\tau}+\delta_d^{(h)}+\varepsilon_{d,\tau}^{(h)},
 \end{split}   
\end{equation*}
where , $W_{d,\tau}$ contains $10$-level imbalance\footnote{For $10$-levels depth, we define, at event time $t$ on date $t$, 
\begin{equation*}
    {\rm D}_{d,t}^{10}:={\rm D}_{d,t}^{A,10}+{\rm D}_{d,t}^{B,10}:=\sum_{\ell=1}^{10} {\rm B}_{\ell, d, t}+\sum_{\ell=1}^{10} {\rm A}_{\ell, d, t},
\end{equation*}
where ${\rm B}_{\ell, d, t}$ (resp. ${\rm A}_{\ell, d, t}$) denotes the (displayed) ask (resp. bid) size at book level $\ell$.  Weighted ask ($\sum_{\ell=1}^{10} w_{\ell} {\rm A}_{\ell, d, t}$) and bid ($\sum_{\ell=1}^{10} w_{\ell} {\rm B}_{\ell, d, t}$) depth are denoted by $\widetilde{{\rm D}}_{d,t}^{A,10}$ and $\widetilde{{\rm D}}_{d,t}^{B,10}$, respectively, where the weights are defined by $w_{\ell}=\frac{e^{-\lambda(\ell-1)}}{\sum_{j=1}^{10}e^{-\lambda(j-1)}}$, $\lambda=0.5$. The $10$-level order flow imbalance is then given by
\begin{equation*}
I_{d,t}^{10}=\frac{\widetilde{{\rm D}}_{d,t}^{B,10}-\widetilde{{\rm D}}_{d,t}^{A,10}}{\widetilde{{\rm D}}_{d,t}^{B,10}+\widetilde{{\rm D}}_{d,t}^{A,10}}
\end{equation*}}, relative spread, log $10$-level total depth, two time-of-day controls ($\operatorname{sin}(2 \pi \tau)$ and $\operatorname{cos}(2 \pi \tau)$, where tau is normalized
time over the $10:00$am to $3:30$pm Eastern time sample; $\delta_d^{(h)}$ denotes day fixed effects, which are daily intercepts (identification comes from within-day variation in order-flow size and
tail-regime labels). All regressions are estimated by ordinary least squares with Newey-West/HAC-corrected $t$-statistics with a Bartlett kernel; the length of the lag is horizon-specific, specifically $L_h=\max\{\lceil \frac{60}{5}, \frac{h}{5}\rceil\}$. We note that this correction is important because forward returns and future spread changes overlap strongly at longer horizons. Forward returns and relative-spread changes are winsorized at the $0.1$ and $99.9$ percentiles. We are particularly interested in the coefficients $\beta_6^{(h)}$, relative to $q_{d,\tau}\mathrm{ModerateLarge}_{d,\tau}
\mathrm{HeavyTail}_{d,\tau}$, and $\beta_7^{(h)}$, relative to $q_{d,\tau}\times \mathrm{Extreme}_{d,\tau}\times \mathrm{HeavyTail}_{d,\tau}$, which we report in Table \ref{tab:aapl_impact}.

\noindent The estimates show that the informativeness of heavy-tail order flow is concentrated in the upper part of the signed-flow distribution. The moderate-large heavy-tail interaction is small and statistically insignificant at the 10-second horizon, but becomes positive and significant from $30$ seconds onward. The extreme heavy-tail interaction is also positive at all horizons and becomes especially large at longer horizons. This pattern is consistent with a far-tail price-discovery interpretation. In heavy-tail regimes, ordinary signed flow has little price-discovery content, while sufficiently large signed flow realizations become informative, especially at longer horizons. In particular, results in Table \ref{tab:aapl_impact} support the model's more nuanced prediction: heavy-tailed liquidity demand expands the range over which large trades remain ambiguous, but the far tail eventually contains trades that are informative about fundamentals. The increasing magnitude of the extreme-trade coefficient from $30$ to $300$ seconds is also consistent with delayed incorporation of information rather than purely instantaneous price pressure. Importantly, estimates should be interpreted as evidence on price discovery, not as a structural estimate of the theoretical crossover depth.

\begin{table}[!ht]
\centering
\caption{\emph{Refined AAPL price-impact interactions.} The dependent variable is the $h$-second forward midquote return in bps. The table reports the interaction of signed log flow with moderate-large or extreme trade status and the HeavyTail regime. Moderate-large trades are nonzero-flow windows between the $90^{th}$ and $97.5^{th}$ percentiles of $|Q|$; extreme trades are above the $97.5^{th}$ percentile. HeavyTail blocks are the lowest-tercile Hill-index blocks. The omitted categories are the middle-tail regime and normal nonzero-flow windows. Controls include $10$-level imbalance, relative spread, log 10-level depth, time-of-day controls, and date fixed effects. Standard errors are Newey--West/HAC with horizon-specific lags.}
\label{tab:aapl_impact}
\begin{tabular}{rrrrr}
\toprule
Horizon & $q_{d,\tau}\mathrm{ModerateLarge}_{d,\tau}
\mathrm{HeavyTail}_{d,\tau}$ & $t$-stat & $q_{d,\tau}\mathrm{Extreme}_{d,\tau}\mathrm{HeavyTail}_{d,\tau}$ & $t$-stat \\
\midrule
10 sec  & 0.068 & 1.49 & 0.024 & 0.49 \\
30 sec  & 0.253 & 3.10 & 0.184 & 2.00 \\
60 sec  & 0.297 & 2.59 & 0.271 & 2.02 \\
300 sec & 0.559 & 2.12 & 1.198 & 3.34 \\
\bottomrule
\end{tabular}
\end{table}

To complement the previous regression result, we compute binned signed-impact estimates. For every horizon, tail regime and normalized-flow bin, we estimate the mean signed return $\operatorname{sign}(Q_{d,t})R_{d,t}^{(h)}$; positive values indicate that prices move in the direction of signed execution flow. Table \ref{tab::price_impact} reports the results. The largest-bin estimates are consistent with the regression evidence. At $30$, $60$, and $300$ seconds, the largest HeavyTail bin has positive signed returns. At $300$ seconds, the largest HeavyTail bin has a mean signed return of $0.738$ bp with a $t$-statistic of $2.87$; the largest LightTail bin is negative, with a mean signed return of $-0.712$ bp and a $t$-statistic of $-2.80$. 

\begin{table}[!ht]
\centering
\caption{\emph{Largest-bin nonparametric signed impact.} For each stock, horizon, and liquidity-tail regime, this table reports the mean signed forward midquote return in the largest valid normalized-flow bin. Positive values indicate price movement in the direction of signed flow. HeavyTail, Middle, and LightTail regimes are based on terciles of the $55$-minute block-level Hill tail-index estimate. The table is a descriptive far-tail price-impact diagnostic; it is not a structural crossover-depth estimate.}
\label{tab::price_impact}
\begin{tabular}{rlrrrr}
\toprule
Horizon & Regime & Bin midpoint & $N$ & Mean signed return & $t$-stat. \\
 & & & & (bp) & \\
\midrule
30 & HeavyTail & 7.35 & 3,615 & 0.195 & 2.38 \\
30 & Middle & 6.66 & 2,476 & -0.077 & -1.03 \\
30 & LightTail & 6.21 & 2,342 & -0.077 & -0.90 \\
60 & HeavyTail & 7.35 & 3,615 & 0.238 & 2.10 \\
60 & Middle & 6.66 & 2,475 & -0.164 & -1.55 \\
60 & LightTail & 6.19 & 2,330 & -0.124 & -1.01 \\
300 & HeavyTail & 7.36 & 3,608 & 0.738 & 2.87 \\
300 & Middle & 6.66 & 2,459 & -0.156 & -0.62 \\
300 & LightTail & 6.20 & 2,255 & -0.712 & -2.80 \\
\bottomrule
\end{tabular}
\end{table}

Finally, in the case of the spread (again, cf. above) we run a linear regression of the change in the relative spread $\Delta S_{d,\tau}^{(h)}$ on (a subset of) the estimated quantities described in the previous subsection: 
\begin{equation*}
    \begin{split}
        \Delta S_{d,\tau}^{(h)} ={}& \theta_0^{(h)}+
\theta_1^{(h)}\mathrm{ModerateLarge}_{d,\tau}
+\theta_2^{(h)}\mathrm{Extreme}_{d,\tau}
+\theta_3^{(h)}\mathrm{HeavyTail}_{d,\tau}
+\theta_4^{(h)}\mathrm{LightTail}_{d,\tau} \\
&+\theta_5^{(h)}\mathrm{ModerateLarge}_{d,\tau}\mathrm{HeavyTail}_{d,\tau}
+\theta_6^{(h)}\mathrm{Extreme}_{d,\tau}\mathrm{HeavyTail}_{d,\tau} \\
&+\theta_7^{(h)}\mathrm{ModerateLarge}_{d,\tau}\mathrm{LightTail}_{d,\tau}
+\theta_8^{(h)}\mathrm{Extreme}_{d,\tau}\mathrm{LightTail}_{d,\tau} \\
&+\Lambda^{(h)\prime}W_{d,\tau}+\eta_d^{(h)}+u_{d,\tau}^{(h)},
    \end{split}
\end{equation*}
where $W_{d,\tau}$ is defined as before, and $\eta_{d}^{(h)}$ denotes (again) day fixed effects. We are particularly interested in the (incremental spread resilience) coefficients $\theta_5^{(h)}$, relative to the term $\mathrm{ModerateLarge}_{d,\tau}\mathrm{HeavyTail}_{d,\tau}$, and $\theta_6^{(h)}$, relative to the term $\mathrm{Extreme}_{d,\tau}\mathrm{HeavyTail}_{d,\tau}$. Table \ref{tab:aapl_spread} displays the results; names are abbreviated for the sake of visualization. The spread-resilience evidence is strong. The $\mathrm{ModerateLarge}_{d,\tau}\mathrm{HeavyTail}_{d,\tau}$ coefficient is positive and statistically significant at all horizons. The other coefficient is larger and statistically significant at $10$, $30$, and $60$ seconds. In economic terms, the extreme HeavyTail interaction is $0.085$ bp at $10$ seconds and $0.073$ bp at $60$ seconds. Relative to the sample median spread of $0.755$ bp, these incremental effects are economically meaningful.

\begin{table}[h!]
\centering
\caption{\emph{AAPL spread-resilience interactions}. The dependent variable is the $h$-second change in relative quoted spread. Coefficients are shown in units of $10^{-6}$. The table reports incremental spread changes after moderate-large and extreme trades in HeavyTail blocks, relative to middle-tail normal-flow windows. Positive coefficients indicate that spreads remain wider after large trades in heavy-tail regimes. Controls include $10$-level imbalance, current relative spread, log 10-level depth, time-of-day controls, and date fixed effects. Standard errors are Newey--West/HAC with horizon-specific lags.}
\label{tab:aapl_spread}
\begin{tabular}{rrrrr}
\toprule
Horizon & $\mathrm{ML}_{d,\tau}\mathrm{HT}_{d,\tau}$ $(10^{-6})$ &  $t$-stat & $\mathrm{E}_{d,\tau}\mathrm{HT}_{d,\tau}$ $(10^{-6})$ & $t$-stat \\
\midrule
10 sec  & 2.62 & 2.39 & 8.53 & 3.06 \\
30 sec  & 3.41 & 3.19 & 7.09 & 2.57 \\
60 sec  & 2.93  & 2.73 & 7.26 & 2.83 \\
300 sec & 2.71  & 2.26 & 4.39 & 1.56 \\
\bottomrule
\end{tabular}
\end{table}

\section{Conclusion}
\label{section::Conclusion}
We develop a dynamic competitive limit order book model in which uninformed liquidity demand is heavy-tailed. The main economic message is that large order imbalances are not automatically informative when liquidity shocks have Student-$t$ tails. In particular, the Student-$t$ specification is structural: it changes the liquidity supplier's inference problem, as well as . Under Gaussian noise, extreme uninformed orders are exponentially unlikely, so large order imbalances are quickly interpreted as informative. Under Student-$t$ noise, the market must continue to assign meaningful probability to rare liquidity shocks. This changes both the economics of price impact and the mathematics of the equilibrium problem.\\
\noindent  

On the theoretical side, we show that the Gaussian proof strategy cannot be imported directly. The Student-$t$ pricing map does not preserve monotonicity on the broad class of bounded increasing schedules, and the fixed-point operator is not continuous on the broad Gaussian compact class. These failures reflect the same economic force that motivates the model: remote liquidity states remain relevant at polynomial order. To handle this difficulty, we separate the analysis into distinct steps. 

First, we prove existence of Student-$t$ fixed points on a tail-controlled compact class. Second, we distinguish fixed-point existence from the monotone limit order book interpretation.\\
\noindent Then, w A fixed point gives a candidate marginal-cost schedule consistent with informed-trader optimality and competitive zero-profit pricing; monotonicity of the associated pricing schedule is the additional condition that makes the solution an admissible LOB. Third, we study learning and 
tail behavior along the selected solution. The learning result shows that, despite heavy-tailed liquidity shocks, repeated order flow 
continues to reveal the fundamental value under suitable stability conditions. Heavy-tailed noise slows learning, but it does not eliminate information aggregation.\\
\noindent  

The tail analysis then identifies explicit regular-variation exponents for the marginal-cost schedule, the price schedule, and the induced informed demand. These exponents show how the number of informed traders, the endpoint behavior of beliefs, and the tail thickness of liquidity shocks determine the shape of price impact. The model also clarifies the role of rare liquidity shocks in the far tail. In the bounded-support environment studied here, sufficiently extreme executions are ultimately dominated by the informed side of the market. This does not contradict the liquidity-shock interpretation. Rather, it identifies the economic role of heavy-tailed liquidity demand.: it determines the crossover scale at which large order flow stops being plausibly attributable to liquidity demand and becomes overwhelmingly informative. The economically relevant region is precisely the pre-asymptotic range in which large trades are substantial but still ambiguous. 

We also use high-frequency LOB data to examine how intraday liquidity-tail regimes, identified by the Hill tail index, affect order flow and price impact in AAPL shares; we (briefly) discussed results also for MSFT and INTC. The analysis shows that heavy-tail periods are marked by more extreme order-flow events, but not by wider quoted spreads. Price-impact regressions reveal that while normal trades in heavy-tail regimes have limited informativeness, moderate-large and extreme trades become increasingly informative over longer horizons. The the estimates are consistent with the model’s far-tail interpretation that heavy-tailed liquidity demand broadens the ambiguous region for order informativeness, but that very large trades eventually provide clear information. Additionally, future spreads widen significantly after large trades in heavy-tail regimes, indicating persistent adverse-selection risk. 
Overall, the findings provide descriptive support for the fact that tail-risk measures capture the distributional shape of order flow and have important implications for understanding market liquidity and information transmission.

In summary, the current paper provides a theory of how heavy-tailed liquidity shocks reshape the limit order book. They flatten impact, slow learning, and require a fixed-point approach adapted to
tail behavior. The results suggest that empirical models of price impact and market learning should treat the tail thickness of liquidity demand as a structural feature of the trading environment, not merely as a volatility parameter.

\bibliographystyle{plain}
\bibliography{refs}

\appendix
\section{On the failure of the monotonicity preservation of the map \texorpdfstring{$g \mapsto \phi_g$}{g -> phi\_g}}\label{app::monotonicity_failure}
We consider, in the trading period $t=1$, $V$ uniformly distributed over the interval $[0,1]$; in symbols, $V \overset{d}{\sim} \texttt{Unif}([0,1])$. Because $t=1$, in what follows, we do not make explicit the dependence on $t$. Besides, the symbol $``\sim"$ stands for asymptotically equivalent.

For any $a>0$, we let $g_{a}(z):=\min\{1,\max\{0, a z\}\}$, where $z \in \mathbb{R}$. Then $g_a:\mathbb{R} \mapsto [0,1]$ is bounded, continuous, non-decreasing, and $a$-Lipschitz. Moreover, $\phi^{+}_{g_a}(x)>\frac{1}{2}$ for every finite $x \in \mathbb{R}$, while $\lim_{x \rightarrow \pm \infty}\phi_{g_a}^{+}(x)=\frac{1}{2}$. As a consequence, $x \mapsto \phi^{+}_{g_a}(x)$ is not non-decreasing. In order to prove the previous statement, notice that for $V \overset{d}{\sim} \texttt{Unif}([0,1])$, we have: 
\begin{equation*}
    \Pi_1^+(y)=
\begin{cases}
1,&y\le 0,\\
1-y,&0<y<1,\\
0,&y\ge 1,
\end{cases}
\qquad
\Phi_1^+(y)=
\begin{cases}
\dfrac12,&y\le 0,\\[0.6em]
\dfrac{1-y^2}{2},&0<y<1,\\[0.6em]
0,&y\ge 1.
\end{cases}
\end{equation*}
and $g_a(z)=0\ (z\le 0)$, $g_a(z)=az\ (0<z<1/a)$ and $g_a(z)=1\ (z\ge 1/a)$. Therefore, 
\begin{equation*}
    \phi^{+}_{g_a}(x) = \frac{\frac{1}{2}\int_{-\infty}^{0}\mathfrak{q}_{\nu}(x-z;0, \sigma)\ud z + \frac{1}{2}\int_{0}^{1/a}(1-a^2 z^2)\mathfrak{q}_{\nu}(x-z;0, \sigma)\ud z}{\int_{-\infty}^{0}\mathfrak{q}_{\nu}(x-z;0, \sigma)\ud z + \int_{0}^{1/a}(1-a z)\mathfrak{q}_{\nu}(x-z;0, \sigma)\ud z}.
\end{equation*}
We note that the previous quantity can be rewritten as $\phi^{+}_{g_a}(x) = \frac{1}{2}+\frac{1}{2}\frac{C_a(x)}{(A_a(x)+B_a(x))}$, where $A_a(x):=\int_{-\infty}^0 \mathfrak{q}_\nu(x-z;0, \sigma)\ud z$, $B_a(x):=\int_0^{1/a}(1-az)\mathfrak{q}_\nu(x-z;0, \sigma)\ud z$, and $C_a(x):=\int_0^{1/a}az(1-az)\mathfrak{q}_\nu(x-z;0, \sigma)\ud z$.  Because $\mathfrak{q}_{\nu}(\cdot;0,\sigma)>0$ on $\mathbb{R}$ and $a z (1-a z)>0$ on $(0,1/a)$, one has $C_a(x)>0$ for every finite $x$; therefore, $\phi_{g_a}^{+}(x)>1/2$ for every finite $x$. In addition, $A_{a}(x)\rightarrow 1$, $B_{a}(x)\rightarrow 0$, and $C_{a}(x)\rightarrow 0$ as $x \rightarrow -\infty$; therefore $\phi^{+}_{g_a}(x) \rightarrow 1/2$. And $A_{a}(x) \sim x^{-\nu}$, $B_{a}(x) \sim O(x^{-(\nu+1)})$, and $C_{a}(x) \sim O(x^{-(\nu+1)})$ as $x \rightarrow +\infty$; therefore, $\phi^{+}_{g_a}(x) \rightarrow 1/2$. Whence, $\phi^{+}_{g_a}(x)$ cannot be non-decreasing. 

\section{On the failure of the continuity of the Student-t fixed-point operator}\label{app::continuity_failure}
Again, we consider a very simple posterior: $V \overset{d}{\sim} \texttt{Unif}([0,1])$, hence $m=0$ and $M=1$. Under this assumption, we have that $\Pi^+(y)=1-y$ and $\Phi^+(y)=\frac{1-y^2}{2}$ when $0<y<1$; cf. also the previous section. We now define
\begin{equation*}
 r(u):=
\begin{cases}
0,&u\le0,\\
u,&0<u<1,\\
1,&u\ge1,
\end{cases}   
\end{equation*}
and $g_n(x):=r(x+n)$. Notice that each $g_n$ is continuous, non-decreasing, $1$-Lipschitz, and takes values in $[0,1]$. Moreover, $g_n(x)=0$ for $x \leq -n$ and $g_n(x)=1$ for $x \ge -n+1$. We let $g_{\infty} \equiv 1$. Because $g_n \ne 1$ only on the set $(-\infty,-n+1)$, we have that $\|g_n-g_\infty\|_{L^2(\mathbb{R}, \mu_0)}^2 \le \mu_0((-\infty,-n+1))$, which converges to zero as $n \rightarrow \infty$; in this section -- notice the abuse of notation -- $\mu_0(\ud x)=\mathfrak{q}_{\bar{\nu}}(x;0,\bar{\sigma})\,\ud x$. Therefore, $g_n \to g_\infty$ in $L^2(\mathbb{R},\mu_0)$.

We now consider the ask-side map $\phi_{g_n}^{+}$ as in \eqref{eq::g_mappings} and fix $x \geq 0$. Because \(g_n(z)=1\)
for \(z\ge -n+1\), then $\Pi^+(1)=0$ and $\Phi^+(1)=0$, and only the region \(z<-n+1\) contributes to the numerator and denominator. Instead, on the flat region $z\le -n$, we have $g_n(z)=0$, hence $\Pi^+(g_n(z))=1$ and  $\Phi^+(g_n(z))=1/2$. Therefore, the flat part of the denominator is 
\begin{equation*}
\int_{-\infty}^{-n}\mathfrak{q}_{\nu}(x-z;0,\sigma)\ud z
= \int_{x+n}^{\infty}\mathfrak{q}_{\nu}(u;0,\sigma)\ud u \sim n^{-\nu},
\end{equation*}
where in the last step we have used the fact that the Student-$t$ survival tail is regularly varying with index $-\nu$. On the other hand, the transition region $(-n, -n+1)$ has length one, and on that region $x-z \sim n$. hence, its
contribution is only $O(n^{-(\nu+1)})$. As a consequence, the transition region is negligible relative to the far-left flat region. We thus have
\begin{equation*}
    \int_{-\infty}^{+\infty} \Pi^+(g_n(z))\mathfrak{q}_{\nu}(x-z;0,\sigma)\ud z
\sim
\int_{x+n}^{\infty}\mathfrak{q}_{\nu}(u;0,\sigma)\ud u,\quad 
\end{equation*}
and 
\begin{equation*}
\int_{-\infty}^{+\infty} \Phi^+(g_n(z))\mathfrak{q}_{\nu}(x-z;0,\sigma)\ud z
\sim
1/2\int_{x+n}^{\infty}\mathfrak{q}_{\nu}(u;0,\sigma)\ud u.
\end{equation*}
The previous two identities imply, in particular, that $\phi_{g_n}^+(x)\rightarrow 1/2$. On the other hand, for the \(L^2(\mathbb{R},\mu_0)\)-limit \(g_\infty\equiv 1\), the ask-side denominator is identically zero:
\begin{equation*}
\int_{-\infty}^{+\infty} \Pi^+(g_\infty(z))\mathfrak{q}_{\nu}(x-z;0,\sigma)\ud z=0.
\end{equation*}
Notice that the natural endpoint convention is $\phi_{g_\infty}^+(x)=M=1$. Therefore $\phi_{g_n}^+(x)\rightarrow 1/2 \neq 1=\phi_{g_\infty}^+(x)$. In addition, notice that the problem is not the lack of dominated convergence for the numerator and denominator separately, but for the ratio, which converges to a value that depends on the way the endpoint function is approached.

The main issue is that the previous discontinuity propagates to the full Student-$t$ fixed-point operator in \eqref{eq::Student_fixed_point_operator}. For every fixed \(z\ge0\), \(\varphi^+_{g_n}(z)\to1/2\). For every fixed \(z<0\), \(\varphi^-_{g_n}(z)\to1/2\). Hence \(\varphi_{g_n}(z)\to1/2\) for every fixed \(z\). However, for the limit function \(g_\infty\equiv1\) we have 
\begin{equation*}
\phi_{g_\infty}(z)
=
\begin{cases}
1,&z\ge0,\\[0.4ex]
1/2,&z<0,
\end{cases}
\end{equation*}
and therefore 
\begin{equation*}
(\mathcal{T}_{t,\nu}g_\infty)(x)=
\frac{1}{2}+
\frac{1}{2}\int_0^\infty
\left[
\frac1{N_t}\mathfrak{q}_{\nu}(x-z;0,\sigma)
+
\frac{N_t-1}{N_t}\mathfrak{\bar{q}}_{\nu}(x,z;0,\sigma)
\right]\ud z.
\end{equation*}
Now, for $x \in [0,1]$
\begin{equation*}
    \int_0^\infty
\left[
\frac1{N_t}\mathfrak{q}_{\nu}(x-z;0,\sigma)
+
\frac{N_t-1}{N_t}\mathfrak{\bar{q}}_{\nu}(x,z;0,\sigma)
\right]\ud z \geq \frac{1}{2}.
\end{equation*}
Hence
\begin{equation*}
    (\mathcal{T}_{t,\nu}g_\infty)(x)-\frac{1}{2} \geq \frac{1}{4},
\qquad x\in[0,1].
\end{equation*}
Because the Student-$t$ reference measure $\mu_0$ has strictly positive density on $[0,1]$, we have
\begin{equation*}
\liminf_{n\to\infty}
\|\mathcal{T}_{t,\nu}g_n-\mathcal{T}_{t,\nu}g_\infty\|_{L^2(\mu_0)}^2
\ge
\left(\frac{1}{4}\right)^2\mu_0([0,1])>0.
\end{equation*}
As a consequence, $\mathcal{T}_{t,\nu}$ is not continuous in the broad $L^2(\mathbb{R}, \mu_0)$ topology, even when $\mu_0$ is Student-$t$.

Importantly, the key object in the continuity calculation is the noise kernel inside the pricing operator. With Gaussian noise, a far-away transition is exponentially localized at its closest edge, where \(g_n\) is already close to the endpoint value. With Student-$t$ noise, the polynomial survival tail keeps the remote flat region visible at order \(n^{-\nu}\), while the transition region contributes only order \(n^{-(\nu+1)}\). Therefore, the conditional-expectation ratio remembers the remote tail region and fails to converge to the endpoint value.

\section{Proofs of the results in Section \ref{section::Equilibrium}}\label{app::proof_existence_fixed_point}
\subsection{Proof Lemma \ref{lemma::monotonicity}.}\label{app::proof_Lemma_monotonicity}
We prove the lemma for the function \( u^{+} \); the proof for \( u^{-} \) is similar. Additionally, the lemma holds in every time period \( t \); therefore, we simplify the notation by omitting the time dependency on the period \( t \), and we prove it over a time interval $[0,T]$.\\
\indent\textit{Proof of (i)} Define $M_t:=\mathbb{E}[\Pi^{+}(g(B_T)) \vert \mathcal{F}_t]$. It is a uniformly integrable $(\mathbb{P},\mathcal{F}_t)$-martingale. Moreover, by using the independence of $T$ and $B$ under $\mathbb{P}$ (second equality), and the Markov property of $B$ together, again, with the independence property (last equality), we have
\begin{equation*}
    \mathtt{1}_{\{T>t\}} M_t = \mathtt{1}_{\{T>t\}} \mathbb{E}[\Pi^{+}(g(B_T)) \vert \mathcal{F}_t]= \mathtt{1}_{\{T>t\}} \frac{\mathbb{E}[\Pi^{+}(g(B_T)) \mathtt{1}_{\{T>t\}} \vert \mathcal{F}_t^{B}]}{\mathbb{P}(T>t)}=\mathtt{1}_{\{T>t\}} \frac{u^{+}(t,B_t)}{\mathbb{P}(T>t)},
\end{equation*}
where $(\mathcal{F}_t^{B})_{t \in [0,T]}$ is the natural filtration of the Brownian motion $B$. The claim follows from an application of Girsanov theorem.\\
\indent\textit{Proof of (ii)} The proof of $(ii)$ begins with the following observation. Let $Y \overset{d}{\sim}\mathtt{T}_{\nu}(x,\sigma)$. Then, it can be written as a Gaussian scale mixture with an inverse-gamma distributed variance, namely
\begin{equation*}
    Y | T \overset{d}{\sim} \mathtt{Normal}(x,T),\quad T \overset{d}{\sim}\texttt{Inv-Gamma}\left(\frac{\nu}{2}, \frac{\sigma^2 \nu}{2}\right).
\end{equation*}
As a consequence, $Y:=B_T$ with $B_0=x$ can be interpreted as a Brownian motion starting from $x$ and observed at an independent inverse-gamma distributed random time $T$. Hence, 
\begin{equation*}
    \phi_g^{+}(x) = \frac{\mathbb{E}[\Phi^{+}(g(Y))]}{\mathbb{E}[\Pi^{+}(g(Y))]} = \frac{\mathbb{E}[\Phi^{+}(g(B_T))]}{\mathbb{E}[\Pi^{+}(g(B_T))]} = \frac{\mathbb{E}[\Psi^{+}(g(B_T))\Pi^{+}(g(B_T))]}{\mathbb{E}[\Pi^{+}(g(B_T))]}=\mathbb{E}^{\mathbb{Q}^{+}}[\Psi^{+}(g(B_T)].
\end{equation*}
Analogously, $\phi_g^{-}(x)=\mathbb{E}^{\mathbb{Q}^{-}}[\Psi^{-}(g(B_T)].$\\
\indent \indent\textit{Proof of (iii)} The proof adheres to the same approach as outlined in Lemma 4.1 of \cite{ccetin2023power}, after the observation that $\phi^{\pm}_{g}(x)$ can be written as $ \phi^{\pm}_{g}(x) = \frac{\mathbb{E}[\Phi^{+}(g(x+R))]}{\mathbb{E}[\Pi^{+}(g(x+R))]}$ where $R\overset{d}{\sim}\mathtt{T}_{\nu}(0,1)$.

\subsection{Proof of Lemma \ref{lem::keyproperties_F}}\label{app::proof_keyproperties_F}
We start by proving that $\lim_{x\rightarrow+\infty}F(x,t,Y^{t-1})=M$. Because $\operatorname{supp}(V)=[m,M]$ and $F(\cdot,t,Y^{t-1})$ is a continuous non-decreasing solution of \eqref{eq::fixed_point_equation}, we have $F(x,t,Y^{t-1}) \leq M\,\forall x \in \mathbb{R}$. By contradiction, we assume that $\lim_{x\rightarrow+\infty}F(x,t,Y^{t-1}):=L<M$.   Since \(\varphi_F\) is bounded by \(m\) and \(M\), dominated convergence applies. 
\begin{equation}
    \lim_{x\rightarrow +\infty}\int_{-\infty}^{+\infty}\phi_{F(\cdot,t,Y^{t-1})}(x+z)\mathfrak{q}_{\nu}(z; 0, \sigma)\,\ud z = \int_{-\infty}^{+\infty}\mathfrak{q}_{\nu}(z; 0, \sigma)\lim_{x\rightarrow +\infty}\phi_{F(\cdot,t,Y^{t-1})}(x+z)\,\ud z.
\end{equation}
Moreover, $\lim_{x\rightarrow+\infty}\phi_{F(\cdot,t,Y^{t-1})}(x+z)=\Psi^{+}(F(\infty,t,Y^{t-1})-)$. It is enough to prove the previous result for $x+z \geq 0$; the case $x+z<0$ is analogous. To simplify notation, in the rest of the proof we set $F_t:=F_t(\cdot):=F(\cdot,t,Y^{t-1})$. By definition, we have
\begin{equation*}
\begin{split}
    \phi_{F_t}^{+}(x+z)=\frac{\int_{-\infty}^{+\infty}\Psi^{+}_{t}(F_t(u))\Pi^{+}_{t}(F_t(u))\mathfrak{q}_{\nu}(x+z-u; 0, \sigma)\,\ud u}{\int_{-\infty}^{+\infty}\Pi^{+}_{t}(F_t(u))\mathfrak{q}_{\nu}(x+z-u; 0, \sigma)\,\ud u}
\end{split}
\end{equation*}
We now show that the following normalized weighted measure converges as $x\rightarrow+\infty$ to the point mass at $\infty$. 
\begin{equation}\label{eq::normalized_weighted}
    \frac{\Pi^{+}_{t}(F_t(u))\mathfrak{q}_{\nu}(x+z-u; 0, \sigma)\,\ud u}{\int_{-\infty}^{+\infty}\Pi^{+}_{t}(F_t(u))\mathfrak{q}_{\nu}(x+z-u; 0, \sigma)\,\ud u},
\end{equation}
i.e., it places arbitrarily small mass on any fixed interval $(-\infty, a]$ as $x\rightarrow+\infty$ with $0<a<\infty$.  We first bound the numerator; $C(\nu,\sigma)>0$ is a constant that may vary from line to line. For $\nu>2$ we have
\begin{equation*}
    \begin{split}
        \int_{-\infty}^{a}\Pi^{+}_{t}(F_t(u))\mathfrak{q}_{\nu}(x+z-u; 0, \sigma)\,\ud u &\leq \int_{-\infty}^{a}\mathfrak{q}_{\nu}(x+z-u; 0, \sigma)\,\ud u\\
        &\leq C(\nu,\sigma) \int_{-\infty}^{a} \left[\frac{(x+z-u)^2}{\nu \sigma^2}\right]^{-\frac{(\nu+1)}{2}}\,\ud u\\
        &=C(\nu,\sigma) (x+z-a)^{-\nu},
    \end{split}
\end{equation*}
which converges to $0$ as $x \rightarrow +\infty$. Now, we find a strictly positive lower bound for the denominator, after a change of variable $k:=u-(x+z)$. Fix $\delta>0$. Because $F_t(x) \uparrow L < M$ as $x \rightarrow +\infty$, there exists a $\bar{x}$ such that for all $x \geq \bar{x}$ and $k \in [-\delta,\delta]$, $F_t(x+z+k) \leq L + \varepsilon_0 < M$ for some small $\varepsilon_0>0$. Since $\Pi_t^{+}(v)=\mathbb{P}_{t}(V \geq v)$ is strictly positive for every $v < M$, there exists $\varepsilon:=\inf_{v \in [m,L+\varepsilon_0]}\Pi_t^{+}(v)>0$. Therefore, for all large $x$ and $k \in [-\delta,\delta]$ we have $\Pi_t^{+}(F_t(x+z+k))\geq\varepsilon$; hence the denominator is uniformly bounded away from zero for large $x$. We have shown that the normalized weighted measure in \eqref{eq::normalized_weighted} converges to zero as $x \rightarrow \infty$, and therefore 
\begin{equation*}
    \lim_{x\rightarrow\infty}\phi_{F_t}(x+z)=\Psi_t^{+}(F_t(\infty)-).
\end{equation*}
Putting, together the previous equation and \eqref{eq::fixed_point_equation} we obtain that $F_t(\infty)=\Psi_t^{+}(F_t(\infty)-)$. However, $\Psi_t^{+}(x-)>x$ for every $x<M$: the previous identity contrasts the assumption $\lim_{x \rightarrow \infty}F_t(x)<M$. Therefore, we must have  $\lim_{x \rightarrow \infty}F_t(x)=M$. Similarly, $\lim_{x \rightarrow -\infty}F_t(x)=m$.

Finally, under the assumption that $h(\cdot, t, Y^{t-1})$ in non-decreasing and non-constant, we need to prove that $F^{'}(x,t,Y^{t-1})$ is strictly positive; it will turn out that Lemma \ref{lemma::monotonicity} is a key result. First, we write down $F(x,t,Y^{t-1})$ in the following way
\begin{equation}\label{eq::rewriting_F}
    \begin{split}
        F(x,t,Y^{t-1}) &= \frac{1}{N_t}\int_{-\infty}^{+\infty} \mathfrak{q}_{\nu}(x-z;0,\sigma)\phi_{F(\cdot,t,Y^{t-1})}(z)\,\ud z + \frac{N_t-1}{N_t}\int_{-\infty}^{+\infty}\bar{\mathfrak{q}}_{\nu}(x,z;0,\sigma)\phi_{F(\cdot,t,Y^{t-1})}(z)\,\ud z\\
                       &:=\frac{1}{N_t} G(x,t,Y^{t-1}) + \frac{N_t-1}{N_t} H(x,t,Y^{t-1})
    \end{split}
\end{equation}
Lemma \ref{lemma::monotonicity}-\emph{(iii)} implies that $\Delta_{F(\cdot,t,Y^{t-1})}:=\phi^{+}_{F(\cdot,t,Y^{t-1})}(0)-\phi^{-}_{F(\cdot,t,Y^{t-1})}(0)>0$; Stieltjes' derivatives is therefore a positive measure that assigns a strictly positive measure $\Delta$ to the point zero. We now compute
\begin{equation*}
\begin{split}
    G^{'}(x,t,Y^{t-1})&=\int_{-\infty}^{+\infty}\partial_x \mathfrak{q}_{\nu}(x-z;0,\sigma)\phi_{F(\cdot,t,Y^{t-1})}(z)\,\ud z = \int_{-\infty}^{+\infty}\mathfrak{q}_{\nu}(x-z;0,\sigma)\,  \phi_{F(\cdot,t,Y^{t-1})}(\ud z)\\
                      &\geq \Delta_{F(\cdot,t,Y^{t-1})} \mathfrak{q}_{\nu}(x ;0,\sigma) > 0,\quad \forall x \in \mathbb{R},
\end{split}
\end{equation*}
where we have used Assumption \ref{ass::integrability} to interchange the integral with the partial derivative. Regarding $H(x,t,Y^{t-1})$, for $x \neq0$ we have
\begin{equation*}
\begin{split}
    H(x,t,Y^{t-1})&=\frac{1}{x}\int_{0}^{x} G(u, t,Y^{t-1})\,\ud u.
\end{split}
\end{equation*}
So
\begin{equation*}
\begin{split}
    H^{'}(x,t,Y^{t-1})&=\frac{1}{x} G(x, t,Y^{t-1}) - \frac{1}{x^2}\int_{0}^{x} G(u, t,Y^{t-1})\,\ud u.
\end{split}
\end{equation*}
From the previous step, $G(\cdot, t,Y^{t-1})$ is strictly increasing. If $x>0$, then $\frac{1}{x}\int_{0}^{x} G(u, t,Y^{t-1})\,\ud u$ is strictly less than $G(x, t,Y^{t-1})$, which implies that $H^{'}(x, t,Y^{t-1})>0$. Instead, if $x<0$, the previous average over $[x,0]$ is strictly greater than $G(x, t,Y^{t-1})$, and since $\frac{1}{x}<0$, we have $H^{'}(x, t,Y^{t-1})>0$. Finally, for $x=0$, we have  
\begin{equation*}
    H^{'}(0, t,Y^{t-1}) = \frac{1}{2}\,G^{'}(0, t,Y^{t-1})
       = \frac{1}{2}\int_{-\infty}^{+\infty} \mathfrak{q}_{\nu}(-z;0,\sigma)\phi_{F(\cdot,t,Y^{t-1})}(\ud z) > 0.
\end{equation*}
Putting together the previous observations, we obtain $F^{'}(x, t,Y^{t-1})>0,\,\forall x \in \mathbb{R}$.
\subsection{Proof of the existence of a fixed-point for the Student-\texorpdfstring{$t$}{t} operator}\label{subsec::proof_existence_fixed_point}
The proof of Theorem \ref{th::existence_of_a_fixed_point} involves several steps given by the propositions and lemmas that follow.
\begin{proposition}\label{prop:compact-class}
For every admissible choice of \(R,L,\varepsilon_t,\underline c_t^\pm,\overline c_t^\pm\), the class \(\mathcal K_{t,\nu}\) in Definition~\ref{def:template-class} is non-empty, convex, and
compact in \(C_b(\mathbb R)\) under the sup norm.
\end{proposition}
\begin{proof}
We start by proving that the class $\mathcal K_{t,\nu}$ is non empty. We choose any $c^+ \in (\underline c_t^+, \overline c_t^+)$ and $c^- \in (\underline c_t^-, \overline c_t^-)$, and define a function $g_0$ by $g_0(x)=M-c^+(1+x)^{\rho_t^+}$ for $x \geq 2 R$ and $g_0(x)=m+c^-(1+|x|)^{\rho_t^-}$ for $x \leq - 2R$ with a smooth interpolation on $[-2R, 2R]$ so that $m \leq g_0 \leq M$ and $\operatorname{Lip}(g_0)\leq L$. Then $g_0 \in \mathcal K_{t,\nu}$. Second, we prove convexity of $\mathcal K_{t,\nu}$. We let $g_1, g_2 \in \mathcal K_{t,\nu}$ and $\lambda \in [0,1]$, and we set $g:=\lambda g_1 + (1-\lambda) g_2$. Then $m \leq g \leq M$ and $\operatorname{Lip}(g)\leq L$. We now define $c^\pm:=\lambda c^\pm_{g_1}+(1-\lambda)c^\pm_{g_2}$, which belongs to $[\underline c_t^\pm, \overline c_t^\pm]$. For $x \geq R$,
\begin{equation*}
    M-g(x)=\lambda c^+_{g_1}(1+x)^{\rho_t^+}(1+e_1(x))+(1-\lambda)c^+_{g_2}(1+x)^{\rho_t^+}(1+e_2(x)),
\end{equation*}
where $|e_i(x)| \leq \varepsilon_t(x)$. Hence 
\begin{equation*}
\left| \frac{M-g(x)}{c^{+}(1+x)^{\rho_t^{+}}}-1 \right| \leq \varepsilon_t(x).
\end{equation*}
The left tail is similar. Therefore, $g \in \mathcal K_{t,\nu}$. Finally, we prove compactness. The class is uniformly bounded and equi-continuous. In addition, for $x \geq R_0 \geq R$, we have
\begin{equation*}
    \sup_{g \in \mathcal K_{t,\nu}} |M-g(x)| \leq \overline c_t^+(1+\varepsilon_t(R_0))(1+R_0)^{\rho_t^+},
\end{equation*}
which converges to zero as $x\rightarrow+\infty$ since $\rho_t^{+}<0$. Similarly, $\sup_{g \in \mathcal K_{t,\nu}} |g(x)-m|$ converges to zero as $x\rightarrow -\infty$. The conclusion follows by an application of Ascoli-Arzelà on compacts set, together with the common tail envelopes. 
\end{proof}

\begin{lemma}\label{lem:student-convolution}
Let \(f:\mathbb{R}\to(0,\infty)\) be locally bounded, ultimately monotone, and regularly
varying at \(+\infty\) with index \(\eta<0\).
\begin{enumerate}
\item[(i)] If \(-\eta<\nu\), then $\int_{-\infty}^{+\infty}\mathfrak q_\nu(x-z;0,\sigma)\,f(z)\,\ud z \sim f(x)$, as $x\to+\infty$. 
\item[(ii)] If, in addition, \(\eta\in(-1,0)\), then $\int_{-\infty}^{+\infty}\bar{\mathfrak q}_\nu(x,z;0,\sigma)\,f(z)\,\ud z
\sim
\frac{1}{1+\eta}\,f(x)$ as $ x\to+\infty$.
\end{enumerate}
\end{lemma}
\begin{proof}
We start from point $(i)$. By assumption, there exists a constant $A>0$ such that for all large $x$ and all $|u| \leq A$, 
\begin{equation*}
    \left|\frac{f(x-u)}{f(x)}-1\right|<\varepsilon.
\end{equation*}
We now split the integral into the sum of $I_1(x)$ and $I_2(x)$, where 
\begin{equation*}
    I_1(x):=\int_{|x-z|\leq A}\mathfrak q_{\nu}(x-z;0,\sigma)f(z)\ud z,\,\,\text{and}\,\,I_2(x):=\int_{|x-z| > A}\mathfrak q_{\nu}(x-z; 0,\sigma)f(z)\ud z
\end{equation*}
We have that 
\begin{equation*}
    I_1(x)=f(x)\left(\int_{|u| \leq A} \mathfrak q_{\nu}(u; 0,\sigma)\,\ud u + o(1) \right).
\end{equation*}
Instead, as regards as $I_2(x)$, we further split it into $|x-z|>A, z \geq x/2$ and $|x-z|>A, z < x/2$. On the first region, one can choose $A$ sufficiently large so that the integral is $<\varepsilon$. On the other hand, on the region $z<x/2$, one has $|x-z|>x/2$, and therefore $\mathfrak q_{\nu}(x-z; 0,\sigma) \leq C x^{-(\nu+1)}$. Because of the assumptions on $f$, this part is at most $C x^{-(\nu+1)}\int_{0}^{x/2}f(z)\,\ud z$, and, by using the fact that $f$ is regularly varying at $+\infty$ with index $\eta<0$, Karamata's theorem gives
\begin{equation*}
    \int_{0}^{x/2}f(z)\,\ud z=
    \begin{cases}
        O(x^{1+\eta}),\quad \eta > -1\\
        O(1),\quad\quad\,\,\, \eta < -1.
    \end{cases}
\end{equation*}
The conclusion follows from the fact that $-\eta<\nu$. We note that of $\eta=-1$, the integral is slowly
varying/logarithmic, and the contribution remains $o(f(x))$ under
$-\eta < \nu$.

We now consider point $(ii)$. By using the definition of $\bar{\mathfrak q}$ we write, for $x>0$,
\begin{equation*}
    \int_{-\infty}^{+\infty}\mathfrak{\bar{q}}(x,z;0,\sigma)f(z)\,\ud z = \frac{1}{x}\int_{0}^{x}\left(\int_{-\infty}^{+\infty}\mathfrak{q}(x-z;0,\sigma)f(z)\,\ud z\right)\,\ud u \sim \frac{1}{1+\eta}f(x),
\end{equation*}
where in the last step we used part $(i)$ together with the fact that $f$ is regularly varyng at $+\infty$ with index $\eta \in (-1,0)$. 
\end{proof}
\begin{lemma}\label{lem:tail-action-phi}
Let \(g\in \mathcal K_{t,\nu}\), where $\mathcal K_{t,\nu}$ is defined in Definition \ref{def:template-class}. Then there exists a continuous decreasing function
\(\widetilde\varepsilon_t(x)\downarrow 0\) such that, uniformly over \(g\in\mathcal K_{t,\nu}\),
\begin{equation*}
M-\phi_g^+(x)
=
\beta_t^+\,c_g^+(1+x)^{\rho_t^+}\bigl(1+\theta_g^+(x)\bigr),
\qquad
|\theta_g^+(x)|\le \widetilde\varepsilon_t(x),
\qquad x\ge R,
\end{equation*}
and
\begin{equation*}
\phi_g^-(x)-m
=
\beta_t^-\,c_g^-(1+|x|)^{\rho_t^-}\bigl(1+\theta_g^-(|x|)\bigr),
\qquad
|\theta_g^-(|x|)|\le \widetilde\varepsilon_t(|x|),
\qquad x\le -R,
\end{equation*}
where $\phi_g^{+}$ and $\phi_g^{-}$ are defined in \eqref{eq::g_mappings}. 
\end{lemma}
\begin{proof}
We provide the proof for the ask-side since for the bid-side the proof is similar. By Definition \ref{def:template-class}, for $x \geq R$, we have
\begin{equation*}
    \delta_{g}^{+}(x):=M-g(x)=c_g^+(1+x)^{\rho_t^+}(1+e_g(x)),\quad |e_g(x)| \leq \varepsilon_t(x).
\end{equation*}
By Assumption \ref{ass:endpoint-regularity}, 
\begin{equation*}
\Pi_t^+(g(z))=C_{t,+}\delta_g^{+}(z)^{\kappa_t^+}(1+o(1))\,\,\text{and}\,\,M-\Psi_t^+(g(z))\Pi_t^+(g(z))=\beta_t^+ C_{t,+} \delta_g^+(z)^{\kappa_t^+ +1}(1+o(1)),
\end{equation*}
uniformly over $g \in \mathcal{K}_{t,\nu}$. Now, 
\begin{equation}\label{eq::equation_error}
    M-\phi_g^+(x)=\frac{\int_{-\infty}^{+\infty}(M-\Psi_t^+(g(z)))\Pi_t^+(g(z))\mathfrak q_{\nu}(x-z; 0, \sigma)\,\ud z}{\int_{-\infty}^{+\infty}\Pi_t^+(g(z))\mathfrak q_{\nu}(x-z; 0, \sigma)\,\ud z} \sim \beta_t^+ c_g^+ (1+x)^{\rho_t^+},
\end{equation}
where in the last step we used the expansion above together with Lemmas \ref{lem:candidate-exponents} and \ref{lem:student-convolution} applied to the powers $\kappa_t^+(-\rho_t^+)<\nu$ and $(\kappa_t^+ +1)(-\rho_t^+)<\nu$. In particular, the error in \eqref{eq::equation_error} can be made uniform over $g \in \mathcal K_{t,\nu}$ because $c_g^+$ ranges in a compact interval and $\varepsilon_t(x)$ converges to zero as $x \rightarrow \infty$. This concludes the proof. 
\end{proof}
\begin{lemma}\label{lemma::finiteness_integral}
    Assume that $\nu > 0$ and $\sigma>0$. Then the integral 
    \begin{equation*}
        I_{\nu} := \int_{-\infty}^\infty \left|\frac{\partial}{\partial u}\mathfrak{q}_{\nu}(u; 0, \sigma)\right|\,\ud u =\int_{-\infty}^{\infty} \frac{(\nu+1)|u|}{\nu\sigma^2 + u^2} \, \mathfrak{q}_{\nu}(u; 0, \sigma)\, du
    \end{equation*}
    is finite and strictly positive.
\end{lemma}
\begin{proof}
    We split the integral into two parts, $|u| \leq 1$ and $|u|>1$. In  the region $|u|>1$, we have
    \begin{equation*}
         \frac{(\nu+1)|u|}{\nu\sigma^2 + u^2} \le (\nu+1) \cdot \frac{1}{|u|}\,\,\,\,\text{and}\,\,\,\,\mathfrak{q}_{\nu}(u;0,\sigma) \leq \frac{C(\nu)}{|u|^{\nu+1}}\,\,\text{for large $|u|$},
    \end{equation*}
    where $C(\nu)$ is a positive constant. Their product is $O\big(|u|^{-(\nu+2)}\big)$ as $|u|\to\infty$, which is integrable since $\nu+2>1$. On the other hand, the integrability in the region $|u|\leq 1$ is implied by the continuity and boundedness of the integrand. Regarding positivity, it is sufficient to observe that the integrand  $\left|\frac{\partial}{\partial u}\mathfrak{q}_{\nu}(u; 0, \sigma)\right|$ is strictly positive for every $u \neq 0$ and, in particular, on any interval $[-1,-\delta] \cup [\delta, 1]$ with $\delta>0$.
\end{proof}
\begin{proposition}\label{ass:compact-invariant-class}
Under Assumptions~\ref{ass:posterior-regularity} and~\ref{ass:endpoint-regularity}, there exist parameters $R>1, L>0$, $0<\underline c_t^\pm<\overline c_t^\pm<\infty$, and a decreasing function $\varepsilon_t(x)\downarrow 0$ such that the class $\mathcal K_{t,\nu}$ from Definition~\ref{def:template-class} satisfies $\mathcal T_{t,\nu}(\mathcal K_{t,\nu})\subset \mathcal K_{t,\nu}$.
\end{proposition}
\begin{proof}
Let $g \in \mathcal{K}_{t,\nu}$. Since $\phi_g(z) \in [m,M]$ and the kernels in \ref{eq::Student_fixed_point_operator} are probability kernels in $z$, we have that
$m \leq \mathcal{T}_{t,\nu}g(x) \leq M$ for every $x$. As regards the Lipschitz bound, we first prove that $(\mathcal{T}_{t,\nu})$ has a derivative that is bounded by a finite constant. 
We first compute
\begin{align}
     \frac{\partial}{\partial x}\mathfrak{q}_\nu(z-x; 0, \sigma)
     &=-\mathfrak{q}_\nu(z-x; 0, \sigma)\frac{(\nu+1)(x-z)}{\nu\sigma^2+(x-z)^2},\label{eq:derivatives_t_densities1}\\
    \frac{\partial \bar{\mathfrak{q}}_\nu(x, z; 0, \sigma)}{\partial x} 
    &= \frac{\mathfrak{q}_\nu(x-z;0,\sigma) - \bar{\mathfrak{q}}_\nu(x,z; 0, \sigma)}{x}\\
    &= \frac{1}{x^2} \int_0^x \left\{\mathfrak{q}_\nu(x - z;0,\sigma)-\mathfrak{q}_\nu(y - z;0,\sigma)\right\}\, \ud y\nonumber \\
    &= \frac{1}{x^2} \int_0^x u\, \frac{\partial}{\partial u}\mathfrak{q}_\nu(u - z;0,\sigma)\,\ud u.\label{eq:derivatives_t_densities2}\nonumber
\end{align}
Then, we have
\begin{align*}
    &\left| \frac{d}{dx}(\mathcal{T}_{t,\nu}g)(x)\right| 
    \leq \int_{-\infty}^{\infty} |\phi_{g}(z)| \left| \frac{1}{N_t} \frac{\partial}{\partial x}\mathfrak{q}_\nu(x-z; 0, \sigma) + \frac{N_t-1}{N_t} \frac{\partial}{\partial x} \bar{\mathfrak{q}}_\nu(x,z;0,\sigma) \right| \ud z \\
    &\leq \bar{M} \left\{ \frac{1}{N_t} \int_{-\infty}^{\infty} \left|\frac{\partial}{\partial x}\mathfrak{q}_\nu(z-x; 0, \sigma) \right| \ud z 
    + \frac{N_t-1}{N_t x^2} \int_0^x |u|\, \ud \int_{-\infty}^{\infty} \left| \frac{\partial}{\partial u}\mathfrak{q}_\nu(u-z; 0, \sigma)\, \ud z \right|\right\} \\
    &= \bar{M} \left\{ \frac{1}{N_t} I_{\nu} + \frac{N_t-1}{N_t x^2} \int_0^x |u| \cdot I_{\nu}\, du \right\} \\
    &= \bar{M} \left\{ \frac{1}{N_t} I_{\nu} + \frac{N_t-1}{2N_t} I_{\nu} \right\} \\
    &= \bar{M} \cdot \frac{N_t+1}{2N_t} \cdot I_{\nu} \leq \bar{M}  \cdot \frac{N_t+1}{2N_t} \cdot I_{\nu},
\end{align*}
with $\bar{M}:=\max(|m|,|M|)$, and where in the last step we have used Lemma \ref{lemma::finiteness_integral}. Therefore, if $L$ is chosen above the previous bound $\mathcal{T}_{t,\nu}g$ is $L$-Lipschitz. Finally, by Lemma \ref{lem:tail-action-phi}, $M-\phi_g^+(x)\sim \beta_t^+ c_g^+ (1+x)^{\rho_t^+}$. Hence, by Lemma \ref{lem:student-convolution}-$(i)$, $M-G_g(x) \sim \beta_t^+ c_g^+ (1+x)^{\rho_t^+}$, where $G_g(x)=\int_{-\infty}^{+\infty}\mathfrak q_{\nu}(x-z;0,\sigma)\phi_g(z)\,\ud z$.  By Lemma \ref{lem:student-convolution}-$(ii)$, instead,  $M-H_g(x) \sim \frac{\beta_t^+}{1+\rho_t^+} c_g^+ (1+x)^{\rho_t^+}$, where $H_g(x):=\frac{1}{x}\int_{0}^{x} G_g(u)\,\ud u$. We have
\begin{equation*}
    M-T_{t,\nu}g(x) \sim \beta_t^{+}\left(\frac{1}{N_t} + \frac{N_t-1}{N_t (1+\rho_t^+)}\right)c_g^+(1+x)^{\rho_t^+} = c_g^+(1+x)^{\rho_t^+},
\end{equation*}
where in the last step we used Lemma \ref{lem:candidate-exponents}. The conclusion follows by choosing $R$ large enough and $\varepsilon_t$ slowly decreasing; the computations for the left-tail are similar. 
\subsection{Proof of Theorem \ref{th::existence_of_a_fixed_point}}
By Proposition \ref{prop:compact-class}, $\mathcal{K}_{t,\nu}$ is non-empty, convex, and compact in the sup norm. By Proposition \ref{ass:compact-invariant-class}, $\mathcal{T}_{t,\nu}(\mathcal{K}_{t,\nu}) \subset \mathcal{K}_{t,\nu}$.  It remains to show the continuity of $\mathcal{T}_{t,\nu}$ on $\mathcal{K}_{t,\nu}$. To this end, let $g_n \rightarrow g$ uniformly in $\mathcal{K}_{t,\nu}$. Since $\Phi_t^{\pm}$ and $\Phi_t^{\pm}$ are continuous on $[m,M]$, $\Phi_t^{\pm}(g_n(z))\rightarrow \Phi_t^{\pm}(g(z))$ and $\Pi_t^{\pm}(g_n(z))\rightarrow \Pi_t^{\pm}(g(z))$ for every $z$. They are uniformly bounded, so dominated convergence gives $\phi_{g_n}(x)\rightarrow\phi_{g}(x)$  for every $x$, and another application of the dominated-convergence step then gives $\mathcal T_{t,\nu}g_n(x) \rightarrow T_{t,\nu}g(x)$ for every $x$. Because $\mathcal K_{t,\nu}$ is compact in the sup norm and $T_{t,\nu} g_n \in \mathcal K_{t,\nu}$, every subsequence has a uniformly convergent further subsequence. Pointwise convergence forces the uniform limit to be $\mathcal{T}_{t,\nu}g$, from which the continuity follows. Schauder’s fixed-point theorem now yields a fixed-point in $\mathcal K_{t,\nu}$. As regards the derivative bound, we refer the reader to the proof of Proposition \ref{ass:compact-invariant-class}.
\end{proof}

\section{Proofs of the Bayesian learning from order flow and related supporting  results}\label{app::consistency_beliefs}
\subsection{Supporting results for the dynamics of Bayesian updating with dependent data \texorpdfstring{ (\cite{shalizi2009dynamics})}{ (Shalizi, 2009)}}\label{subsec::results_posterior}
Let $\Theta$ be a compact subset of $\mathbb{R}$ be the parameter space, $(\Theta, \mathcal{T})$ a measurable space,  $Y_1, Y_2, \ldots,$ for short $Y^{\infty}$ be a sequence of random variables taking values in the measurable space $(\mathcal{Y},\mathcal{A})$; for every $t \in \mathbb{N}$, $(\mathcal{Y}^{t},\mathcal{A}^{t})$ denotes the corresponding product space, where $\mathcal{A}^{t}$ is the product $\sigma$-algebra. The natural filtration of the process is $\sigma(Y^{t})$. It will be clear from the context when $Y_t$ refers to a random variable or the corresponding realization. The sequence $Y^{\infty}$ is generated under a parameter $\theta_0 \in \Theta$, and the true joint probability and probability density function are denoted by $\mathbb{P}_{\theta_0}^{\infty}$  and $\mathfrak{p}_{\theta_0}(Y^{\infty})$, respectively. Instead, the model joint probability and probability density function are denoted by $\mathbb{P}_{\theta}^{\infty}$ and $\mathfrak{f}_{\theta}(Y^{\infty})$, respectively. Analogous notations hold for a fixed $t$. Let $\Pi_0$ be the prior probability measure on $\Theta$, and $\Pi_t(\cdot)$ the posterior probability measure (on $(\Theta, \mathcal{T})$) after observing $Y^{t}$. Then, the Bayesian updating can be written, for a set $A$ as
\begin{equation}\label{eq::update_priors}
    \Pi_t(A) = \frac{\int_{A} \mathfrak{f}_{\theta}(Y^{t})\Pi_0(\ud \theta)}{\int_{\Theta} \mathfrak{f}_{\theta}(Y^{t})\Pi_0(\ud \theta)}=\frac{\int_{A} \frac{\mathfrak{f}_{\theta}(Y^{t})}{\mathfrak{p}_{\theta_0}(Y^{t})}\Pi_0(\ud \theta)}{\int_{\Theta} \frac{\mathfrak{f}_{\theta}(Y^{t})}{\mathfrak{p}_{\theta_0}(Y^{t})}\Pi_0(\ud \theta)} :=\frac{\Pi_0\left(\frac{\mathfrak{f}_{\theta}(Y^{t})}{\mathfrak{p}_{\theta_0}(Y^{t})}A\right)}{\Pi_0\left(\frac{\mathfrak{f}_{\theta}(Y^{t})}{\mathfrak{p}_{\theta_0}(Y^{t})}\right)},
\end{equation}
where $R_t(\theta):=\frac{\mathfrak{f}_{\theta}(Y^{t})}{\mathfrak{p}_{\theta_0}(Y^{t})}$ is the ratio of the model likelihood to true likelihood. Notice that in Section \ref{section::Equilibrium} we denote by $\mathbb{P}_{t}(\cdot)$ the liquidity suppliers' belief \emph{at the beginning of trading round} $t$, i.e., conditional on the public history $Y^{t-1}$; see \eqref{eq::beliefsupdate}. In this section, we follow \cite{shalizi2009dynamics} and index posteriors by the number of observations, writing $\Pi_t(\cdot)=\mathbb{P}(V \in \cdot \mid Y^{t})$. As a consequence, $\Pi_t$ corresponds to $\mathbb{P}_{t+1}(\cdot)$ in the notation of 
Section \ref{section::Equilibrium}.

We give the following useful definition(s).
\begin{definition}\label{def::KL}
    Let $\theta_0 \in \Theta \subset \mathbb{R}$ with $\Theta$ compact be the true parameter. The Kullback-Leibler (KL) divergence rate between any $\theta \in \Theta$ and $\theta_0$ is given by
    \begin{equation}\label{eq::KL_divergence}
        K(\theta, \theta_0)=\lim_{t\rightarrow\infty}\frac{1}{t}\mathbb{E}^{\theta_0}\left[\log\left(\frac{\mathfrak{p}_{\theta_0}(Y^{t})}{\mathfrak{f}_{\theta}(Y^{t})}  \right)\right].
    \end{equation}
For any $\epsilon>0$, the KL-neighbourhood of $\theta_0$ is defined as
    \begin{equation}\label{eq::KL_neighbourhood}
        K_{\epsilon}(\theta_0):=\{\theta\in\Theta\,:\,K(\theta, \theta_0)\leq\epsilon\}.
    \end{equation}
Finally, $\theta_0$ is said to lie in the KL-support of a prior probability measure $\Pi_0$ if for every $\varepsilon>0$ we have that $\Pi_0(K_{\epsilon}(\theta_0))>0$.
\end{definition}
\noindent We now give the following set of assumptions (cf. \cite{shalizi2009dynamics}).
\begin{assumption}\label{ass::assumption_consistency_zero}
    The likelihood ratio $R_t(\theta)$ is $\sigma(Y^{t})\times\mathcal{T}$-measurable for all $t$. 
\end{assumption}
\noindent Assumption \ref{ass::assumption_consistency_zero} implies that prices and learning use only public information; said differently, the likelihood ratio is a function of the publicly observed history.
\begin{assumption}\label{ass::assumption_consistency_one}
    The map $\theta \rightarrow K(\theta, \theta_0)$ is lower semi-continuous, $K(\theta,\theta_0)=0$ if and only if $\theta=\theta_0$, and the limit in \eqref{eq::KL_divergence} exists and it is finite for every $\theta \in \Theta$, and is $\mathcal{T}$ measurable.
\end{assumption}
\noindent Economically, the previous assumption has the following two motivations. First, different fundamentals produce observably different order-flow distributions in the long run; otherwise no learning is possible. Second, small parameter changes do not create discontinuous jumps in distance from truth, which is a stability property of the inference problem.
\begin{assumption}\label{ass::assumption_consistency_two}
    For every $\varepsilon>0$ we have that $\Pi_0(K_{\varepsilon}(\theta_0))>0$ (cf. Definition \ref{def::KL}).  
\end{assumption}
\noindent In our framework, the previous assumption has a precise economic motivation: liquidity suppliers are not dogmatically excluding the true value. The following assumption, instead, is a law-of-large-numbers statement, in the sense that in the long run the data reveals how well each candidate value explains the aggregate order flow.

\begin{assumption}\label{ass::assumption_consistency_three}
    For every $\theta \in \Theta$, the empirical log-likelihood ratio converges to KL-divergence rate $\mathbb{P}_{\theta_0}^{\infty}$-\emph{a.s.}, i.e., 
    \begin{equation}\label{eq::convergence_rate}
        \lim_{t\rightarrow\infty}\frac{1}{t}\log\left(\frac{\mathfrak{f}_{\theta}(Y^{t})}{\mathfrak{p}_{\theta_0}(Y^{t})}  \right)=-K(\theta,\theta_0).
    \end{equation} 
\end{assumption}
\begin{assumption}\label{ass::assumption_consistency_four}
    There exists a sequence of sets $G_t \rightarrow \Theta$, such that
    \begin{itemize}
        \item $\Pi_0(G_t) \geq 1-\alpha \exp(-\beta t)$ for some $\alpha, \beta>0$.
        \item The convergence in \eqref{eq::convergence_rate} is uniform in $\theta$ over $G_t$.
    \end{itemize}
\end{assumption}
\noindent The previous assumption rules out extreme candidate values that the prior essentially doesn’t entertain and requires stable learning on the economically relevant part of the parameter space.\\

We have the following theorem.
\begin{theorem}\label{th:theorem41}
    Let Assumptions \ref{ass::assumption_consistency_zero}--\ref{ass::assumption_consistency_four} hold. Let $\Theta \subset \mathbb{R}$ be compact, and $A \subset \Theta$ be measurable with $\Pi_0(A)>0$ and $\inf_{\theta \in A} K(\theta,\theta_0)>\delta$ for some $\delta>0$. Then $\Pi_t(A) \rightarrow 0$ as $t \rightarrow \infty$, $\mathfrak{p}_{\theta_0}^{\infty}$-$a.s$.
\end{theorem}

\subsection{Proof of Lemma \ref{lem:KL_separation}}\label{app::proof_KL_separation}
\begin{proof}
(i) Since $F(\cdot,t,Y^{t-1})$ is strictly increasing, its inverse $m_t(\cdot)$ is strictly increasing, hence $v\neq v_0$ implies $m_t(v)\neq m_t(v_0)$.
(ii) Under $v_0$ we have $Y_t=m_t(v_0)+Z_t$ with $Z_t\sim \mathtt{T}_\nu(0,\sigma)$ independent of $Y^{t-1}$, so the conditional density of $Y_t$ given $Y^{t-1}$ is the Student-$t$ location family with location $m_t(\cdot)$. Identifiability of the location family implies $k(\Delta)>0$ for $\Delta\neq 0$, and continuity follows by dominated convergence.
(iii) The bound $F'\le K_0$ yields
$|v-v_0|=|F(m_t(v),t,Y^{t-1})-F(m_t(v_0),t,Y^{t-1})|\le K_0|m_t(v)-m_t(v_0)|$.
The conclusion follows since $k$ is continuous and strictly positive away from $0$.
\end{proof}
\subsection{Proof of Proposition \ref{lemm::Lemma_uniform_convergence_loglikelihood_ratio}}\label{app::proof_Lemma_uniform_convergence_loglikelihood_ratio}
This section proves Proposition \ref{lemm::Lemma_uniform_convergence_loglikelihood_ratio}; the proof is divided into several steps. First, we start with the following lemma. 
\begin{lemma}\label{lem:log_envelope}
Fix $\nu>2$ and $\sigma>0$, and set $a:=\nu\sigma^2$. Let $\mathfrak{q}_{\nu}(\cdot; \sigma,\mu)$ be the density function of a location-scale Student's $t$ distribution. Then, for all $y\in\mathbb R$ and all $\mu,\mu_0\in\mathbb R$ we have
\begin{equation*}
\left|\log\left(\frac{\mathfrak{q}_{\nu}(y; \mu,\sigma)}{\mathfrak{q}_{\nu}(y; \mu_0, \sigma)}\right)\right|
\le C_1 + C_2 \log\left(1+\frac{(\mu-\mu_0)^2}{a}\right),
\end{equation*}
where $C_1$ and $C_2$ are positive constants depending only on $(\nu,\sigma)$.
\end{lemma}
\begin{proof}
We write $y=\mu_0+z$ where $z\in\mathbb R$. Up to an additive constant that cancels in the ratio the density function of a location-scale Student's $t$ distribution satisfies 
\begin{equation*}
    \log(\mathfrak{q}_{\nu}(y; \mu_0+z,\sigma))=-\frac{(\nu+1)}{2}\log\left(a+(z-(\mu-\mu_0))^2\right)+C(\nu,\sigma),
\end{equation*}
for some constant $C(\nu,\sigma)$. By setting $\Delta:=\mu-\mu_0$, we obtain 
\begin{equation*}
    \log\left(\frac{\mathfrak{q}_{\nu}(y; \mu,\sigma)}{\mathfrak{q}_{\nu}(y; \mu_0,\sigma)}\right)=-\frac{\nu+1}{2}\log\!\left(\frac{a+(z-\Delta)^2}{a+z^2}\right)
\end{equation*}
Now, we use the following elementary bound $a+(z-\Delta)^2 \le 2(a+z^2)+2\Delta^2 = 2(a+z^2)\left(1+\Delta^2/a\right)$ and write
\begin{equation*}
\frac{a+(z-\Delta)^2}{a+z^2}\le 2\left(1+\Delta^2/a\right).
\end{equation*}
Additionally, by symmetry it also holds $\frac{a+z^2}{a+(z-\Delta)^2}\le 2\left(1+\Delta^2/a\right)$. Therefore
\begin{equation*}
    \left|\log\!\left(\frac{a+(z-\Delta)^2}{a+z^2}\right)\right|
\le \log 2 + \log\!\left(1+\Delta^2/a\right).
\end{equation*}
The conclusion follows by setting $C_1:=\frac{\nu+1}{2}\log 2$ and $C_2:=\frac{\nu+1}{2}$.
\end{proof}

\noindent At this point, we remind that $Y_t | (Y^{t-1}, V=v) \overset{d}{\sim} \texttt{T}_{\nu}(F^{-1}(v,t,Y^{t-1}), \sigma)$, and, in order to ease the notation, we denote the corresponding density function $\mathfrak{q}_{\nu}(\cdot; F^{-1}(v,t,Y^{t-1}),\sigma)$ by $\mathfrak{f}_{Y_t | Y^{t-1}, v}(\cdot)$, to highlight the role of $v \in \Theta_{\varepsilon}$. We define 
\begin{equation*}
    \ell_t(v):=\log\left(
\frac{\mathfrak{f}_{Y_t | Y^{t-1}, v}(Y_t)}{\mathfrak{f}_{Y_t | Y^{t-1}, v_0}(Y_t)}\right),\quad\text{and}\quad D_t(v):=\ell_t(v)-\mathbb E_{t-1}^{v_0}[\ell_t(v)].
\end{equation*}
In particular, $(D_t(v))_{t\ge1}$ is a martingale difference sequence with respect to $\sigma(Y^t)$. We now prove that
\begin{equation}\label{eq::suptendtozero}  
    \sup_{v \in \Theta_\varepsilon}\Big|\frac{1}{t}\sum\limits_{s=1}^{t} D_s(v)\Big| \rightarrow 0,\quad \mathbb{P}^{v_0}-a.s.
\end{equation}
The proof of the result in the previous equation is divided in three steps.

\indent\emph{Step 1} We set $m_t(v):=F^{-1}(v,t,Y^{t-1})$ and $\Delta_t(v):=m_t(v)-m_t(v_0)$. Lemma \ref{lem:log_envelope} implies that there exist positive constants $C_1$ and $C_2$ such that   
\begin{equation*}
    |\ell_t(v)| \le C_1 + C_2\log\left(1+\frac{\Delta_t(v)^2}{\nu\sigma^2}\right)
\end{equation*}
In particular, on $\Theta_\varepsilon$ we have $|\Delta_t(v)|\le L_{\varepsilon,t}(Y^{t-1})\,|v-v_0|
\le L_t\,\mathrm{diam}(\Theta_\varepsilon)$. Hence:
\begin{equation*}
\sup_{v\in\Theta_\varepsilon}|\ell_t(v)|
\le \bar b_t
:=C_1 + C_2\log\!\left(1+\frac{\mathrm{diam}(\Theta_\varepsilon)^2 L_t^2}{\nu\sigma^2}\right).
\end{equation*}
Therefore, for all $v\in\Theta_\varepsilon$, $|D_t(v)| \le 2\bar b_t$. Assumption \ref{ass:moderate_steepening} implies that $\bar b_t = O(\log t)$.

\indent\emph{Step 2} We observe that for the location-scale Student-$t$ family, we have 
\begin{equation*}
\sup_{y,\mu}\left|\partial_\mu\log \mathfrak{q}_{\nu}(y;  \mu,\sigma)\right|=\frac{\nu+1}{2\sigma\sqrt{\nu}}:=C_0<\infty.
\end{equation*}
Therefore, for all $v,v'\in\Theta_\varepsilon$, $|\ell_t(v)-\ell_t(v')|
\le C_0\,|m_t(v)-m_t(v')|
\le C_0\,L_t\,|v-v'|$, and similarly $|D_t(v)-D_t(v')|\le 2C_0\,L_t\,|v-v'|$.
\indent\emph{Step 3} We fix $\eta>0$ and prove the uniform martingale bound in \eqref{eq::suptendtozero} by discretization of $\Theta_\varepsilon$. In this step, we use Assumption \ref{ass:moderate_steepening} and we assume that
there exist constants $C_L>0$ and $\kappa>0$ such that $L_s\le C_L(1+s)^\kappa$ for all $s\ge1$. We set $\delta_t := t^{-(\kappa+2)}$ and let $\mathcal N_t\subset \Theta_\varepsilon$ be a deterministic $\delta_t$-net of $\Theta_\varepsilon$: for every $v\in\Theta_\varepsilon$ there exists $u\in\mathcal N_t$ such that $|v-u|\le \delta_t$.
Since $\Theta_\varepsilon$ is an interval, we can choose $\mathcal N_t$ so that
\begin{equation*}
|\mathcal N_t|\le 1+\frac{\operatorname{diam}(\Theta_\varepsilon)}{\delta_t}
\le C_\varepsilon\, t^{\kappa+2}
\end{equation*}
for some constant $C_\varepsilon>0$ depending only on $\Theta_\varepsilon$. Now, we fix $u\in\mathcal N_t$. The sequence $(D_s(u))_{s\ge1}$ is a martingale difference sequence and,
by \emph{Step 1}, satisfies $|D_s(u)|\le 2b_s$ for all $s$. Therefore, Azuma-Hoeffding's inequality yields
\begin{equation*}
\mathbb P_{v_0}\!\left(\left|\sum\limits_{s=1}^t D_s(u)\right|\ge t\eta\right)
\le 2\exp\!\left(-\frac{t^2\eta^2}{2\sum\limits_{s=1}^t (2b_s)^2}\right)
=2\exp\!\left(-\frac{\eta^2 t^2}{8\sum\limits_{s=1}^t b_s^2}\right).
\end{equation*}
Under $L_s\le C_L(1+s)^\kappa$, the bound in \emph{Step 1} implies $b_s\le C(1+\log(1+s))$ for some
constant $C>0$, hence $\sum\limits_{s=1}^t b_s^2 \le C' t(\log t)^2$ for some $C'>0$ and all large $t$.
Consequently, for a constant $c>0$,
\begin{equation*}
\mathbb P_{v_0}\!\left(\left|\sum\limits_{s=1}^t D_s(u)\right|\ge t\eta\right)
\le 2\exp\!\left(-c\,\frac{\eta^2 t}{(\log t)^2}\right).
\end{equation*}
A union bound over $\mathcal N_t$ gives
\begin{equation*}
\mathbb P_{v_0}\!\left(\max_{u\in\mathcal N_t}\left|\frac{1}{t}\sum\limits_{s=1}^t D_s(u)\right|\ge \eta\right)
\le 2\,|\mathcal N_t|\exp\!\left(-c\,\frac{\eta^2 t}{(\log t)^2}\right)
\le 2C_\varepsilon\, t^{\kappa+2}\exp\!\left(-c\,\frac{\eta^2 t}{(\log t)^2}\right),
\end{equation*}
which is summable in $t$. By the Borel--Cantelli lemma,
\begin{equation*}
\max_{u\in\mathcal N_t}\left|\frac{1}{t}\sum\limits_{s=1}^t D_s(u)\right|\rightarrow 0,
\qquad \mathbb P_{v_0}\text{-a.s.}
\end{equation*}
Finally, we fix an arbitrary $v\in\Theta_\varepsilon$ and choose $v_t\in\mathcal N_t$ such that
$|v-v_t|\le\delta_t$. By \emph{Step 2},
\begin{equation*}
\left|\frac{1}{t}\sum\limits_{s=1}^t D_s(v)-\frac{1}{t}\sum\limits_{s=1}^t D_s(v_t)\right|
\le \frac{1}{t}\sum\limits_{s=1}^t 2C_0L_s\,|v-v_t|
\le 2C_0\delta_t \frac{1}{t}\sum\limits_{s=1}^t L_s.
\end{equation*}
Using $L_s\le C_L(1+s)^\kappa$, we have $\frac{1}{t}\sum\limits_{s=1}^t L_s \le C'' t^\kappa$ for some
$C''>0$, hence
\begin{equation*}
2C_0\delta_t  \frac{1}{t}\sum\limits_{s=1}^t L_s
\le C\, t^{-(\kappa+2)}  t^\kappa
= C\,t^{-2}\rightarrow 0.
\end{equation*}
Combining with the net bound above yields
\begin{equation*}
\sup_{v\in\Theta_\varepsilon}\left|\frac{1}{t}\sum\limits_{s=1}^t D_s(v)\right|\rightarrow 0,
\qquad \mathbb P_{v_0}\text{-a.s.},
\end{equation*}
which is exactly \eqref{eq::suptendtozero}.

At this point, to prove Proposition \ref{lemm::Lemma_uniform_convergence_loglikelihood_ratio}, we write $\log R_t(v)=\sum\limits_{s=1}^t \ell_s(v)$ and decompose
\begin{equation*}
\ell_s(v)=\mathbb E_{s-1}^{v_0}[\ell_s(v)] + D_s(v).
\end{equation*}
Then, we define the one-step conditional KL increment
\begin{equation*}
K_s(v,v_0):=-\mathbb E_{s-1}^{v_0}[\ell_s(v)]\ge 0.
\end{equation*}
Then
\begin{equation*}
\frac{1}{t}\log R_t(v) + \frac{1}{t}\sum\limits_{s=1}^t K_s(v,v_0)
=
\frac{1}{t}\sum\limits_{s=1}^t D_s(v).
\end{equation*}
By \emph{Step 3} above the right-hand side converges to $0$ uniformly on $\Theta_\varepsilon$.
Under Assumption~\ref{ass:KL_stabilization} we have
$\sup_{v\in\Theta_\varepsilon}\left|\frac{1}{t}\sum\limits_{s=1}^t K_s(v,v_0)-K(v,v_0)\right|\to 0$,
and therefore
\begin{equation*}
\sup_{v\in\Theta_\varepsilon}\left|\frac{1}{t}\log R_t(v)+K(v,v_0)\right|\to 0,
\qquad \mathbb P_{v_0}^\infty\text{-a.s.}
\end{equation*}

\section{Proofs of the asymptotic price impact and related supporting  results}\label{app::Supporting_Results_Market_Impact}
\subsection{Regularly varying functions and corresponding results}\label{subsec::regular_variation}
\begin{definition}
	A function $g: (0,+\infty) \mapsto (0,+\infty)$ is said to be \textit{regularly varying of index} $\rho$ at $\infty$ if
    \begin{equation*}
        \lim_{\lambda \rightarrow \infty} \frac{g(\lambda x)}{g(\lambda)} = x^{\rho}, \quad \forall x > 0.
    \end{equation*}
	Analogously, a function $g: (-\infty,0)\mapsto (0,+\infty)$ is said to be regularly varying of index $\rho$ at $-\infty$ if $g(-x)$ is regularly varying of index $\rho$ at $+\infty$.\\
	For $\rho=0$, that is, if
	\begin{equation*}
	    	\lim_{\lambda \rightarrow +\infty} \frac{g(\lambda x)}{g(\lambda)} = 1, \quad \forall x > 0,
	\end{equation*}
	$g$ is said to be \textit{slowly varying}.
\end{definition}

We now recall the following results due to Karamata (e.g., \cite{nikolic2018karamata}).
\begin{theorem}\label{thm:KaramataDirect}
Let $g$ be a regularly varying function with index $\rho$ at $\infty$, and be locally bounded in $[X,\infty)$ for a certain $X \in \mathbb{R}$. Then
		\begin{enumerate}
			\item for any $\sigma > -(\rho+1)$,
			\begin{equation}
			\frac{x^{\sigma+1}g(x)}{\int_{X}^{x}t^{\sigma}g(t)\ud t}\rightarrow \sigma+\rho+1,\quad x\rightarrow \infty;
			\end{equation}
			\item for any $\sigma<-(\rho+1)$,
			\begin{equation}
			\frac{x^{\sigma+1}g(x)}{\int_{x}^{\infty}t^{\sigma}g(t)\ud t}\rightarrow -(\sigma+\rho+1),\quad x\rightarrow \infty.
			\end{equation}
		\end{enumerate}
    The equality case requires separate slowly varying/logarithmic treatment.
\end{theorem}
\begin{theorem}\label{thm:KaramataConverse}.
		Let $g$ be positive and locally integrable in $[X,\infty)$ for a certain $X \in \mathbb{R}$
		\begin{enumerate}
			\item If for some $\sigma>-(\rho+1)$,
			\begin{equation}
			\frac{x^{\sigma+1}g(x)}{\int_{X}^{x}t^{\sigma}g(t)\ud t}\rightarrow \sigma+\rho+1,\quad x\rightarrow\infty,
			\end{equation}
			then $g$ is regularly varying of index $\rho$ at $\infty$.
			\item If for some $\sigma<-(\rho+1)$,
			\begin{equation}
			\frac{x^{\sigma+1}g(x)}{\int_{x}^{\infty}t^{\sigma}g(t)\ud t}\rightarrow -(\sigma+\rho+1),\quad x\rightarrow\infty,
			\end{equation}
			then again $g$ is regularly varying of index $\rho$ at $\infty$.
		\end{enumerate}
\end{theorem}

\subsection{Proof of Theorem \ref{th::power_law_price_impact}}\label{app::price_impact_proof}
We prove the statement for $M-F(\cdot,t,Y^{t-1})$; the one for $F(\cdot,t,Y^{t-1})-m$ is analogous. The result is proved by induction over trading periods. For \(t=1\), the sums over previous periods are empty. For \(t>1\), we assume that the previous-period branches have the regular-variation exponents \(\rho_s^\pm\), \(s<t\), established by the induction hypothesis. Moreover, we assume the previously constructed branch tails are smoothly varying, in the sense that
\begin{equation*}
    \frac{xF'(x,s,Y^{s-1})}{M-F(x,s,Y^{s-1})}\to -\rho_s^+
\end{equation*}
Recall that $\Pi_t^{+}(x)=\mathbb{P}_t(V \geq x|Y^{t-1})=\int_{x}^{M}\mathfrak{p}_{t,V}(v \vert Y^{t-1})\,\ud v$. 
Hence $\partial_x\Pi_t^+(x)=-\mathfrak{p}_{t,V}(x\mid Y^{t-1})$, and by L'H\^opital rule
\begin{equation}\label{eq::first_equation}
    \begin{split}
        \lim_{x \rightarrow M} (M-x)\frac{\partial_x \Pi_t(x)}{\Pi_t(x)} &=  \lim_{x \rightarrow M}\frac{-(M-x)\mathfrak{p}_{t,V}(x \vert Y^{t-1})}{\int_{x}^{M}\mathfrak{p}_{t,V}(v \vert Y^{t-1})\,\ud v}\\
        &=-1+\lim_{x \rightarrow M}\frac{(M-x) \mathfrak{p}_{t,V}^{'}(x\vert Y^{t-1})}{\mathfrak{p}_{t,V}(x\vert Y^{t-1})}\\
        &=-1+\lim_{x \rightarrow M}(M-x)[\log(\mathfrak{p}_{t,V}(x\vert Y^{t-1}))]^{'}
    \end{split}
\end{equation}
By using \eqref{eq::beliefsupdate}, we have
\begin{equation*}
    \log(\mathfrak{p}_{t,V}(x\vert Y^{t-1})) = C + \log(\mathfrak{p}_{V}(x))-\frac{(\nu+1)}{2}\sum\limits_{s=1}^{t-1}\log\left[1+\frac{(Y_s-F^{-1}(x,s,Y^{s-1}))^2}{\nu\sigma^2}\right].
\end{equation*}
It follows, 
\begin{equation*}
\begin{split}
    (M-x)[\log(\mathfrak{p}_{t,V}(x\vert Y^{t-1}))]^{'} &= \frac{(M-x)\mathfrak{p}_{V}^{'}(x)}{\mathfrak{p}_{V}(x)}\\
                                             &-(\nu+1)\sum\limits_{s=1}^{t-1}\frac{M-x}{F^{'}(F^{-1}(x,s,Y^{s-1}))}\frac{F^{-1}(x,s,Y^{s-1})-Y_{s}}{\nu\sigma^2 + (Y_s-F^{-1}(x,s,Y^{s-1}))^2}
\end{split}
\end{equation*}
By inserting the previous expression in \eqref{eq::first_equation}, using the assumption on the fundamental asset value, the properties of $F(\cdot,s,Y^{s-1})$ (i.e., $F(\cdot,s,Y^{s-1}) \uparrow M$ as $x\rightarrow + \infty$), and the fact (to be proved below) that $M-F(x,s,Y^{s-1})$ is regularly varying at $+\infty$ of index $\rho_{s}^{+}$
\begin{equation}\label{eq::first_equation_updated}
    \begin{split}
        &\lim_{x \rightarrow M} (M-x)\frac{\partial_x \Pi_t(x)}{\Pi_t(x)}\\
        &=-1+\lim_{x \rightarrow M}\frac{(M-x)\mathfrak{p}_{V}^{'}(x)}{\mathfrak{p}_{V}(x)}-(\nu+1)\sum\limits_{s=1}^{t-1}\left[\lim_{x \rightarrow +\infty}\frac{1}{(x-Y_s)+\frac{\nu\sigma^2}{x-Y_s}}\frac{M-F(x,s,Y^{s-1})}{F^{'}(x,s,Y^{s-1})}\right]\\
        &=-1+L+(\nu+1)\left[\lim_{x \rightarrow +\infty}\frac{1}{(x-Y_s)+\frac{\nu\sigma^2}{x-Y_s}}\frac{(M-F(x,s,Y^{s-1}))}{(M-F(x,s,Y^{s-1}))^{'}}\right]\\
        &=-1+L+(\nu+1)\sum\limits_{s=1}^{t-1}\lim_{x \rightarrow +\infty}\frac{1}{\rho_s^+}\frac{1}{1-\frac{Y_s}{x}+\frac{\nu\sigma^2}{x(x-Y_s)}}\\
        &=-1+L+(\nu+1)\sum\limits_{s=1}^{t-1}\frac{1}{\rho_s^{+}}.
    \end{split}
\end{equation}
{\color{black}
By using \eqref{eq::derivative_rewriting}, we obtain 
\begin{equation*}
     \partial_x\Psi_t^{+}(M) = \frac{1-L-(\nu+1)\sum\limits_{s=1}^{t-1}\frac{1}{\rho_s^{+}}}{2-L-(\nu+1)\sum\limits_{s=1}^{t-1}\frac{1}{\rho_s^{+}}}:=\frac{A(t)}{1+A(t)}.
\end{equation*}
Recall that $\Psi_t^+(x)=\mathbb{E}_t[V \vert V \geq x ]\geq x$ for all $x<M$, and that $\partial_x \Psi_t^+(M)\geq0$ by definition. It implies
\begin{equation*}
    0\leq\partial_x \Psi_t^+(M)=\lim_{x\rightarrow M}\frac{M-\Psi_t^+(x)}{M-x}\leq \lim_{x\rightarrow M}\frac{M-x}{M-x}=1.
\end{equation*}
Since $\partial_x \Psi_t^+(M)=\frac{A(t)}{1+A(t)}\leq 1$, we must have $1+A(t)>0$. Combining this with $\frac{A(t)}{1+A(t)}\geq 0$ yields $A(t)\geq0$. Therefore, $\partial_x \Psi_t^+(M)\in [0,1)$. Moreover, as $t\rightarrow\infty$ we have $A(t)\rightarrow\infty$, and consequently $\lim_{t\rightarrow+\infty}\partial_x\Psi_t^{+}(M)=1$. 
}

We are left to prove that $M-F(x,s,Y^{s-1})$ is regularly varying at $+\infty$ of index $\rho_{s}^{+}$ as in the statement of the theorem, and that $\lim_{t\rightarrow\infty}|\rho_t^{+}|=0$. We write
\begin{equation*}
    \begin{split}
        &\lim_{\alpha\rightarrow+\infty}\frac{M-F(\alpha x, t, Y^{t-1})}{M-F(\alpha, t, Y^{t-1})}\\
        &=\lim_{\alpha\rightarrow+\infty}\int_{-\infty}^{+\infty}\left[\frac{1}{N_t}\mathfrak{q}_{\nu}(x-z; 0, \sigma)+\frac{N_t-1}{N_t}\bar{\mathfrak{q}}_{\nu}(x, z; 0, \sigma)\right]\frac{M-\phi_{F(\cdot, t, Y^{t-1})}(z)}{M-F(\alpha, t, Y^{t-1})}\, \ud z\\
        &=\lim_{\alpha\rightarrow+\infty}\int_{-\infty}^{+\infty}\left[\frac{1}{N_t}\mathfrak{q}_{\nu}(x-z; 0, \frac{\sigma}{\alpha})+\frac{N_t-1}{N_t}\bar{\mathfrak{q}}_{\nu}(x, z; 0, \frac{\sigma}{\alpha})\right]\frac{M-\phi_{F(\cdot, t, Y^{t-1})}(\alpha z)}{M-F(\alpha, t, Y^{t-1})}\, \ud z\\
    \end{split}
\end{equation*}
\noindent Notice that when $z>0$, $\phi_{F(\cdot, t, Y^{t-1})}(\alpha z)$ can be rewritten as 
\begin{equation}\label{eq::trick_proof}
    \begin{split}
        &\phi_{F(\cdot, t, Y^{t-1})}(\alpha z)=\int_{-\infty}^{+\infty}\Psi_t^{+}(F(\alpha y, t, Y^{t-1}))\Lambda_t(\alpha, z, \ud y)\\
        &\text{with}\quad \Lambda_t(\alpha, z, \ud y):=\frac{\Pi_t^{+}(F(\alpha y, t, Y^{t-1}))\mathfrak{q}_{\nu}(z-y;0,\frac{\sigma}{\alpha})}{\int_{-\infty}^{+\infty}\Pi_t^{+}(F(\alpha u, t, Y^{t-1}))\mathfrak{q}_{\nu}(z-u;0,\frac{\sigma}{\alpha})\,\ud u}\,\ud y,
     \end{split}
\end{equation}   
and it converges to $\Psi_t^{+}(F(\alpha z, t, Y^{t-1}))$ because the probability measure that converges to the point mass at $z$ as $\alpha\rightarrow+\infty$. We now justify the localization step in \eqref{eq::trick_proof}. To simplify notation, we write $F_t(x):=F(x,t,Y^{t-1})$, $\Gamma_t^+(x):=M-F_t(x)$, $\chi_t^+:=\partial_x\Psi_t^+(M)$. Then, we fix \(z>0\) and, along a subsequence \(\alpha\to+\infty\), and suppose that $\gamma(t,y):=\lim_{\alpha\to+\infty}
\frac{\Gamma_t^+(\alpha y)}{\Gamma_t^+(\alpha)}$ exists locally uniformly for \(y\in(0,\infty)\). We prove that
\begin{equation}\label{eq::loc_phi_limit}
    \frac{M-\phi_{F_t}^+(\alpha z)}
{\Gamma_t^+(\alpha)}
\rightarrow
\chi_t^+\gamma(t,z)
\end{equation}
To this end, we let $\mathfrak q_{\nu,\alpha}(r):=
\mathfrak q_\nu\left(r;0,\frac{\sigma}{\alpha}\right)$, and define $D_\alpha(z):=\int_{-\infty}^{+\infty}
\Pi_t^+(F_t(\alpha u))\mathfrak q_{\nu,\alpha}(z-u)\,\ud u$.
Then the probability measure in \eqref{eq::trick_proof} is
\begin{equation*}
\Lambda_t(\alpha,z,\ud y)
=
\frac{
\Pi_t^+(F_t(\alpha y))\mathfrak q_{\nu,\alpha}(z-y)
}{
D_\alpha(z)
}\,\ud y.
\end{equation*}
Moreover, $M-\phi_{F_t}^+(\alpha z)=
\int_{-\infty}^{+\infty}
\left(M-\Psi_t^+(F_t(\alpha y))\right)
\Lambda_t(\alpha,z,\ud y).
\label{eq::loc_M_minus_phi}
$. By Assumption~\ref{ass:endpoint-regularity},
\begin{equation}
\Pi_t^+(M-s)\sim C_{t,+}s^{\kappa_t^+},
\qquad
M-\Psi_t^+(M-s)\sim \chi_t^+s,
\qquad s\downarrow0.
\label{eq::loc_endpoint}
\end{equation}
Since \(F_t(\alpha y)\uparrow M\) locally uniformly on compact subsets of \((0,\infty)\),
\eqref{eq::loc_endpoint} implies that, for every compact \(K\subset(0,\infty)\),
\begin{equation}
\sup_{y\in K}
\left|
\frac{\Pi_t^+(F_t(\alpha y))}
{C_{t,+}\Gamma_t^+(\alpha y)^{\kappa_t^+}}
-1
\right|
\to0,
\label{eq::loc_Pi_uniform}
\end{equation}
and
\begin{equation}
\sup_{y\in K}
\left|
\frac{M-\Psi_t^+(F_t(\alpha y))}
{\Gamma_t^+(\alpha y)}
-\chi_t^+
\right|
\to0.
\label{eq::loc_Psi_uniform}
\end{equation}

We first show that
\begin{equation}
\Lambda_t(\alpha,z,\cdot)\Rightarrow \delta_z.
\label{eq::loc_Lambda_delta}
\end{equation}
Let \(U_\delta:=(z-\delta,z+\delta)\), with \(0<\delta<z/2\). Since
\(0\le \Pi_t^+\le1\),
\begin{equation}
\int_{U_\delta^c}
\Pi_t^+(F_t(\alpha y))\mathfrak q_{\nu,\alpha}(z-y)\,\ud y
\le
\int_{|y-z|\ge\delta}\mathfrak q_{\nu,\alpha}(z-y)\,\ud y.
\end{equation}
The last term is a Student-\(t\) tail probability with scale \(\sigma/\alpha\), hence
\begin{equation}
\int_{|y-z|\ge\delta}\mathfrak q_{\nu,\alpha}(z-y)\,\ud y
=
\mathbb P\left(|\mathtt T_\nu(0,\sigma)|\ge \alpha\delta\right)
=
O(\alpha^{-\nu}).
\label{eq::loc_outside_kernel}
\end{equation}

On the other hand, by \eqref{eq::loc_Pi_uniform} and the local uniform convergence of
\(\Gamma_t^+(\alpha y)/\Gamma_t^+(\alpha)\), there exists \(c_\delta>0\) such that, for all
large \(\alpha\) and all \(u\in U_{\delta/2}\), $\Pi_t^+(F_t(\alpha u))
\ge
c_\delta\,\Pi_t^+(F_t(\alpha z))$.
Therefore, $D_\alpha(z) \ge
c_\delta\,\Pi_t^+(F_t(\alpha z))
\int_{U_{\delta/2}}\mathfrak q_{\nu,\alpha}(z-u)\,\ud u$.
The integral on the right converges to one, so
\begin{equation}
D_\alpha(z)\gtrsim \Pi_t^+(F_t(\alpha z));
\label{eq::loc_D_lower}
\end{equation}
in general, $f \gtrsim g \Longleftrightarrow \exists 0<C<+\infty\,\,\text{s.t.}\,\,f \geq C \cdot g$. Since the branch belongs to the tail-controlled class, we have $\Pi_t^+(F_t(\alpha z)) \asymp \Gamma_t^+(\alpha z)^{\kappa_t^+}$. By Lemma~\ref{lem:candidate-exponents}, $\kappa_t^+(-\rho_t^+)<\nu$, and thus 
\begin{equation}
\alpha^{-\nu}
=
o\!\left(\Pi_t^+(F_t(\alpha z))\right).
\label{eq::loc_tail_compare_one}
\end{equation}
Combining \eqref{eq::loc_outside_kernel}, \eqref{eq::loc_D_lower}, and
\eqref{eq::loc_tail_compare_one}, we get $\Lambda_t(\alpha,z,U_\delta^c)\to0$.
Since \(\delta>0\) is arbitrary, \eqref{eq::loc_Lambda_delta} follows.

We now prove \eqref{eq::loc_phi_limit}. From \eqref{eq::loc_M_minus_phi},
\begin{equation}
\frac{M-\phi_{F_t}^+(\alpha z)}
{\Gamma_t^+(\alpha)}
=
\int
\frac{M-\Psi_t^+(F_t(\alpha y))}
{\Gamma_t^+(\alpha)}
\Lambda_t(\alpha,z,\ud y).
\label{eq::loc_norm_integral}
\end{equation}
On \(U_\delta\), \eqref{eq::loc_Psi_uniform} and the local uniform convergence of
\(\Gamma_t^+(\alpha y)/\Gamma_t^+(\alpha)\) imply
\begin{equation}
\sup_{y\in U_\delta}
\left|
\frac{M-\Psi_t^+(F_t(\alpha y))}
{\Gamma_t^+(\alpha)}
-
\chi_t^+\gamma(t,y)
\right|
\to0.
\label{eq::loc_integrand}
\end{equation}
Hence the contribution from \(U_\delta\) converges to $\chi_t^+\gamma(t,z)$, using \eqref{eq::loc_Lambda_delta}.

It remains to show that the contribution from \(U_\delta^c\) is negligible. Let
\begin{equation}
N_\alpha^{out}(z):=
\int_{U_\delta^c}
\left(M-\Psi_t^+(F_t(\alpha y))\right)
\Pi_t^+(F_t(\alpha y))
\mathfrak q_{\nu,\alpha}(z-y)\,\ud y.
\end{equation}
Since \(M-\Psi_t^+(F_t(\alpha y))\le M-m\) and \(\Pi_t^+\le1\),
\begin{equation}
N_\alpha^{out}(z)
=
O(\alpha^{-\nu}).
\label{eq::loc_Nout}
\end{equation}
Moreover, by \eqref{eq::loc_D_lower}, $\Gamma_t^+(\alpha)D_\alpha(z)
\gtrsim
\Gamma_t^+(\alpha)\Pi_t^+(F_t(\alpha z))
\asymp
\Gamma_t^+(\alpha)^{\kappa_t^++1}$. The last comparison uses the local uniform convergence of
\(\Gamma_t^+(\alpha z)/\Gamma_t^+(\alpha)\). By Lemma~\ref{lem:candidate-exponents}, $(\kappa_t^++1)(-\rho_t^+)<\nu$,
so
\begin{equation}
\alpha^{-\nu}
=
o\!\left(\Gamma_t^+(\alpha)^{\kappa_t^++1}\right).
\label{eq::loc_tail_compare_two}
\end{equation}
Equations \eqref{eq::loc_Nout} and \eqref{eq::loc_tail_compare_two} imply
\begin{equation}
\frac{N_\alpha^{out}(z)}
{\Gamma_t^+(\alpha)D_\alpha(z)}
\to0.
\end{equation}
Therefore the outside contribution is negligible, and \eqref{eq::loc_phi_limit} follows. At this point, for any $\alpha>0$, by mean value theorem, there exists $z^{*} \in [F(\alpha z, t, Y^{t-1}), M]$ such that  
\begin{equation*}
    \lim_{\alpha\rightarrow\infty}\frac{M-\Psi_t^{+}(F(\alpha z, t, Y^{t-1}))}{M-F(\alpha , t, Y^{t-1})}=\partial_x \Psi_t^{+}(z^{*})\frac{M-F(\alpha z, t, Y^{t-1})}{M-F(\alpha , t, Y^{t-1})}.
\end{equation*}
By assuming that $\gamma(t,z):=\lim_{\alpha\rightarrow+\infty}\frac{M-F(\alpha z, t, Y^{t-1})}{M-F(\alpha , t, Y^{t-1})}$ exists, since $F(\alpha z, t, Y^{t-1}) \uparrow M$ as $\alpha \rightarrow + \infty$, so $z^{*}\rightarrow M$, and $\partial_x \Psi_t^{+}(\cdot)$ is continuous in a neighbourhood of $M$, we have
\begin{equation*}
    \lim_{\alpha\rightarrow+\infty} \frac{M-\Psi_t^{+}(F(\alpha z, t, Y^{t-1}))}{M-F(\alpha , t, Y^{t-1})}=\partial_x \Psi_t^{+}(M) \cdot \gamma(t,z).
\end{equation*}
So, we need to solve
\begin{equation*}
		\gamma\left(t,x\right)=\frac{\partial_x\Psi_{t}^{+}\left(M\right)}{N_t}\gamma\left(t,x\right)+\frac{N_t-1}{N_t x}\partial_x\Psi_{t}^{+}\left(M\right)\int_{0}^{x}\gamma\left(t,y\right)\,dy.
\end{equation*}
Using \eqref{eq::loc_phi_limit} in the scaled fixed-point equation, every subsequential limit
\(\gamma(t,\cdot)\) satisfies, for \(x>0\),
\[
\gamma(t,x)
=
\frac{\chi_t^+}{N_t}\gamma(t,x)
+
\frac{\chi_t^+(N_t-1)}{N_t x}
\int_0^x\gamma(t,y)\,\ud y,
\label{eq::loc_Volterra}
\]
with normalization \(\gamma(t,1)=1\). The first term comes from the ordinary kernel, whose
scaled version converges to \(\delta_x\); the second term comes from the averaged kernel,
whose scaled limiting action is \(x^{-1}\int_0^x(\cdot)\,\ud y\). We solve \eqref{eq::loc_Volterra}. Rearranging,
\begin{equation*}
\left(1-\frac{\chi_t^+}{N_t}\right)\gamma(t,x)
=
\frac{\chi_t^+(N_t-1)}{N_t x}
\int_0^x\gamma(t,y)\,\ud y.
\end{equation*}
We let $A(x):=\int_0^x\gamma(t,y)\,\ud y$. Then $A'(x)=\gamma(t,x)$, and $xA'(x)=\frac{\chi_t^+(N_t-1)}{N_t-\chi_t^+}A(x)$. Therefore $A(x)=C x^{\frac{\chi_t^+(N_t-1)}{N_t-\chi_t^+}}$, and hence $\gamma(t,x)=x^{\rho_t^+}$, where
\begin{equation*}
\rho_t^+
=
\frac{\chi_t^+-1}{1-\chi_t^+/N_t}
=
\frac{\partial_x\Psi_t^+(M)-1}
{1-\partial_x\Psi_t^+(M)/N_t}.
\end{equation*}
Since every subsequential limit is the same and \(\gamma(t,1)=1\), the full family converges.
Thus
\begin{equation*}
\frac{M-F(\alpha x,t,Y^{t-1})}
{M-F(\alpha,t,Y^{t-1})}
\to
x^{\rho_t^+},
\qquad x>0,
\end{equation*}
which proves that \(M-F(\cdot,t,Y^{t-1})\) is regularly varying at \(+\infty\) with index
\(\rho_t^+\). We thus have
\begin{equation*}
    \rho_t^{+} = \frac{\partial_x \Psi_t^{+}(M) - 1}{1 - \frac{\partial_x \Psi_t^{+}(M)}{N_t} },
\end{equation*}
from which it follows that $\lim_{t \rightarrow +\infty}|\rho_t^{+}|=0$ since $\partial_x \Psi_t^{+}(M) \rightarrow 1$ and the denominator does not vanish under the assumption that $N_t>1$.
\subsection{Proof of Corollary \ref{coroll::volume_varying}}\label{app::volume_varying}
We prove the result for $\Pi_t^{+}(F(x,t,Y^{t-1}))$; the result for $\Pi_t^{-}(F(x,t,Y^{t-1}))$ is analogous. To do this, we leverage Theorem \ref{thm:KaramataDirect}, and show that
\begin{equation*}
    \lim_{x\rightarrow\infty}\frac{\Pi^{+}_t(F(x,t,Y^{t-1}))}{\int_{x}^{\infty}\frac{\Pi^{+}_t(F(u,t,Y^{t-1}))}{u}\ud u} = -\frac{\partial_x\Psi_{t}^{+}(M)}{1-\partial_x\Psi_{t}^{+}(M)}\rho_{t}^{+}.
\end{equation*}
Henceforth, to further simplify the notation, we denote $F(\cdot,t,Y^{t-1})$ by $F_t(\cdot)$. The proof is divided into four steps.

\emph{Step 1} By standard integration by parts, we obtain $\Pi^{+}_t(x)=\frac{\int_{x}^{M}\Pi^{+}_t(y)\ud y}{\Psi^{+}_t(x)-x}$, which further gives us $-\frac{\partial_x\Pi_{t}^{+}(x)}{\Pi^{+}_t(x)}=\frac{\partial_x\Psi_{t}^{+}(x)}{\Psi^{+}_t(x)-x}$. Therefore, 
\begin{align*}
    \lim_{x\rightarrow\infty}\frac{\Pi^{+}_t(F_t(x))}{\int_{x}^{\infty}\frac{\Pi_t^{+}(F_t(u))}{u}\,du}&=\lim_{x\rightarrow\infty}\frac{x\partial_x\Pi_{t}^{+}(F_t(x))F_t^{\prime}(x)}{-\Pi^{+}_t(x)}=\lim_{x\rightarrow\infty}\frac{x\partial_x\Psi_{t}^{+}(F_t(x))F^{\prime}_t(x)}{\Psi^{+}_t(F_t(x))-F_t(x)}.
\end{align*}

\emph{Step 2} We show that $\Psi^{+}_t(F_t)-F_t$ is regularly varying of index $\rho_{t}^{+}$ at $+\infty$. We write
	\begin{align*}
		\lim_{\alpha\rightarrow\infty}\frac{\Psi^{+}_t(F_t(\alpha x))-F_t(\alpha x)}{\Psi^{+}_t(F_t(\alpha))-F_t(\alpha)}=\lim_{\alpha\rightarrow\infty}\frac{\frac{\Psi^{+}_t(F_t(\alpha x))-F_t(\alpha x)}{M-F_t(\alpha x)}}{\frac{\Psi^{+}_t(F_t(\alpha ))-F_t(\alpha )}{M-F_t(\alpha)}}\frac{M-F_t(\alpha x)}{M-F_t(\alpha)}=\lim_{\alpha\rightarrow\infty}\frac{M-F_t(\alpha x)}{M-F_t(\alpha)}=x^{\rho_{t}^{+}},
	\end{align*}
	where we have used the following fact
	\begin{align*}
		\lim_{x\rightarrow\infty}\frac{\Psi^{+}_t(F_t(x))-F_t(x)}{M-F_t(x)}=\lim_{x\rightarrow\infty}\frac{(\partial_x\Psi^{+}_{t}(F_t(x))-1)F^{\prime}_t(x)}{-F^{\prime}_t(x)}=1-\partial_x\Psi_{t}^{+}(M).
	\end{align*}

\emph{Step 3} We apply Theorem \ref{thm:KaramataDirect} with $\sigma=-1$ and we obtain
	\begin{align*}
		\lim_{x\rightarrow\infty}\frac{\Psi^{+}_t(F_t(x))-F_t(x)}{\int_{x}^{\infty}\frac{\Psi^{+}_t(F_t(u))-F_t(u)}{u}\,du}=-\rho_{t}^{+}.
	\end{align*}
Therefore, 
\begin{align*}
		\lim_{x\rightarrow\infty}\frac{x\partial_x\Psi_{t}^{+}(F_t(x))F_t^{\prime}(x)}{\Psi^{+}_t(F_t(x))-F_t(x)}=&\rho_{t}^{+}+\lim_{x\rightarrow\infty}\frac{xF^{\prime}_t(x)}{\Psi^{+}_t(F_t(x))-F_t(x)}\\
		=&\rho_{t}^{+}+\lim_{x\rightarrow\infty}\frac{xF^{\prime}_t(x)}{M-F_t(x)}\frac{M-F_t(x)}{\Psi^{+}_t(F_t(x))-F_t(x)}.
\end{align*}
At this point, notice that $\lim_{x\rightarrow\infty}\frac{M-F_t(x)}{\Psi^{+}_t(F_t(x))-F_t(x)}=\frac{1}{1-\partial_x\Psi^{+}_{t}(M)}$. By using, again Theorem \ref{thm:KaramataDirect}, we have $\lim_{x\rightarrow\infty}\frac{M-F_t(x)}{\int_{x}^{\infty}\frac{M-F_t(u)}{u}\,du}=\lim_{x\rightarrow\infty}\frac{xF^{\prime}_t(x)}{M-F_t(x)}=-\rho_{t}^{+}$.

\emph{Step 4} By combining the previous steps, we obtain
\begin{align*}
		\lim_{x\rightarrow\infty}\frac{\Pi^{+}_t(F_t(x))}{\int_{x}^{\infty}\frac{\Pi^{+}_t(F_t(u))}{u}\,du} = -\frac{\partial_x\Psi_{t}^{+}(M)}{1-\partial_x\Psi^{+}_{t}(M)}\rho_{t}^{+}.
\end{align*}
Therefore, by Theorem \ref{thm:KaramataConverse}, we have $\Pi^{+}_t(F_t)$ is regularly varying with index $ \frac{\partial_x\Psi_{t}^{+}(M)}{1-\partial_x\Psi^{+}_{t}(M)}\rho_{t}^{+}$ at $\infty$. This concludes the proof. 

\subsection{Proof of Proposition \ref{prop:tilted-posterior}}\label{app::tilted-posterior}
We fix a period $t$ and a period-$(t-1)$ history $Y^{t-1}$, and set $X_t^\ast:=F^{-1}(V,t,Y^{t-1})$, $\overline H_t(y):=\mathbb P_t(X_t^\ast\ge y)$, and $\overline S_\nu(y):=\mathbb P(Z_t\ge y)$. We know that $\overline H_t$ is regularly varying at $+\infty$ of index $-\alpha_t^{+}:=-\frac{\partial_x\Psi_{t}^{+}\left(M\right)}{1-\frac{\partial_x\Psi_{t}^{+}\left(M\right)}{N_t}}$ and $\overline S_{\nu}$ is regularly varying at $+\infty$ of index $-\nu$. We now choose $\gamma\in(\alpha_t^+/\nu,1)$ and define $b(y):=y^\gamma$. Then $b(y)\to\infty$, $\frac{b(y)}{y}\to 0$, and $S_\nu(b(y))=o(\overline H_t(y))$. In addition, we have $\overline H_t(y\pm b(y))\sim \overline H_t(y)$. We now show that 
\begin{equation}\label{eq:D5-denominator}
\mathbb P_t(X_t^\ast+Z_t\ge y)\sim \overline H_t(y).
\end{equation}
For the lower bound, $\{X_t^\ast\ge y+b(y),\, Z_t\ge -b(y)\}\subset \{X_t^\ast+Z_t\ge y\}$, hence $\mathbb P_t(X_t^\ast+Z_t\ge y) \ge
\overline H_t(y+b(y))\,\mathbb P(Z_t\ge -b(y))$. Since $\overline H_t(y+b(y))\sim \overline H_t(y)$ and $\mathbb P(Z_t\ge -b(y))\to 1$, we get
\begin{equation*}
\liminf_{y\to\infty}
\frac{\mathbb P_t(X_t^\ast+Z_t\ge y)}{\overline H_t(y)}
\ge 1
\end{equation*}
For the upper bound, $\{X_t^\ast+Z_t\ge y\}\subset \{X_t^\ast\ge y-b(y)\}\cup\{Z_t\ge b(y)\}$, so $\mathbb P_t(X_t^\ast+Z_t\ge y) \le \overline H_t(y-b(y))+\overline S_\nu(b(y)) \sim \overline H_t(y)$. This proves \eqref{eq:D5-denominator}. 

Now we define the tilted posterior $\pi_{y,t}(A):=\mathbb P_t(X_t^\ast\in A\mid X_t^\ast+Z_t\ge y)$. In order to prove the concentration, notice that 
\begin{equation*}
  \{X_t^\ast< y,\ X_t^\ast+Z_t\ge y\}
\subset
\{y-b(y)\le X_t^\ast< y\}
\cup
\{X_t^\ast< y-b(y),\, Z_t\ge b(y)\} .
\end{equation*}
Hence
\begin{equation*}
\mathbb P_t(X_t^\ast< y,\ X_t^\ast+Z_t\ge y)
\le
\overline H_t(y-b(y))-\overline H_t(y)+\overline S_\nu(b(y))
\end{equation*}
Since $\overline H_t(y-b(y))-\overline H_t(y)=o(\overline H_t(y))$ and $\overline S_\nu(b(y))=o(\overline H_t(y))$, we obtain $\mathbb P_t(X_t^\ast< y,\ X_t^\ast+Z_t\ge y)=o(\overline H_t(y))$. Finally, dividing by \eqref{eq:D5-denominator}, $\pi_{y,t}([y,\infty))=1-\mathbb P_t(X_t^\ast< y\mid X_t^\ast+Z_t\ge y) \rightarrow 1$. This proves the proposition.

\subsection{Proof of Lemma \ref{lem:ask-derivative}}\label{app::ask-derivative}
We write $A(y):=\mathbb E_t\!\big[V\,\mathbbm{1}_{\{X_t^\ast+Z_t\ge y\}}\big]$ and $B(y):=\mathbb P_t(X_t^\ast+Z_t\ge y)$, so that $h(y,t,Y^{t-1})=A(y)/B(y)$. Since $Z_t$ is independent of $V$ conditional on $Y^{t-1}$, we have $A(y)=\mathbb E_t\big[V\,S_\nu(y-X_t^\ast)\big]$ and $B(y)=\mathbb E_t \big[S_\nu(y-X_t^\ast)\big]$, where $S_\nu(u):=\mathbb P(Z_t\ge u)$. Differentiating under the expectation yields $A'(y)=-\mathbb E_t \big[V\,\mathfrak{q}_\nu(y-X_t^\ast;0,\sigma)\big]$ and $B'(y)=-\mathbb E_t\big[\mathfrak{q}_\nu(y-X_t^\ast;0,\sigma)\big]$. Using the quotient rule and the identity $\mathfrak{q}_\nu(u;0,\sigma)=S_\nu(u)\,r_\nu(u)$, one obtains the desired equation.

\subsection{Proof of Proposition \ref{prop:tail-gap}}\label{subsec:tail-gap-proof}
We fix a period $t$ and a period-$(t-1)$ history $Y^{t-1}$, and set $F_t(x):=F(x,t,Y^{t-1})$, $\Gamma_t(x):=M-F_t(x)$, $X_t^\ast:=F_t^{-1}(V)$, $\overline H_t(y):=\mathbb P_t(X_t^\ast\ge y)$, $\mathfrak f_t^\ast(y):=\mathfrak p_{t,X^\ast}(y\mid Y^{t-1})$, $S_\nu(u):=\mathbb P(Z_t\ge u)$, $r_\nu(u):=\frac{\mathfrak q_\nu(u;0,\sigma)}{S_\nu(u)}$, $D(y):=\mathbb P_t(X_t^\ast+Z_t\ge y)$, and set $\alpha:=\alpha_t^+$ and $\rho:=\rho_t^+$. We know that $\overline H_t$ is regularly varying at $+\infty$ of index $-\alpha$ and $\Gamma_t$ is regularly varying at $+\infty$ of index $-\rho$, where $\alpha<\nu$ and $\rho\in(-1,0)$. By assumption, \(\mathfrak{f}_t^\ast\) is eventually monotone, and so the monotone density theorem gives
\begin{equation}
\label{eq:tailgap-density}
\mathfrak{f}_t^\ast(y)\sim \alpha\,\frac{\overline H_t(y)}{y}.
\end{equation}
Because \(X_t^\ast\) has tail index \(\alpha<\nu\), it is heavier-tailed than the Student-\(t\)
noise. Hence
\begin{equation}
\label{eq:tailgap-denominator}
D(y)\sim \overline H_t(y).
\end{equation}
Moreover, a standard convolution equivalence then yields
\begin{equation}
\label{eq:tailgap-convolution-density}
\int_{-\infty}^{+\infty}\mathfrak q_\nu(y-x;0,\sigma)\mathfrak f_t^\ast(x)\,\ud x
\sim
\mathfrak f_t^\ast(y),
\end{equation}
and
\begin{equation}
\label{eq:tailgap-convolution-weighted}
\int_{-\infty}^{+\infty}\Gamma_t(x)\mathfrak q_\nu(y-x;0,\sigma)\mathfrak f_t^\ast(x)\,\ud x
\sim
\Gamma_t(y)\mathfrak f_t^\ast(y).
\end{equation}
The second equivalence uses \(\alpha-\rho<\nu\), which follows from the exponent inequalities in Lemma~\ref{lem:candidate-exponents}. Combining
\eqref{eq:tailgap-density}--\eqref{eq:tailgap-convolution-weighted}, we obtain
\begin{equation}
\label{eq:tailgap-basic-ratios}
\frac{\int \mathfrak q_\nu(y-x;0,\sigma)\mathfrak f_t^\ast(x)\,\ud x}{D(y)}
\sim
\frac{\alpha}{y},
\qquad
\frac{\int \Gamma_t(x)\mathfrak q_\nu(y-x;0,\sigma)\mathfrak f_t^\ast(x)\,\ud x}{D(y)}
\sim
\frac{\alpha\Gamma_t(y)}{y}.
\end{equation}
Moreover, $\mathbb E_t[V\mathbf 1_{\{X_t^\ast+Z_t\ge y\}}]=
M D(y)-
\int_{-\infty}^{+\infty}\Gamma_t(x)S_\nu(y-x)\mathfrak f_t^\ast(x)\,\ud x$. By heavy-tail dominance applied to the weighted tail and then by Karamata's theorem,
\begin{equation}
\label{eq:tailgap-weighted-tail}
\int_{-\infty}^{+\infty}\Gamma_t(x)S_\nu(y-x)\mathfrak f_t^\ast(x)\,\ud x
\sim
\int_y^\infty \Gamma_t(x)\mathfrak f_t^\ast(x)\,\ud x
\sim
\frac{\alpha}{\alpha-\rho}\Gamma_t(y)\overline H_t(y).
\end{equation}
Together with \eqref{eq:tailgap-denominator}, this implies
\begin{equation}
\label{eq:tailgap-mean-gap}
M-\mathbb E_{\pi_{y,t}}[V]
\sim
\frac{\alpha}{\alpha-\rho}\Gamma_t(y).
\end{equation}
We now use Lemma~\ref{lem:ask-derivative}. Under \(\pi_{y,t}\),
\begin{equation}
\label{eq:tailgap-hazard-mean}
\mathbb E_{\pi_{y,t}}\!\left[r_\nu(y-X_t^\ast)\right]
=
\frac{\int \mathfrak q_\nu(y-x;0,\sigma)\mathfrak f_t^\ast(x)\,\ud x}{D(y)}
\sim
\frac{\alpha}{y},
\end{equation}
and, since \(V=M-\Gamma_t(X_t^\ast)\),
\begin{equation}
\label{eq:tailgap-weighted-hazard}
\mathbb E_{\pi_{y,t}}\!\left[(M-V)r_\nu(y-X_t^\ast)\right]
\sim
\frac{\alpha\Gamma_t(y)}{y}.
\end{equation}
Therefore
\begin{equation}
    \begin{split}
\label{eq:tailgap-covariance}
\operatorname{Cov}_{\pi_{y,t}}\!\big(V,r_\nu(y-X_t^\ast)\big)
&=
\left(M-\mathbb E_{\pi_{y,t}}[V]\right)
\mathbb E_{\pi_{y,t}}\!\left[r_\nu(y-X_t^\ast)\right]
-
\mathbb E_{\pi_{y,t}}\!\left[(M-V)r_\nu(y-X_t^\ast)\right] \nonumber\\
&\sim
\left(\frac{\alpha}{\alpha-\rho}\Gamma_t(y)\right)\frac{\alpha}{y}
-
\frac{\alpha\Gamma_t(y)}{y} \nonumber\\
&=
\frac{\alpha\rho}{\alpha-\rho}\frac{\Gamma_t(y)}{y}.
    \end{split}
\end{equation}
Since \(\rho<0\),
\begin{equation}
\label{eq:tailgap-negative-covariance}
-\operatorname{Cov}_{\pi_{y,t}}\!\big(V,r_\nu(y-X_t^\ast)\big)
\sim
\frac{\alpha(-\rho)}{\alpha-\rho}
\frac{\Gamma_t(y)}{y}.
\end{equation}
Restoring notation,
\begin{equation}
\label{eq:tailgap-final-asymptotic}
-\operatorname{Cov}_{\pi_{y,t}}\!\big(V,r_\nu(y-X_t^\ast)\big)
\sim
\frac{\alpha_t^+(-\rho_t^+)}
{\alpha_t^+-\rho_t^+}
\frac{M-F(y,t,Y^{t-1})}{y}.
\end{equation}
By Lemma~\ref{lem:ask-derivative}, $\partial_y h(y,t,Y^{t-1})
=
-\operatorname{Cov}_{\pi_{y,t}}\!\big(V,r_\nu(y-X_t^\ast)\big)$.
Thus \(\partial_y h(y,t,Y^{t-1})>0\) for all sufficiently large \(y\). The bid-side proof is
identical, replacing the upper-tail quantities by their lower-tail analogues.

\section{Equilibrium Dynamics of \texorpdfstring{$F^{*}_t$}{F\_t*} and \texorpdfstring{$h^{*}_t$}{h\_t*}: Additional figures}\label{app::numerics_dynamics}
\begin{figure}[htbp]
    \centering
    \includegraphics[width=\linewidth]{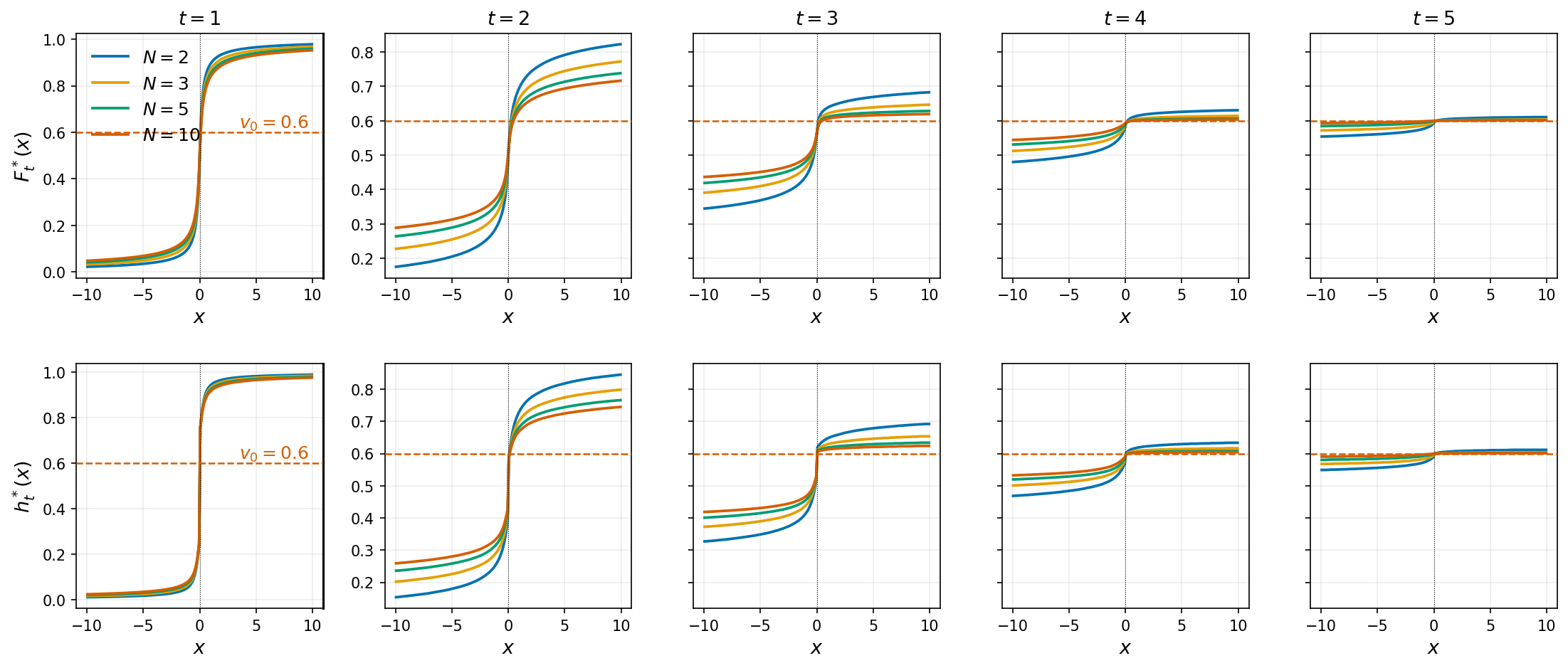}
    \caption{Equilibrium marginal cost $F_t^\star$ (\emph{Top panel}) and limit prices $h_t^\star$ (\emph{Bottom panel}) across the trading periods $t\in \{1,\ldots,5\}$, for the $\mathtt{Unif}
    ([0,1])$  asset distribution with $v_0=0.6$ and $N\in\{2,3,5,10\}$.}
\end{figure}
\begin{figure}[htbp]
    \centering
    \includegraphics[width=\linewidth]{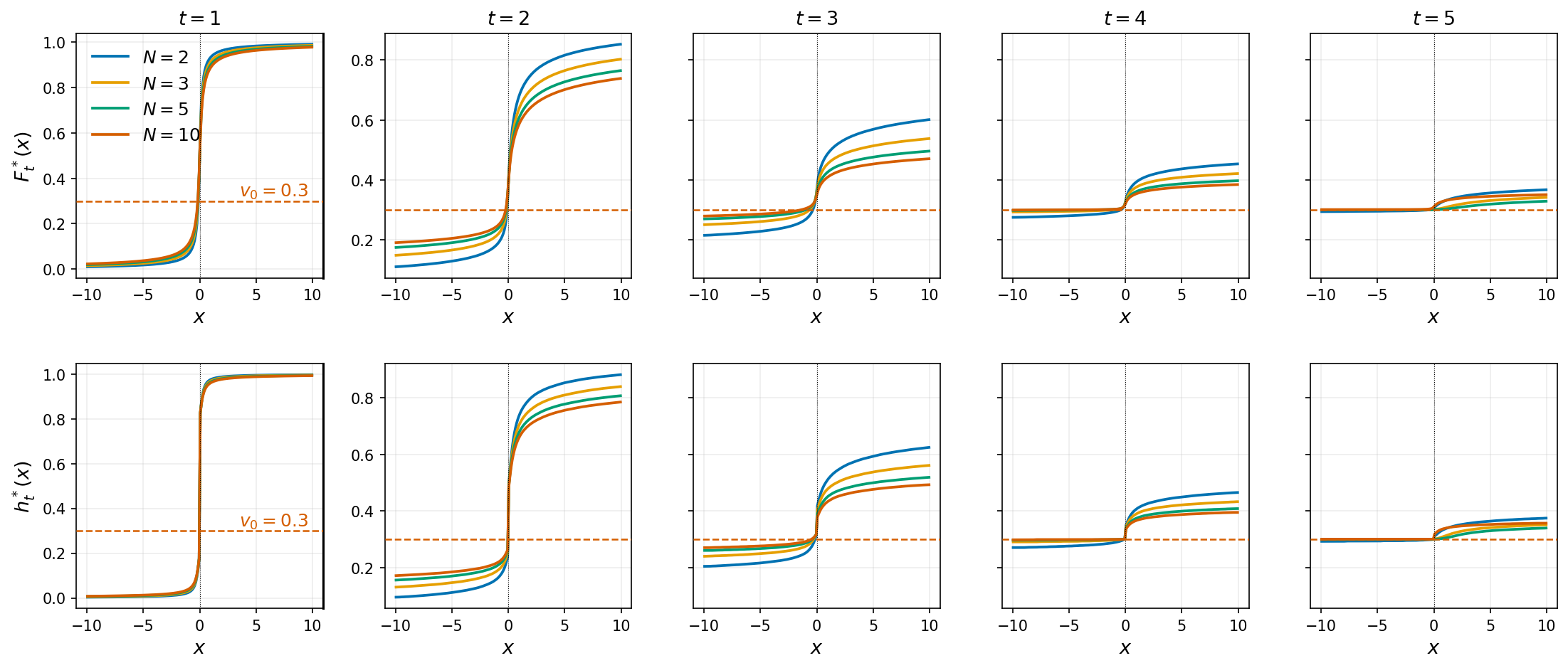}
    \caption{Equilibrium marginal cost $F_t^\star$ (\emph{Top panel}) and limit prices $h_t^\star$ (\emph{Bottom panel}) across the trading periods $t\in \{1,\ldots,5\}$, for the $\mathtt{Beta}
    (2, 2)[0,1]$  asset distribution with $v_0=0.3$ and $N\in\{2,3,5,10\}$.}
\end{figure}
\begin{figure}[htbp]
    \centering
    \includegraphics[width=\linewidth]{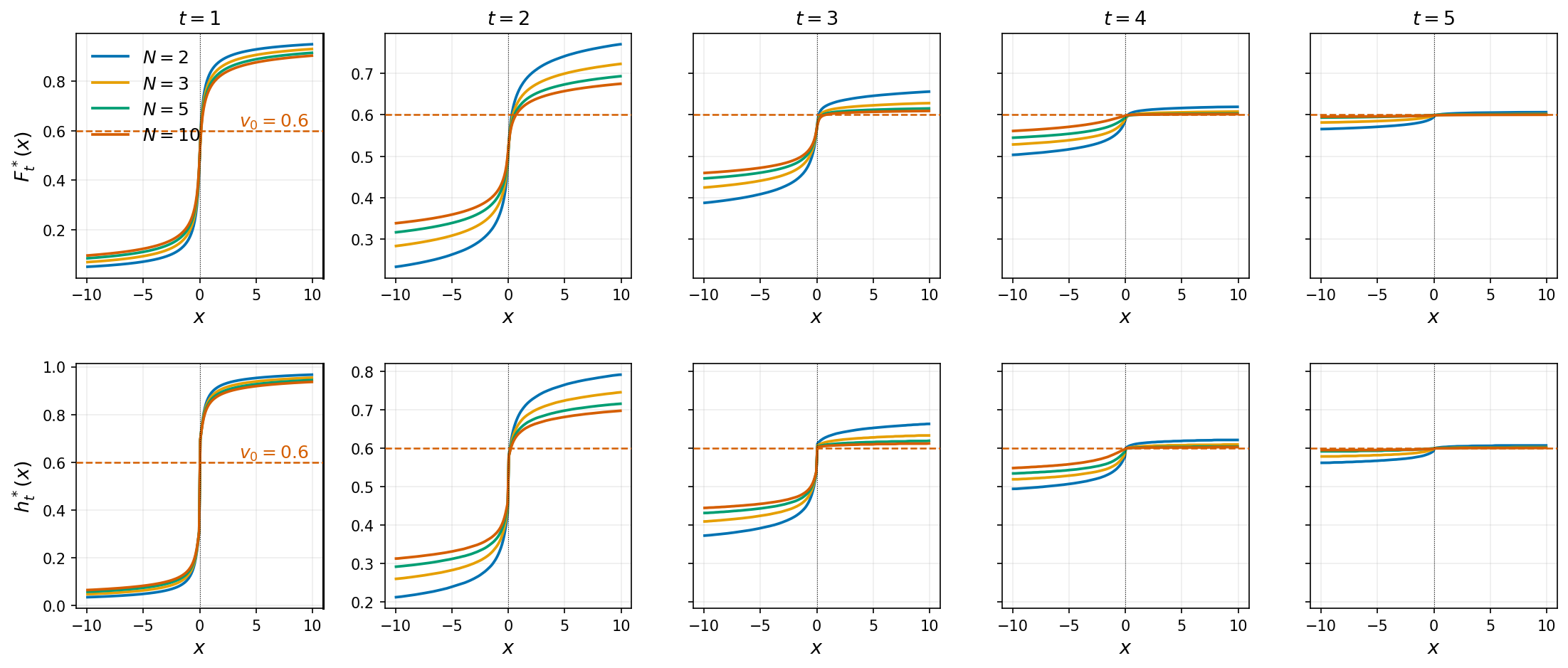}
    \caption{Equilibrium marginal cost $F_t^\star$ (\emph{Top panel}) and limit prices $h_t^\star$ (\emph{Bottom panel}) across the trading periods $t\in \{1,\ldots,5\}$, for the $\mathtt{Beta}
    (1/2, 1/2)[0,1]$  asset distribution with $v_0=0.6$ and $N\in\{2,3,5,10\}$.}
\end{figure}
\begin{figure}[htbp]
    \centering
    \includegraphics[width=\linewidth]{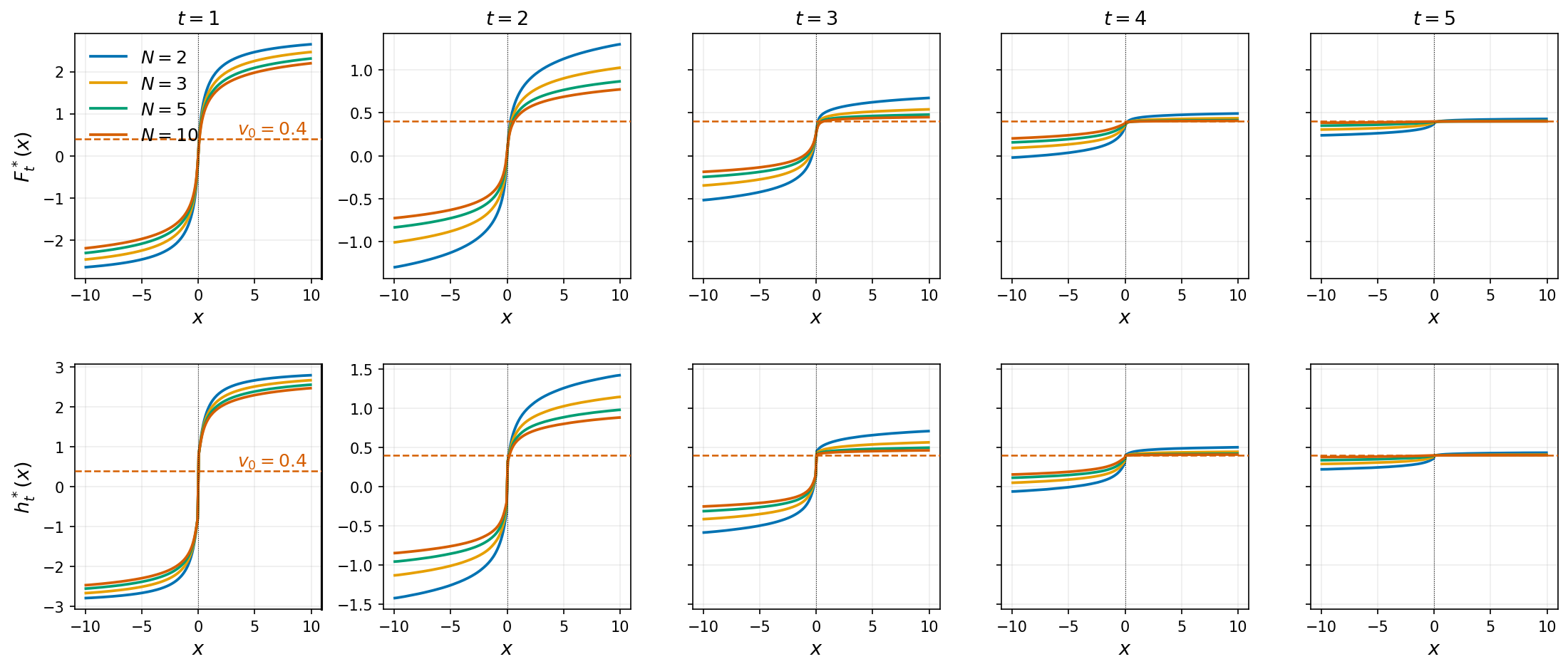}
    \caption{Equilibrium marginal cost $F_t^\star$ (\emph{Top panel}) and limit prices $h_t^\star$ (\emph{Bottom panel}) across the trading periods $t\in \{1,\ldots,5\}$, for the $\mathtt{Trunc-Gaussian}(0,1,[-3,3])$  asset distribution with $v_0=0.4$ and $N\in\{2,3,5,10\}$.}
\end{figure}
\begin{figure}[htbp]
    \centering
    \includegraphics[width=\linewidth]{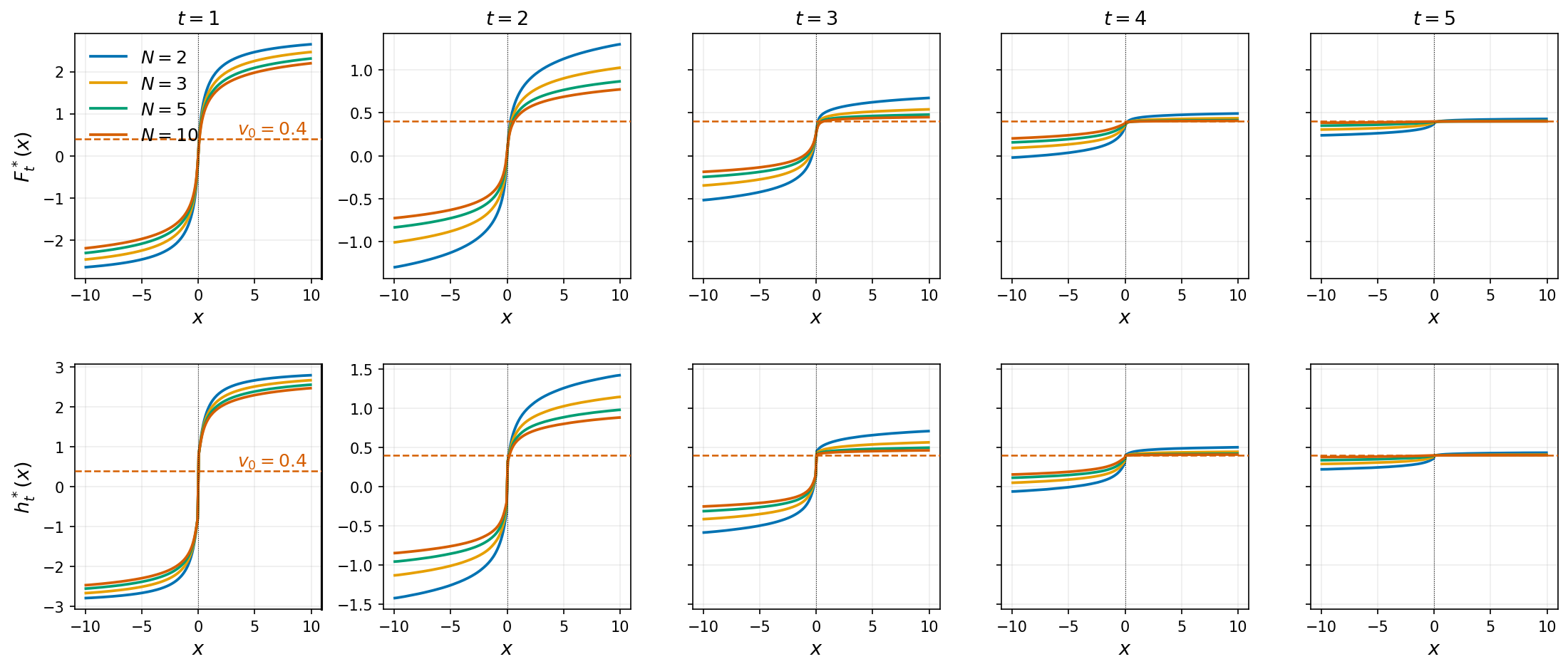}
    \caption{Equilibrium marginal cost $F_t^\star$ (\emph{Top panel}) and limit prices $h_t^\star$ (\emph{Bottom panel}) across the trading periods $t\in \{1,\ldots,5\}$, for the $\mathtt{Student}-t(5,1)$  asset distribution with $v_0=0.4$ and $N\in\{2,3,5,10\}$.}
\end{figure}

\newpage
\section{Asymptotics of price impact: additional figures}\label{app::price_figures}
\begin{figure}[htbp]
    \centering
    \includegraphics[width=1\linewidth]{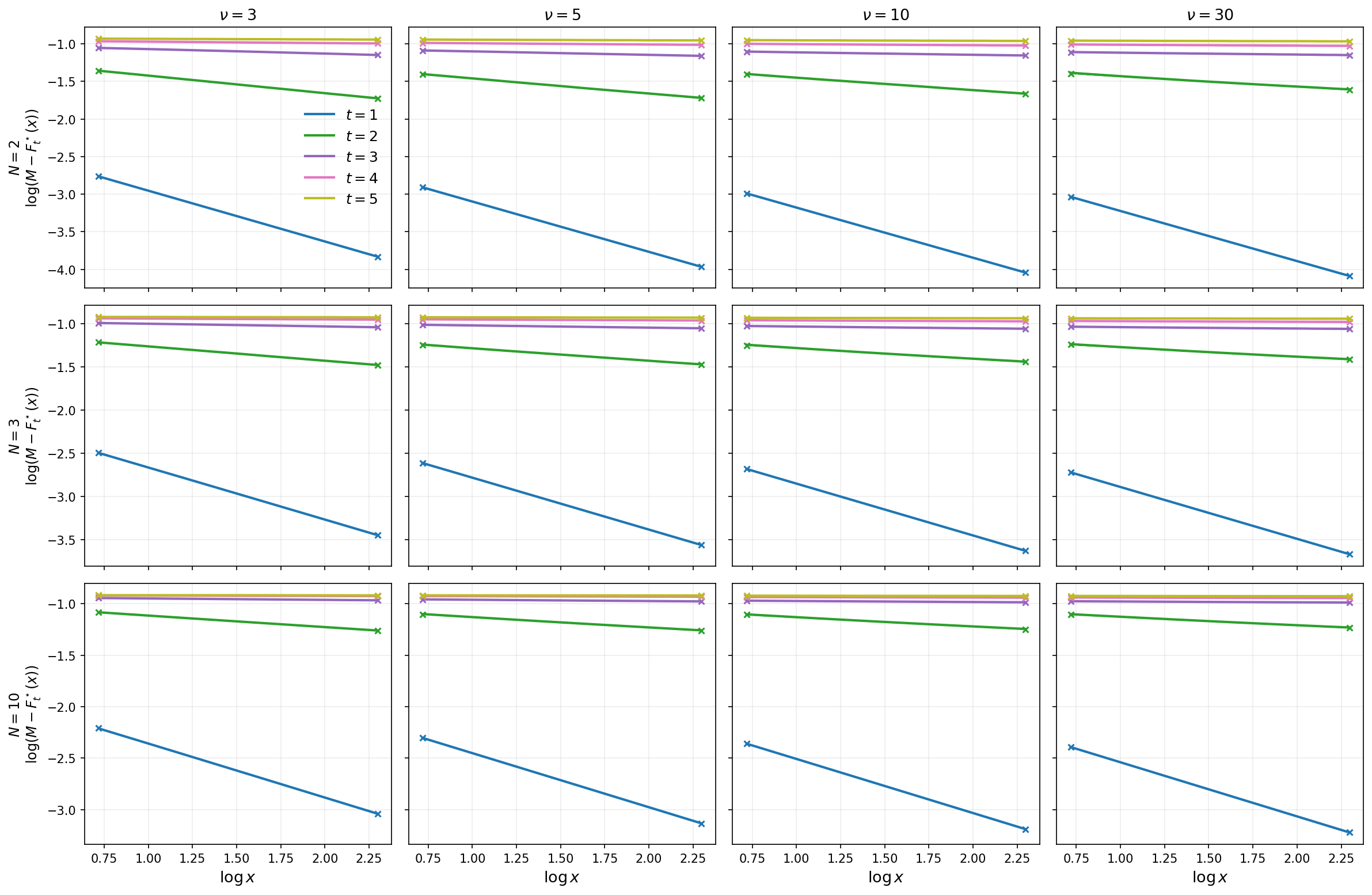}
    \caption{Log-log plot of $\log(M- F_t^\star(x))$ against $\log x$ for the $\mathtt{Unif}([0, 1])$ asset distribution, across trading periods $t=1,\ldots,5$, for $\nu\in\{3,5,10,30\}$ and $N\in\{2,3,10\}$. Solid lines: $\log F_t^\star$ from fixed point solutions; dashed lines with markers: fitted 
    slope $\hat{\gamma}_t$. Fit region $x > 1.5$.}
\end{figure}
\begin{figure}[htbp]
    \centering
    \includegraphics[width=1\linewidth]{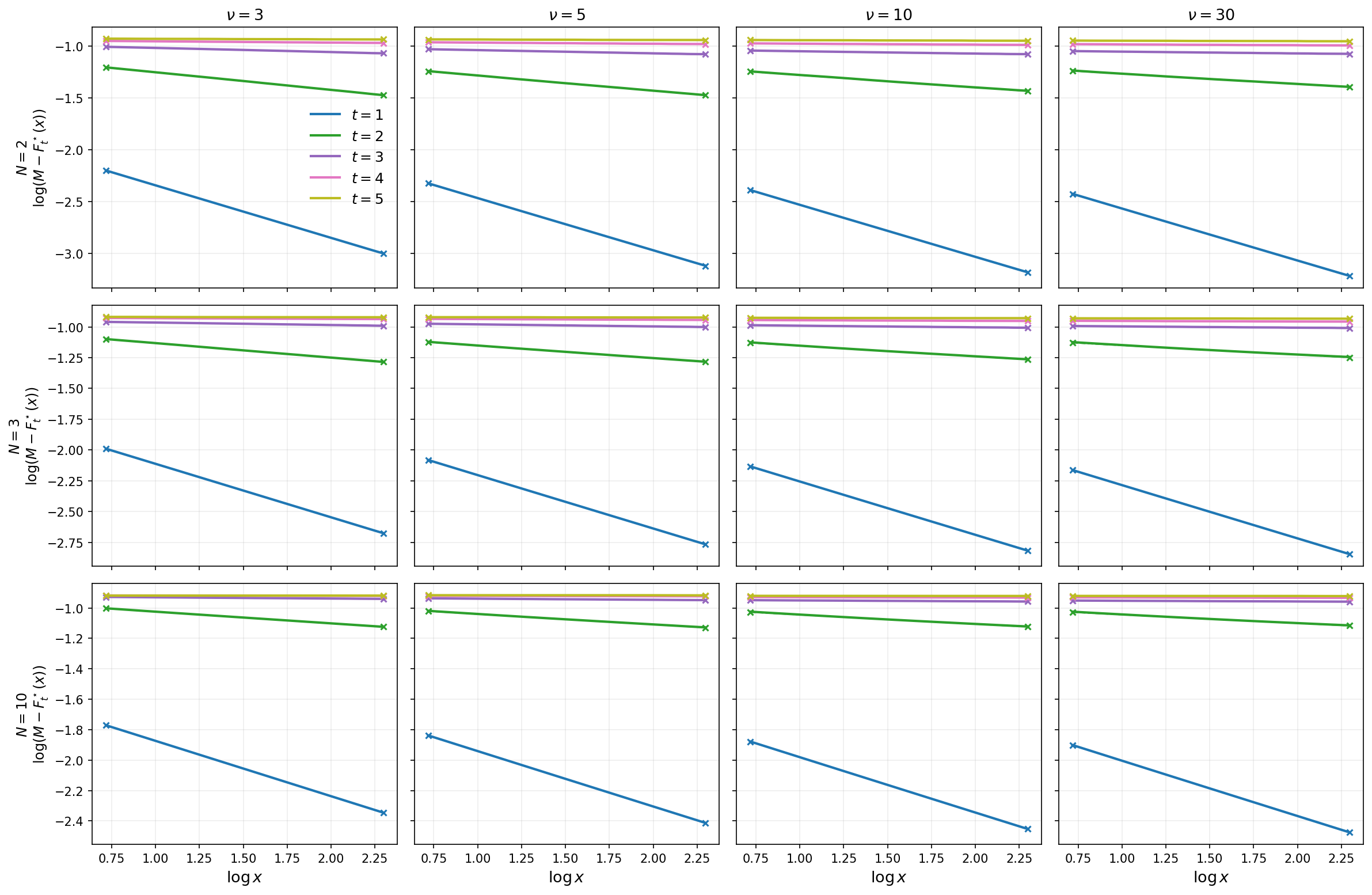}
    \caption{Log-log plot of $\log(M-F_t^\star(x))$ against $\log x$ for the $\mathtt{Beta}(2,2)$ asset distribution, across trading periods $t=1,\ldots,5$, for $\nu\in\{3,5,10,30\}$ and $N\in\{2,3,10\}$. Solid lines: $\log F_t^\star$ from fixed point solutions; dashed lines with markers: fitted 
    slope $\hat{\gamma}_t$. Fit region $x > 1.5$.}
\end{figure}
\begin{figure}[htbp]
    \centering
    \includegraphics[width=1\linewidth]{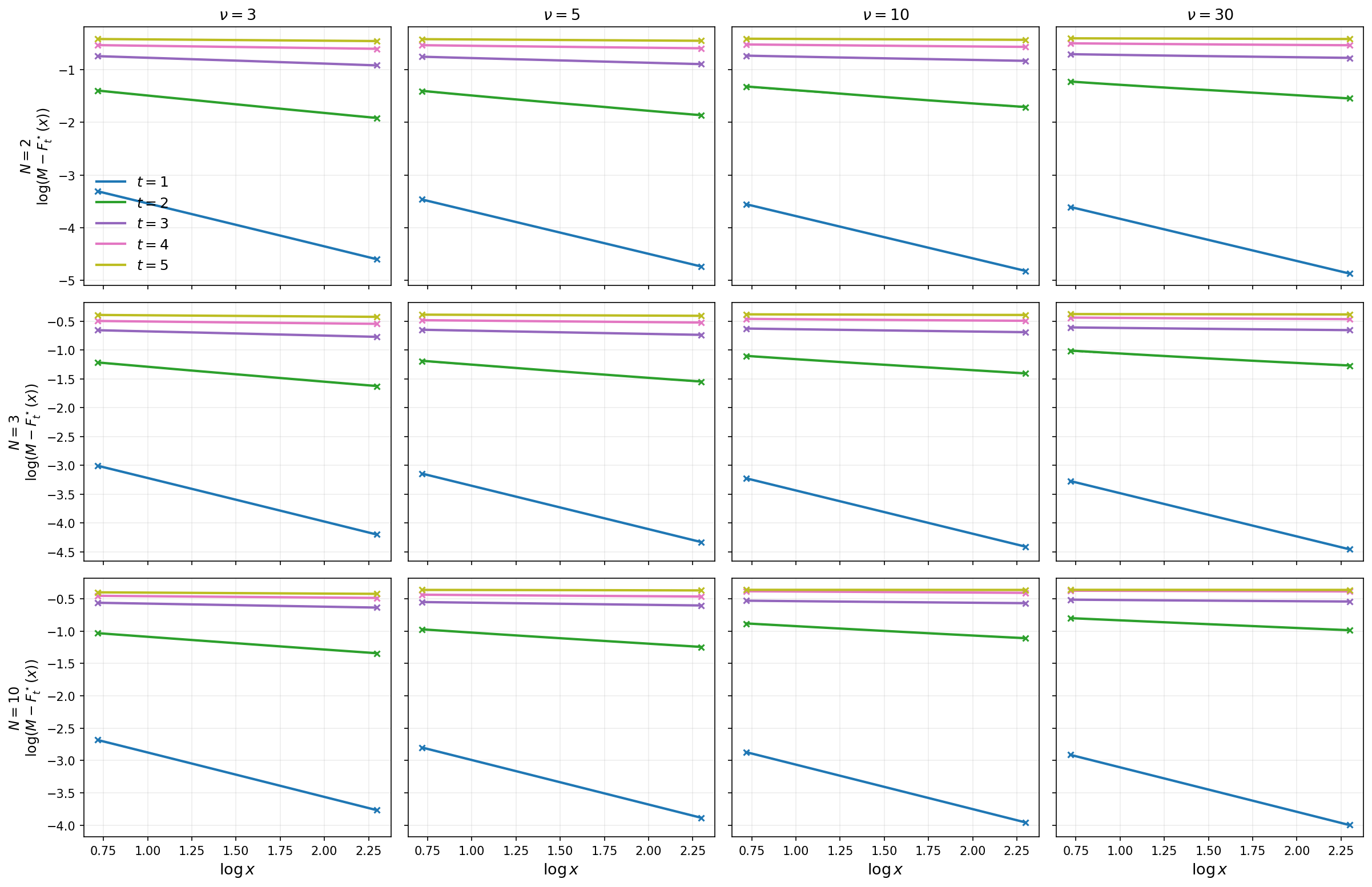}
    \caption{Log-log plot of $\log F_t^\star(x)$ against $\log x$ for the $\mathtt{Beta}(1/2,1/2)$ asset distribution, across trading periods $t=1,\ldots,5$, for $\nu\in\{3,5,10,30\}$ and $N\in\{2,3,10\}$. Solid lines: $\log F_t^\star$ from fixed point solutions; dashed lines with markers: fitted 
    slope $\hat{\gamma}_t$. Fit region $x > 1.5$.}
\end{figure}
\begin{figure}[htbp]
    \centering
    \includegraphics[width=1\linewidth]{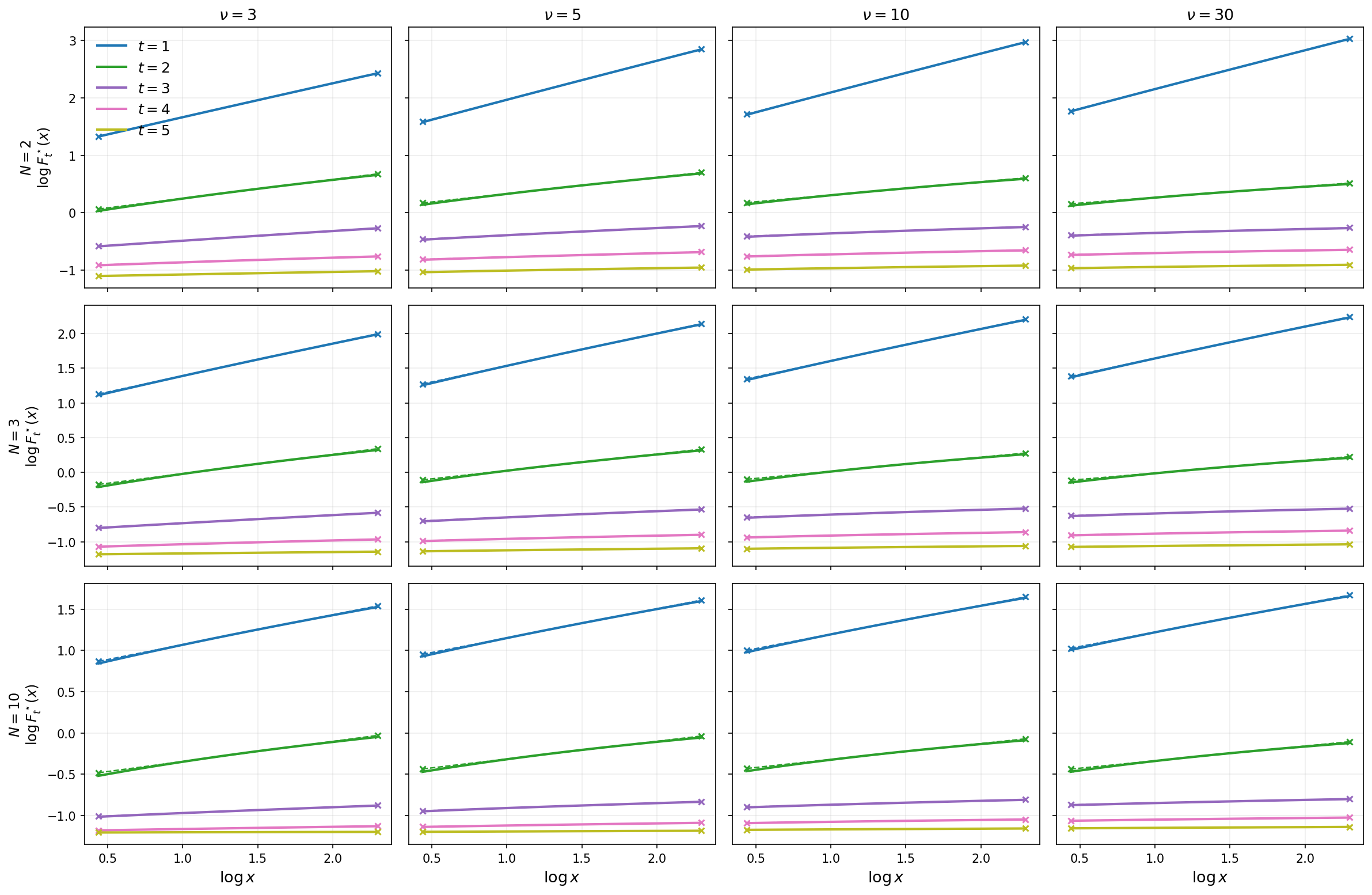}
    \caption{Log-log plot of $\log F_t^\star(x)$ against $\log x$ for the $\mathtt{Student}-t(5,1)$ asset distribution, across trading periods $t=1,\ldots,5$, for $\nu\in\{3,5,10,30\}$ and $N\in\{2,3,10\}$. Solid lines: $\log F_t^\star$ from fixed point solutions; dashed lines with markers: fitted 
    slope $\hat{\gamma}_t$. Fit region $x > 1.5$.}
\end{figure}

\begin{figure}
    \centering
    \includegraphics[width=\linewidth]{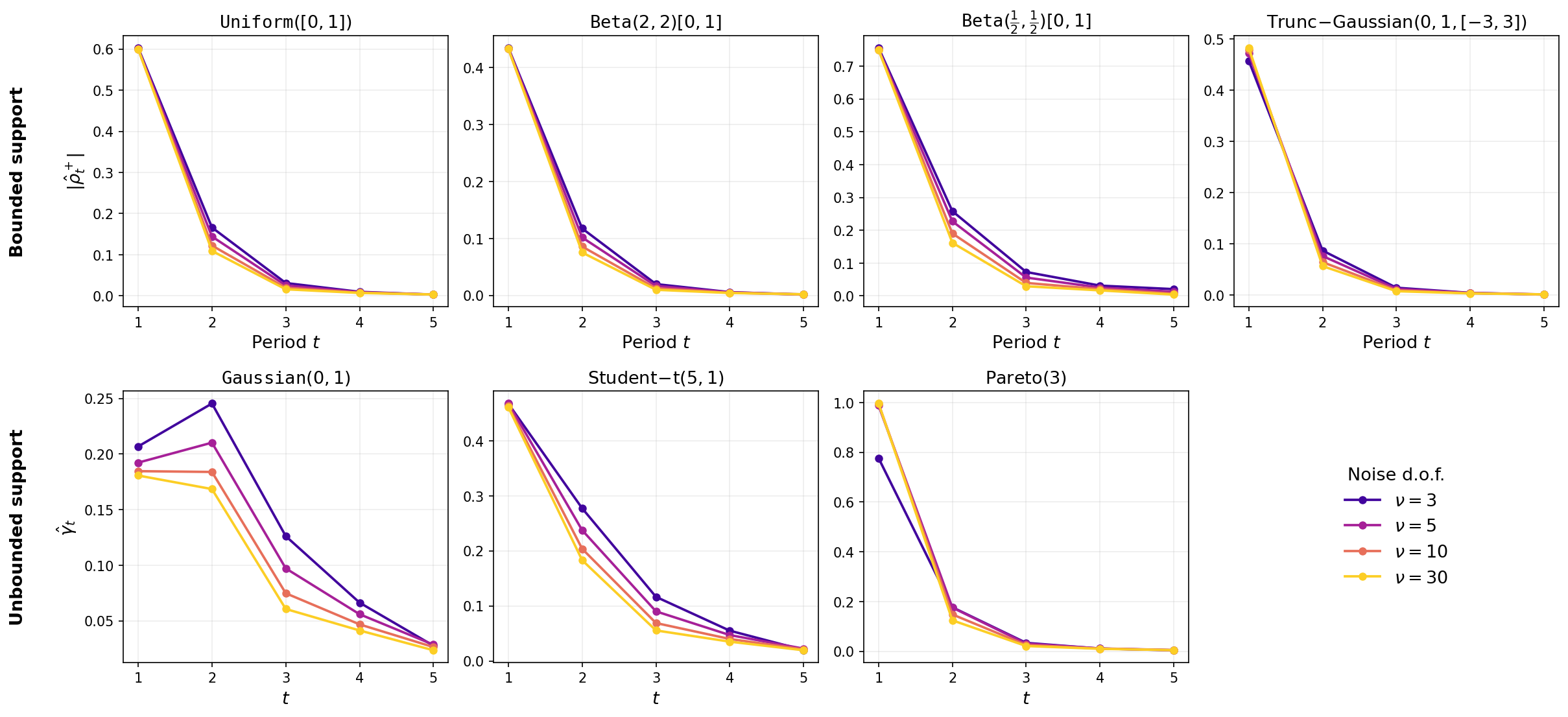}
    \caption{Fitted power law exponents $|\hat{\rho}_t^+|$ (bounded priors, from 
    $M - F_t^\star$) and $\hat{\gamma}_t$ (unbounded priors, from $F_t^\star$) 
    across trading periods $t=1,\ldots,5$, for noise trades $\nu\in\{3,5,10,30\}$ and $N=3$.}
\end{figure}

\begin{figure}
    \centering
    \includegraphics[width=\linewidth]{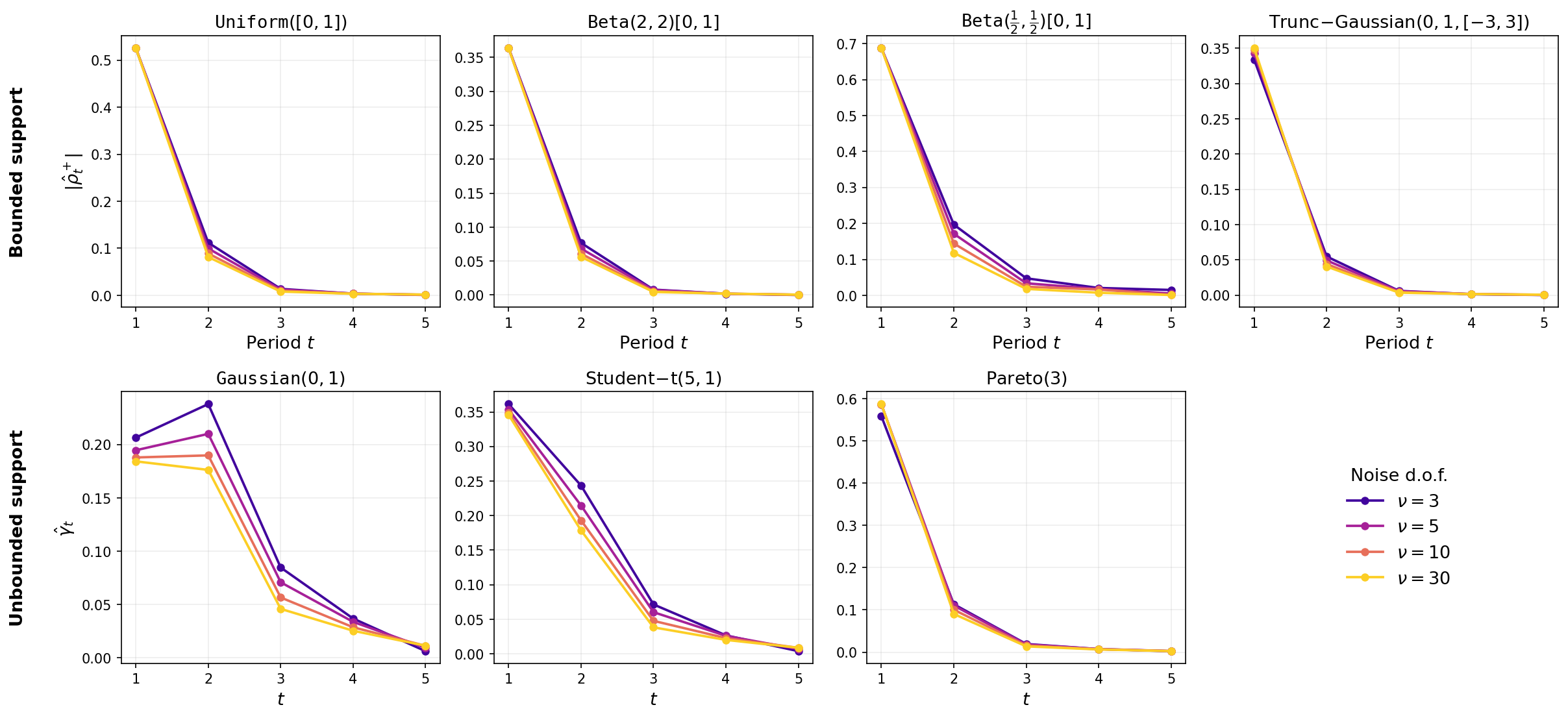}
    \caption{Fitted power law exponents $|\hat{\rho}_t^+|$ (bounded priors, from 
    $M - F_t^\star$) and $\hat{\gamma}_t$ (unbounded priors, from $F_t^\star$) 
    across trading periods $t=1,\ldots,5$, for noise trades $\nu\in\{3,5,10,30\}$ and $N=10$.}
\end{figure}

\newpage
\section{Additional empirical results for INTL and MSFT}\label{app::additional_empirical_results}
We repeat the refined 10-level specification for MSFT over the same 22 trading days, comprising $87,120$ five-second windows. Tail-state separation remains sharp: mean Hill estimates are $1.53$,
$1.83$, and $2.18$ in {\rm HeavyTail}, {\rm Middle}, and {\rm LightTail} blocks, and the corresponding $99^{th}$ percentiles
of normalized absolute signed flow are $20.38$, $17.97$, and $15.62$. Price-discovery evidence is selective but consistent with a far-tail mechanism. The ${\rm Extreme}\times {\rm HeavyTail}$ signed-flow
interaction is $0.390$ at $60$ seconds ($t=2.64$) and $0.868$ at $300$ seconds ($t=1.90$), while the ${\rm ModerateLarge}\times{\rm HeavyTail}$ interaction is $0.802$ at $300$ seconds ($t=2.25$). Spread-resilience effects are weaker than for AAPL: only the $30$-second ${\rm Extreme}\times{\rm HeavyTail}$ interaction is significant
($0.064$ bp, $t=2.04$), and the $300$-second ${\rm ModerateLarge}\times{\rm HeavyTail}$ interaction is marginal ($0.043$
bp, $t=1.88$). The {\rm HeavyTail} state main effect is negative at all horizons, so the positive interactions should be interpreted as incremental size effects rather than evidence that spreads are unconditionally wider in heavy-tail blocks. Overall, MSFT supports tail-regime separation and selective far-tail price discovery, but not a broad spread-persistence effect.

For INTC, the same $22$-day, $10$-level specification yields mean Hill estimates of $1.35$, $1.60$, and $1.92$ in HeavyTail, Middle, and LightTail blocks. Thus, signed flow is heavy-tailed throughout the sample, although the HeavyTail state remains substantially more extreme: its $99^{th}$ percentile of normalized absolute signed flow is $27.73$, compared with $26.31$ and $20.56$ in the Middle and LightTail regimes. The price-impact interactions are generally imprecise. The only suggestive estimate is the $300$-second $\text{Extreme}\times\text{HeavyTail}$ coefficient of $1.717$ ($t=1.83$), and the largest-bin HeavyTail signed-impact estimates are not statistically different from zero. The strongest INTC result is instead a tail-state effect in spreads: conditional future spread changes are $0.044--0.048$ bp larger in HeavyTail blocks from $10$ to $300$ seconds, with $t$-statistics from 3.56 to 5.18. The $\text{ModerateLarge}\times\text{HeavyTail}$ and $\text{Extreme}\times\text{HeavyTail}$ spread interactions are not robust. INTC therefore supports the relevance of liquidity-tail states for conditional liquidity, but provides limited evidence on size-specific crossover or far-tail price discovery.

\end{document}